\documentclass[aoas]{imsart}
\def\journal@name{}

\RequirePackage{amsthm,amsmath,amsfonts,amssymb,bm}
\RequirePackage[authoryear]{natbib}
\RequirePackage[colorlinks,citecolor=blue,urlcolor=blue]{hyperref}
\RequirePackage{graphicx}
\RequirePackage{tabularx,booktabs,multirow,longtable}
\RequirePackage{rotating}
\RequirePackage{physics}
\RequirePackage{placeins}
\RequirePackage[nameinlink]{cleveref}
\RequirePackage{algorithm, algpseudocode}

\RequirePackage{bibunits}
\defaultbibliographystyle{imsart-nameyear}
\graphicspath{{./}{appendix/}}
\pdfoutput=1

\RequirePackage{lmodern}
\startlocaldefs
\numberwithin{equation}{section}

\newtheorem{prop}{Proposition}

\theoremstyle{definition}
\newtheorem{remark}{Remark}
\newcommand{\be}{\bm{e}}

\newcommand{\bx}{\bm{x}}

\newcommand{\bB}{\mathbf{B}}
\newcommand{\bD}{\mathbf{D}}
\newcommand{\bE}{\mathbf{E}}

\newcommand{\bI}{\mathbf{I}}
\newcommand{\bK}{\mathbf{K}}
\newcommand{\bM}{\mathbf{M}}

\newcommand{\bP}{\mathbf{P}}
\newcommand{\bQ}{\mathbf{Q}}
\newcommand{\bS}{\mathbf{S}}
\newcommand{\bT}{\mathbf{T}}
\newcommand{\bU}{\mathbf{U}}
\newcommand{\bW}{\mathbf{W}}

\newcommand{\bOmega}{\bm{\Omega}}
\newcommand{\bomega}{\bm{\omega}}

\newcommand{\bdelta}{\bm{\delta}}
\newcommand{\bGamma}{\bm{\Gamma}}
\newcommand{\bbeta}{\bm{\beta}}

\newcommand{\EE}{\mathbb{E}}
\newcommand{\RR}{\mathbb{R}}

\newcommand{\PP}{\mathbb{P}}

\newcommand{\mcA}{\mathcal{A}}
\newcommand{\mcB}{\mathcal{B}}

\newcommand{\mcS}{\mathcal{S}}

\newcommand{\mcX}{\mathcal{X}}

\newcommand{\red}[1]{\textcolor{red}{#1}}
\newcommand{\blu}[1]{\textcolor{blue}{#1}}

\endlocaldefs

\begin{document}

\begin{bibunit}

\begin{frontmatter}
\title{Bayesian modeling and prediction of \\ generalized contact matrices}
\runtitle{Inferring and Predicting Generalized Social Contact Matrices}

\begin{aug}
\author[A]{\fnms{Shozen}~\snm{Dan}\ead[label=e1]{shozen.dan21@imperial.ac.uk}},
\author[A]{\fnms{David A.}~\snm{van Dyk}},
\author[B]{\fnms{Zhi}~\snm{Ling}},
\author[B,C]{\fnms{Swapnil}~\snm{Mishra}}
\and
\author[A]{\fnms{Oliver}~\snm{Ratmann}\ead[label=e3]{oliver.ratmann@imperial.ac.uk}}
\address[A]{Imperial College London, Department of Mathematics \\ \printead{e1,e3}}
\address[B]{National University of Singapore, Saw Swee Hock School of Public Health}
\address[C]{National University of Singapore, Institute of Data Science}
\end{aug}

\begin{abstract}
Social contact matrices are essential tools in infectious disease epidemiology as they quantify close-range human contact patterns which directly drive the transmission of airborne infectious diseases. In this work we propose a Bayesian modeling framework for inferring generalized contact matrices which stratify contact matrices beyond contemporary age dimensions. The model is designed to satisfy fundamental structural assumptions of contacts while leveraging tensor structures and smoothing constraints to make high-dimensional matrix estimation computationally feasible and statistically stable. We discover a link between multi-dimensional matrix stratification subject to structural constraints with the theory of contingency tables. This enables us to approach a challenging missing-data problem commonly encountered in real-world analysis where feature information on the contacts is unobserved. We benchmark the framework against existing methods through simulation studies and illustrate the framework’s practical utility through two real-world datasets: BICS (United States) and COVIMOD (Germany). Our models are implemented in an open-source Python package to facilitate adoption in the wider scientific community.
\end{abstract}


\end{frontmatter}

\section{Introduction}

This article develops the \emph{generalized Bayesian social mixing} (\texttt{g-mix}) model, a principled framework for inferring human social contact patterns by features beyond the age of contacting and contacted individuals (e.g., gender, race and ethnicity, poverty, or socioeconomic status) in the form of \emph{generalized contact matrices}. 
Understanding such patterns of social interaction is particularly important in infectious disease epidemiology, because many pathogens such as coronaviruses and influenza transmit in close-range proximity \citep{anderson_infectious_1991}. 
Indeed, the COVID-19 pandemic highlighted how strongly transmission dynamics can be modulated by demographic and socioeconomic factors beyond age \citep{buckee_thinking_2021, manna_generalized_2024,liu_impact_2021, li_temporal_2021, malani_seroprevalence_2021}, adding to evidence from tuberculosis \citep{pelissari_tuberculosis_2017}, Ebola \citep{brainard_ebola_2016}, Mpox \citep{endo_mpox_2022}, and other epidemics \citep{mamelund_influenza_2021}.
This is spurring efforts to understand how contact patterns vary across these dimensions and to incorporate this understanding into transmission models \citep{manna_generalized_2024, di_domenico_individual_2025}. 

Social contact patterns can be captured through various data streams, including mobile phone telemetry \citep{rudiger_predicting_2021}, digital contact tracing logs \citep{cattuto_dynamics_2010, stopczynski_measuring_2014} and social contact surveys \citep{mossong_social_2008, hoang_systematic_2019, liu_rapid_2021, mousa_social_2021}.
Among these, social contact surveys serve as the cornerstone approach as they capture disease-relevant interactions, provide rich contextual detail about both the respondents and their contacts, and face fewer privacy barriers compared to digital tracing methods. 
Data from social contact surveys are typically summarized as \emph{social contact matrices}. Each cell of the matrix encodes the \emph{contact intensity}: the average number of contacts a person in a specific demographic group has with members of another group. 
These matrices serve as inputs to downstream tasks, improving the estimation of reproduction numbers \citep{gimma_changes_2022, monod_age_2021}, the monitoring of intervention impacts on disease-relevant behavior \citep{coletti_comix_2020, tomori_individual_2021, willem_socrates_2020}, or the design of optimal vaccine allocation strategies \citep{medlock_optimizing_2009}. 

Our work is inspired by model-based spatial statistics~\citep{banerjee_hierarchical_2014}, as contact intensities by age of respondent and contact can be thought of as smooth, potentially non-isotropic surfaces over a compact domain akin to longitude and latitude, and additional features result in variations of the age-age surface. 
Crucially, in a closed population, the true unknown number of contacts between two groups A and B is the same, regardless if measured from people in A or B.
This \emph{reciprocity} property provides symmetries on latent model parameters that can be exploited for statistically efficient inference, especially when we wish to infer generalized contact matrices by features from sparse data. 
We adopt a Bayesian framework that naturally accounts for finite sample sizes, enables principled modeling of marginal and interaction effects across multiple features beyond age and gender, and unlocks machinery for feature selection, uncertainty quantification, and prediction.
Missing data is a particularly relevant challenge, since respondents know their household size, income, or educational attainment but rarely those of their contacts, and thus the contact intensities between groups defined by these features are only partially observable.
We show it is possible to derive bounds for these unobserved contact intensities in the same way as bounds for inner cells in contingency tables can be derived from fixed margins \citep{dobra_bounds_2000}, and further it is possible to incorporate external information within the Bayesian predictive distribution when available.
Scientifically, this moves us closer to understanding how contact patterns vary across dimensions of social stratification beyond age, and thus to better understanding the drivers of infectious disease disparities and improving the design of interventions to mitigate them.

Earlier work has focused on inferring contact intensities between age groups \citep{mossong_social_2008, wallinga_using_2006, goeyvaerts_estimating_2010, prem_projecting_2017, kassteele_efficient_2017, dan_estimating_2023}. 
The popular \texttt{socialmixr} package \citep{funk_socialmixr_2024} infers contact intensities using sample means and quantifies uncertainty using the bootstrap \citep{efron_bootstrap_1979}. In this approach, reciprocity is enforced with a post-hoc correction, where the final contact matrix estimate is a weighted sum of the original contact intensity matrix and its transpose.
\citet{wallinga_using_2006} pioneered the model-based approach that exploits reciprocity.
A key idea of their work is to model contact counts as over-dispersed Negative Binomial random variables, and to enforce reciprocity on the latent population-level contacts by maximizing a composite likelihood that pools information from both perspectives of the contact matrix. 
This allows the method to derive a consensus estimate for the contact intensity between any two groups. 
The neighboring cells of a contact matrix often exhibit strong autocorrelation. Recognizing this, \citet{goeyvaerts_estimating_2010, mossong_social_2008, hens_mining_2009, ogunjimi_using_2009} employed tensor-product splines to model contact intensity matrices with one-year age bins as continuous surfaces. Within this framework, reciprocity is enforced using a smooth-then-constrain approach \citep{mammen_general_2001}, where the surface is first estimated and subsequently adjusted to satisfy reciprocity.
\citet{prem_projecting_2017} developed a Bayesian approach for contact matrix inference, allowing for context-specific estimates by household, work, school, and community settings. They utilized first-order intrinsic Gaussian Markov random field (IGMRF) priors to borrow information across neighboring age groups, but did not intrinsically enforce reciprocity. In the IGMRF context, \citet{kassteele_efficient_2017} realized that reciprocity by age and gender can be enforced with Kronecker product operations on lower triangular matrices, which also brings computational gains by cutting the parameter space approximately in half and placing regularising priors on these. \citet{dan_estimating_2023} showed how reciprocity can be leveraged to reconstruct contact patterns when there is information loss due to contact age information being reported in wide age ranges.
While effective for estimating age-by-age contact matrices, these methods lack a principled mechanism for handling reciprocity in the generalized case, and therefore do not readily extend to the inference of contact patterns by additional features.

In parallel with our work, \citet{di_domenico_individual_2025} extended the \texttt{socialmixr} approach to the generalized case. To this, a model-based approach adds the ability to handle the many zeros in sparse data by features, model marginal and interaction effects, estimate effect sizes, and utilize standard model selection and prediction techniques.

The remainder of this paper is organized as follows. Section 2 introduces basic notation and fundamental concepts, followed by our methodological contributions. Section 3 outlines the prior distributions for the parameters in our model. Section 4 describes the numerical inference procedure and predictive inference for fully-stratified intensities from partially-stratified data. Section 5 describes the model selection procedure. Section 6 validates and benchmarks our framework against existing methods through simulations. Section 7 evaluates predictive accuracy when contact feature information is missing for reported contacts. Section 8 illustrates model-based inference of generalized contact matrices on two empirical datasets. Finally, Section 9 discusses our findings, limitations, and directions for future research.
\section{Methodology}\label{sec:methodology}

\subsection{Notation, Definitions, and Key Properties}

To introduce the basic notation, key definitions, and the fundamental properties, we begin by considering the simple case of inferring contact patterns by the age of survey respondents and the age of individuals that they contact. We extend this to the generalized case in \Cref{sec:inferring-generalised-matrices}.
Let $\mcA$ denote the set of 1-year age groups (e.g., $\{0,1,\ldots,83,84\}$) or disjoint age ranges (e.g., $\{\text{0-4},\text{5-9},\ldots,\text{75-79},\text{80-84}\}$). 
For $a,b \in \mcA$, $Y_{a,b}$ denotes the random variable for the total number of contacts reported by $N_a$ respondents of age $a$ had with individuals of age $b$ during a fixed time window.
Next, let $Z_{a,b}$ be the unknown population-level equivalent of the observed $Y_{a,b}$ in the sample, i.e., the random variable for the total number of contacts by $P_a$ individuals of age $a$ to the individuals of age $b$ in the entire population during the same time window. We consider $Z_{a,b}$ to be a quantity that exists in principle but is not observed. In this work, the population size $P_a$ and sample size $N_a$ are considered to be fixed quantities and we condition on them throughout.

The main quantity of scientific interest is the \emph{contact intensity} which is defined as the average number of contacts from a person of age $a$ in the population to all people of age $b$ in the population:
\begin{equation}\label{eq:contact-intensity-ab}
    m_{a,b} = \frac{\EE[Z_{a,b}]}{P_a}, \quad \forall a,b \in \mathcal{A}.
\end{equation}
Throughout, we assume that contact behavior of respondents are representative of the contact behavior of the population, that is:
\begin{equation} \label{eq:EY}
    m_{a,b} = \frac{\EE[Z_{a,b}]}{P_a} = \frac{(N_a/P_a)\EE[Z_{a,b}]}{N_a} = \frac{\EE[Y_{a,b}]}{N_a}, \quad \forall a,b \in \mathcal{A}.
\end{equation}
As social contacts are reciprocal events,
\begin{equation}
Z_{a,b} = Z_{b,a}, \quad \forall a,b \in \mathcal{A},
\end{equation}
which we refer to as the \emph{fundamental reciprocity property}. From Eq.~\eqref{eq:contact-intensity-ab}, it is clear that contact intensities are not reciprocal. However, the contact rates, defined as
\begin{equation}\label{eq:contact-rates-ab}
    \gamma_{a,b} = m_{a,b} / P_b, \quad \forall a,b \in \mathcal{A},
\end{equation}
are reciprocating, since $\EE[ Z_{a,b} / (P_a P_b)] = \EE[ Z_{b,a} / (P_b P_a)]$ and so $\gamma_{a,b} = \gamma_{b,a}$ for all $a,b\in\mathcal{A}$.
We can intuit the contact rate as follows. Letting $P$ be the total population size, $P = \sum_a P_a$, we have $\EE[Z_{a,b}] = [ \gamma_{a,b} (P_a/P) (P_b/P) ] \cdot P^2$. This shows that if the contact rate were constant across all $a$ and $b$, then the contact pattern would be completely determined by the demographic structure. This gives the contact rate its physical meaning: it captures the departure from purely random mixing, in this case for specific age pairs.

The contacts intensities and rates have particular aggregation properties, and it is helpful to understand them from a simple example. For any two groups $a_1, a_2 \in \mcA$ and $b_1, b_2 \in \mcA$, we have $Z_{a_1 \cup a_2, b} = Z_{a_1,b} + Z_{a_2,b}$ and $Z_{a, b_1 \cup b_2} = Z_{a,b_1} + Z_{a,b_2}$. Eqs.~(\ref{eq:contact-intensity-ab},\ref{eq:contact-rates-ab}) then give:
\begin{subequations}\label{eq:agg-props}
\begin{align}
& m_{a_1 \cup a_2, b} = \frac{P_{a_1}}{P_{a_1 \cup a_2}} m_{a_1,b} + \frac{P_{a_2}}{P_{a_1 \cup a_2}}m_{a_2,b}, \quad m_{a, b_1 \cup b_2} = m_{a,b_1} +  m_{a,b_2}, \\
& \gamma_{a_1 \cup a_2, b} = \frac{P_{a_1}}{P_{a_1 \cup a_2}} \gamma_{a_1,b} + \frac{P_{a_2}}{P_{a_1 \cup a_2}} \gamma_{a_2,b}, \quad\gamma_{a, b_1 \cup b_2} = \frac{P_{b_1}}{P_{b_1 \cup b_2}} \gamma_{a,b_1} + \frac{P_{b_2}}{P_{b_1 \cup b_2}} \gamma_{a,b_2},
\end{align}
\end{subequations}
where $P_{a \cup b} = P_{a} + P_{b}$.

\subsection{Bayesian Modeling of Generalized Contact Matrices}\label{sec:inferring-generalised-matrices}
We develop a Bayesian model for inferring generalized contact matrices, which we refer to as the \texttt{g-mix} model, and describe its key components in this section.

Let us consider inferring contact matrices by age and one additional, categorical feature (e.g., sex, socioeconomic status, or occupation) with a total of $K$ strata indexed by $k,\ell \in \mcX := \{1,\ldots,K\}$. 
For the moment, we assume that we also have information for the same categorical feature for the contacts, and refer to this as the \emph{complete-data} setting. We consider the case where the feature information is only available for respondents in \Cref{sec:partially-stratified-case}.
We let $Y^{k,\ell}_{a,b}$ denote the observed number of contacts from respondents of stratum $k$ and age $a$ in the sample to members of stratum $\ell$ and age $b$ in the population, and extend our previous notation analogously to $P^k_a$, $N^k_a$, $Z^{k,\ell}_{a,b}$, $m^{k,\ell}_{a,b}$, and $\gamma^{k,\ell}_{a,b}$. 
Note that the reciprocity and aggregation properties extend naturally to the generalized case.
We consider $Y^{k,\ell}_{a,b}$ as realizations of the following Negative Binomial model
\begin{equation}\label{eq:fully-stratified-model} 
Y^{k,\ell}_{a,b} \sim \text{NegBin}\Big(\EE[Y^{k,\ell}_{a,b}], \varphi \Big), \quad \forall k,\ell \in \mcX,\ \forall a,b \in \mcA,
\end{equation}
where $\varphi$ is a global overdispersion parameter such that such that $\operatorname{Var}[Y] = \EE[Y] + \EE[Y]^2/\varphi$.
Other count distributions such as the Poisson or Beta Negative Binomial~\citep{ling_contact_2026} can also be used. Generalizing (\ref{eq:EY}) and (\ref{eq:contact-rates-ab}), we specify a log linear model on the expected number of total reported contacts,
\begin{subequations}\label{eq:fully-stratified-linear-predictor}    
\begin{align}
    \log \EE[Y^{k,\ell}_{a,b}] & = \log m^{k,\ell}_{a,b} + \log N^{k}_a \\
    \log m^{k,\ell}_{a,b} & = \log \gamma_{a,b} + \log \delta^{k,\ell}_{a,b} + \log P^\ell_b.
\end{align}
\end{subequations} 
Here, $\gamma_{a,b}$ is the baseline age-specific contact rate, and we refer to the $\delta^{k,\ell}_{a,b}$ as the \emph{complete contact rate modifiers}, as they modify the age-specific contact rate in both $k$ and $\ell$ to account for heterogeneity given the additional feature. Thus, if $\delta^{k,\ell}_{a,b} > 1$ for some $k$ and $\ell$, then the corresponding population groups mix more than expected given the age-specific contact rate, and if $\delta^{k,\ell}_{a,b} < 1$, then they mix less than expected. We assume the population sizes $P^\ell_b$ are known for all $\ell$ and $b$. We refer to this hierarchical Bayesian model, Eqs.~\eqref{eq:fully-stratified-model}--\eqref{eq:fully-stratified-linear-predictor}, as the \texttt{g-mix} model.

To ensure the estimated contact patterns are valid, our framework imposes structural constraints.
Following from $\EE[Z^{k,\ell}_{a,b}] = \EE[Z^{\ell,k}_{b,a}]$ and  Eqs.~\eqref{eq:agg-props}, the model parameters must satisfy:
\begin{subequations}
\begin{align}
    \text{(age reciprocity)} \quad\quad & \gamma_{a,b} = \gamma_{b,a}, \quad \forall a,b \in \mcA 
    \label{property:age-reciprocity} 
    \\    
    \text{(modifier reciprocity)} \quad\quad  & \delta^{k,\ell}_{a,b} = \delta^{\ell,k}_{b,a}, \quad \forall k,\ell \in \mcX, \ \forall a,b \in \mcA 
    \label{property:full-reciprocity} 
    \\
    \text{(marginal consistency)} \quad\quad  & \sum_{k \in \mcX} \sum_{\ell \in \mcX} \delta^{k,\ell}_{a,b} (P^k_a/P_a) (P^\ell_b/P_b) = 1 \quad \forall a,b \in \mcA.
    \label{property:full-consistency}
\end{align}
\end{subequations}
The benefits of these constraints are: they (1) ensure that the estimated contact patterns are valid and consistent with the underlying population structure; (2) allow us to cut the number of free parameters in approximately half and thereby reduce sample size requirements on contact surveys; and (3) enable us to deduce unseen contact patterns by multiple features. Proofs of all properties, constraint equations, and propositions in this section are provided in Section~A of the Supplement~\citep{dan_supplement_2026}.

We are able to satisfy all $A^2(2+K^2)$ constraint equations through a combination of parameter permutations and simplex transformations.
We can meet age reciprocity, Eq.~\eqref{property:age-reciprocity}, by considering only the lower triangular part of the $\gamma_{a,b}$ matrix and then filling in the upper triangular part as in~\citep{kassteele_efficient_2017, dan_estimating_2023}. 
Next, the terms that constitute the sum in \eqref{property:full-consistency} must sum to one over $(k,\ell)$, and thus lie in a simplex space for each fixed age pair $(a,b)$. 
Since optimization or Monte Carlo sampling across many simplexes is computationally difficult, we want to construct the actual model parameters in an unconstrained latent real space and then use one of many possible link functions to transform into all of our simplexes.
Of these link functions, we want one that is permutation equivariant, and thus preserves modifier reciprocity, Eq.~\eqref{property:full-reciprocity}. 

To develop this technically, it is helpful to collect the modifiers and the product population size proportions into $K^2$-dimensional vectors for each age pair $(a,b)$,
\begin{equation*}
\bdelta_{a,b} := \text{vec}\left( \big\{\delta^{k,\ell}_{a,b}: k,\ell \in \mcX\big\} \right) \in \RR^{K^2}, \ \ \bm{s}_{a,b} := \text{vec}\left( \bigg\{\frac{P^k_a}{P_a} \frac{P^\ell_b}{P_b}: k,\ell \in \mcX\bigg\} \right) \in \RR^{K^2}.
\end{equation*}
Reciprocity corresponds to a particular permutation of the elements of $\bdelta_{a,b}$. We define the permutation matrix $\Pi \in \RR^{K^2 \times K^2}$, which swaps the indices corresponding to $(k,\ell)$ and $(\ell,k)$. 
%
%
Specifically, the matrix is constructed as $\Pi_{(\ell-1)K + k, (k-1)K + \ell} = 1$ for $k,\ell \in \{1,\ldots,K\}$, with other entries set to zero. This allows us to express modifier reciprocity as $\bdelta_{b,a} = \Pi \bdelta_{a,b}$, and consistency as $\bdelta_{a,b}^\top \bm{s}_{a,b} = 1$ for all $a,b\in\mcA$.
%
%
We now relate $\bdelta_{a,b}$ to an unconstrained latent parameter vector $\bomega_{a,b}$ via a link function that maps into the $K^2$-simplex space $\mathbb{S}^{K^2}$, $h: \RR^{K^2} \mapsto \mathbb{S}^{K^2}$: $\bdelta_{a,b} = h(\bomega_{a,b}) \oslash \bm{s}_{a,b}$,  where $\oslash$ denotes element-wise division. 
To satisfy Eq.~\eqref{property:full-reciprocity}, $h$ must be permutation equivariant, $\Pi h(\bomega) = h(\Pi \bomega)$. Under this condition, if we enforce symmetry on the latent parameters such that $\bomega_{b,a} = \Pi \bomega_{a,b}$, the symmetry of the resulting modifiers $\bdelta_{b,a} = \Pi \bdelta_{a,b}$ is automatically guaranteed. 
There are three commonly used simplex maps: the additive log-ratio transform \citep{aitchison_statistical_1982}, centered log-ratio transform (a.k.a. the inverse softmax), and the isometric log-ratio transform \citep{egozcue_isometric_2003}. Among them, the additive log-ratio transform and the centered log-ratio transform are permutation equivariant (see section A of \citet{dan_supplement_2026}). We chose the softmax as it is easier to implement in practice:
\begin{equation}\label{eq:deriving-delta-vector}
    \bdelta_{a,b} = h(\bomega_{a,b}) \oslash \bm{s}_{a,b}
    = \Bigg[ \frac{\exp(\omega^{1,1}_{a,b})}{\sum_{k,\ell} \exp(\omega^{k,\ell}_{a,b})} \big/ \frac{P_a^1 P_b^1}{P_a P_b}, 
    \cdots, 
    \frac{\exp(\omega^{K,K}_{a,b})}{\sum_{k,\ell} \exp(\omega^{k,\ell}_{a,b})} \big/ \frac{P_a^K P_b^K}{P_a P_b} \Bigg].
\end{equation}

Figure~\ref{fig:schematic} visualizes the transformation process of the unconstrained latent parameters $\bomega_{a,b}$ in our model, which we organize into  a 3D tensor $\bOmega$ of dimensions $K^2 \times A \times A$. We can view this tensor in two distinct ways, handling the constraints and regularization in the context of sparse data. 
First, we have fibers handling constraints. If we fix the age pair $(a,b)$ and look down through the tensor, we have a vector (or fiber) containing the values for all $(k,\ell)$ pairs at that specific age pair. The softmax function is applied fiber-wise. It takes this vector of unconstrained values, exponentiates then, normalizes the results relative to its sum, and produces valid proportions.
Second, we have slices handling regularization. If we fix a specific stratification pair $(k,\ell)$, e.g., the contacts between male and female individuals, we obtain an $A \times A$ matrix $\bOmega^{k,\ell}$. Since contact patterns typically vary smoothly with age, we place a spatial smoothing prior (see \Cref{sec:priors} for details) on each $\bOmega^{k,\ell}$ slice. This allows the model to borrow strength across adjacent age groups.

Due to reciprocity, we can reduce the dimension of the tensor by approximately half. We can classify contacts into two types: between-strata (interactions between different groups, $k \neq \ell$) and within-stratum (interactions within the same group, $k = \ell$). In the between-strata case, reciprocity requires that the matrix for $(k, \ell)$ be the transpose of the matrix for $(\ell,k)$. Consequently, we only need to estimate one of these. In the within-stratum case ($k = \ell$), reciprocity requires that the matrix itself be symmetric. Therefore, we only need to estimate the lower triangular elements including the diagonal.

\begin{figure}
    \centering
    \includegraphics[width=\linewidth]{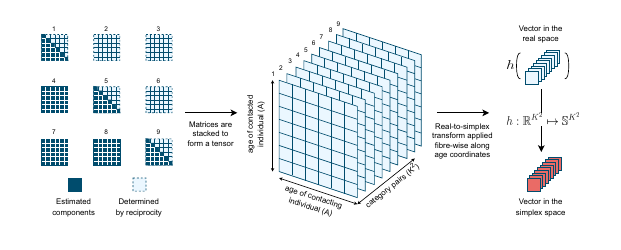}
    \caption{Schematic of the estimation procedure for enforcing structural constraints. The diagram illustrates the workflow for a scenario with $A=6$ age groups and a single stratification feature with $K=3$ strata. (Left) We start by estimating the individual matrices of dimension $A \times A$ that make up a tensor $\Omega$ of dimension $K^2 \times A \times A$. For contacts between members of different strata (between-strata, $k\neq\ell$), we estimate a single matrix (matrices 4, 7, 8); the reverse matrix is determined via transposition (matrices 2, 3, 6). For contact between members of the same stratum (within-stratum, $k=\ell$), we estimate only the lower-triangular elements of the matrices, mirroring them to populate the upper triangle (matrices 1, 5, 9). (Middle) These matrices are stacked to form the tensor $\Omega$. (Right) The fiber at each age coordinate $(a,b)$ is passed through a softmax transformation. This projects the estimates onto a $K^2$-simplex, ensuring that the resulting vector sums to 1.}
    \label{fig:schematic}
\end{figure}

\subsection{Requirements on Model Flexibility}
Allowing the rate modifiers $\delta^{k, \ell}_{a,b}$ to depend on $a$ and $b$ results in a highly parameterized model. A potential simplifying assumption sets $\delta^{k,\ell}_{a,b} = \delta^{k,\ell}$ for all $a,b\in \mcA$, which drastically reduce the parameter space. However, this simplification is often too restrictive as we show in the following proposition.

\begin{prop}[The Rank Condition]\label{proposition:infeasibility} Assume age-constant modifiers $\delta^{k, \ell}_{a,b} = \delta^{k, \ell}$ for all $a,b\in\mcA$. Define the $A^2 \times K^2$ matrix $\bS$ with elements, $\bS_{(a,b),(k,\ell)} = P^k_a P^\ell_b / (P_a P_b)$.
Modifier consistency~\eqref{property:full-consistency} can be expressed as $\bS\bdelta = \bm{1}_{A^2}$ where $\bdelta \in \mathbb{R}^{K^2}$ is the vector of age-constant modifiers and $\bm{1}_A^2$ is a $A^2$ vector of ones. The following hold:
\begin{enumerate}
    \item[(1)] If $\operatorname{rank}(\bS) = K^2$ and $A^2 \geq K^2$, then $\bdelta = \bm{1}_{K^2}$ is the unique solution.
    \item[(2)] If $\operatorname{rank}(\bS) < K^2$, infinitely many solutions exist.
\end{enumerate}
\end{prop}

In most practical settings where $A$ is moderately large and different strata have distinct age distributions, the matrix $\bS$ achieves full column rank with $A^2 \geq K^2$; this falls into the first case where the unique solution is $\delta^{k,\ell} = 1$ for all $k,\ell$. This defeats the objective: the model is insufficiently flexible, must always return $1$'s regardless of the data, and is therefore unable to answer the scientific question.
If one wishes to use age-constant modifiers and is prepared to reduce the number of age groups $A$, say to 3 or 4 age groups, the rank condition must be checked: when $\operatorname{rank}(\bS) < K^2$ (the second case), the simplification is feasible and infinitely many non-trivial modifier vectors exist. If $\bS$ retains full column rank despite the reduction in $A$, then $A$ must be reduced further or the modifiers must be allowed to vary by age.

\subsection{Modeling of Partially Observed Data}\label{sec:partially-stratified-case}
In practice, survey respondents can provide accurate information on their own income, vaccination status, or occupation, but they rarely know such detailed information for their contacts. We refer to this as the \emph{partial-data} setting.
In this case, we only have data $Y^{k,.}_{a,b}$ on the total number of contacts from respondents of stratum $k$ and age $a$ to members of age $b$ in the population.
The corresponding model on the partial data follows analogously from the complete-data model, 
\begin{subequations}\label{eq:partially-stratified-model}
\begin{align}
    Y^{k,.}_{a,b} & \sim \text{NegBin}(\EE[Y^{k,.}_{a,b}], \varphi)     
    \\
    \log \EE[Y^{k,.}_{a,b}] & = \log\gamma_{a,b} + \log \delta^{k,.}_{a,b} + \log P_b + \log N^{k}_a
    \\
    \bdelta_{a,b} & 
    = \Bigg[ \frac{\exp(\omega^{1,.}_{a,b})}{\sum_{k} \exp(\omega^{k,.}_{a,b})} \big/ \frac{P_a^1}{P_a}, 
    \cdots, 
    \frac{\exp(\omega^{K,.}_{a,b})}{\sum_{k} \exp(\omega^{k,.}_{a,b})} \big/ \frac{P_a^K}{P_a} \Bigg].\label{eq:partially-stratified-model-c}
\end{align}
\end{subequations}
The corresponding consistency constraint is $\sum_{k \in \mcX} \delta^{k,.}_{a,b} (P^k_a/P_a) = 1$ for all $a,b\in\mcA$, and is met automatically by the softmax transformation in Eq.~\eqref{eq:partially-stratified-model-c}. The partial modifier is related to the complete modifier by $\delta^{k,.}_{a,b} = \sum_{\ell} \delta^{k,\ell}_{a,b} P^\ell_b / P_b$, i.e., the partial-data model marginalizes over the unobserved contact strata.
Correspondingly, we need to consider the $A \times K$ matrix $\bS'$ with elements $\bS'_{a,k} = P^k_a / P_a$, and replace in the rank conditions $K^2$ with $K$. 

Crucially, because of reciprocity, feature information on the respondents alone provides insights beyond the partial contact intensities $m^{k,.}_{a,b}$. To see this, imagine a contingency table with margins $Z_{a,b}^{k,.}$ and $Z_{a,b}^{.,\ell}$ for fixed $a,b$ and a table over $k,\ell \in \mathcal{X}$. First, if we learn the margins $m^{k,.}_{a,b}$ with our model Eq.~\eqref{eq:partially-stratified-model}, the other margin is determined by
\begin{equation}
    P_b m^{.,\ell}_{b,a} = \EE[Z_{b,a}^{.,\ell}] = \EE[Z_{a,b}^{\ell,.}] = P^\ell_a m^{\ell,.}_{a,b}, \quad \forall \ell \in \mcX, \forall a,b \in \mcA.
\end{equation}
Next, given fixed margins, we can place bounds on the inner cells of the contingency table due to the aggregation properties for each $a,b$. Any posterior predictive distribution on the complete contact intensities $m^{k,\ell}_{a,b}$ must satisfy these bounds, and additionally the constraints that arise from the reciprocity in the complete-data contacts $Z_{a,b}^{k,\ell}$.  

A central quantity that characterizes pathogen spread is the reproduction number $R_0$, the expected number of infections that one person generates in an entirely susceptible population. 
To estimate $R_0$ in structured populations stratified by $a,b,k,\ell$, we need the next-generation matrix (NGM)
$\bK$, which quantifies the expected number of infections that one infectee in $(a,k)$ generates among people in $(b,l)$~\citep{diekmann_mathematical_2000}. 
The standard notation is recipient-first, such that $\bK_{b,a}^{\ell,k} = \beta^{k,\ell}_{a,b} m^{k,\ell}_{a,b} D^k_a (S^{\ell}_b / P^{\ell}_b)$, where $\beta^{k,\ell}_{a,b}$ are per-contact transmission rates, $m^{k,\ell}_{a,b}$ are the complete contact intensities, $D^k_a$ is the mean duration of infection in one infectee in $(a,k)$,  and $S^{\ell}_b$ are the number of susceptible people in $(b,l)$. 
The reproduction number $R_0$ is the dominant eigenvalue of $\bK$. Thus, small changes in $\bK$ can disproportionally affect $R_0$ when they align with the dominant eigenvector. 
In particular, different levels of assortativity in the complete contact intensities can have a disproportionate impact, even when the partial intensities $m^{k,.}_{a,b}$ remain unchanged.
This motivates the following target quantity for estimation, which we refer to as the \emph{attributable fractions}:
\begin{equation}
    \eta_{a,b}^{k,\ell}  =
    \frac{m^{k,\ell}_{a,b}}{m^{k,.}_{a,b}},  
    \quad \text{such that} \quad \sum_\ell \eta_{a,b}^{k,\ell} = 1.\label{eq:eta} 
\end{equation}
Intuitively, $\eta_{a,b}^{k,\ell}$ is the unknown proportion of contacts that members of stratum $k$ age $a$ have with members of stratum $\ell$ age $b$ and so if $\eta_{a,b}^{k,k} > 1/K$, then within-group contact mixing is assortative, either due to demographics or preferential contact rates. 
We have the following result, which follows from~\citep{dobra_bounds_2000}:\footnote{The bounds in Eq.~\eqref{eq:feasible-mixing-limits} are instances of the classical Fr\'{e}chet--Hoeffding bounds, which provide the sharpest possible bounds on a joint distribution given fixed marginal distributions~\citep{dobra_bounds_2000}. In our setting, the joint distribution is over the attributable fractions and the marginals are determined by the two contact-count margins $\EE[Z^{k,.}_{a,b}]$ and $\EE[Z^{\ell,.}_{b,a}]$ recovered from the partially-stratified model via reciprocity.}
\begin{prop}[Feasible Mixing Bounds]
For any $k,\ell\in\mcX$ and $a,b \in \mcA$, a \emph{sharp} (tightest) bound for $\eta^{k,\ell}_{a,b}$ is given by: 
\label{proposition:frechet-hoeffding-bounds}
\begin{equation}\label{eq:feasible-mixing-limits}
    \max \left(0, \ \frac{\EE[Z^{\ell,.}_{b,a}]}{\EE[Z^{k,.}_{a,b}]} - \sum_{t\neq k} \frac{\EE[Z^{t,.}_{a,b}]}{\EE[Z^{k,.}_{a,b}]} \right) \leq \eta^{k,\ell}_{a,b} \leq \min \left(1, \ \frac{\EE[Z^{\ell,.}_{b,a}]}{\EE[Z^{k,.}_{a,b}]} \right)
\end{equation}
\end{prop}

\Cref{proposition:frechet-hoeffding-bounds} may allow us, depending on the values of the inferred modifiers and known population sizes, to constrain the attributable fractions to lie in an interval tighter than $[0,1]$. This can provide information on posterior predictions of the complete contact intensities via $m^{k,\ell}_{a,b} = m^{k,.}_{a,b}\eta^{k,\ell}_{a,b}$, the elements of $\bK$, and $R_0$ under assumptions on the transmission rates and infection durations.

\subsection{Extension to Multiple Features}\label{sec:multi-variable-extension}
In most applications, we wish to stratify by multiple features. 
Let $\mcX_1, \ldots, \mcX_J$ denote $J$ distinct sets of stratification features, where the $j$-th feature has $K_j$ categories. The complete set of strata is the Cartesian product $\mcS = \mcX_1 \times \ldots \times \mcX_J$. For instance, if stratifying by gender and income, an element $s \in \mcS$ represents a specific tuple such as (Female, High Income). The total number of unique strata in $\mcS$ is  $K_* = \prod^J_{j=1} K_j$. To distinguish this multi-feature setting from the single-feature case while keeping the notation simple, we index elements of $\mcS$ using $s$ (for source) and $t$ (for target).

The structural constraints derived previously expand naturally to this setting: we still require a tensor $\bOmega$ that satisfies reciprocity, which is then mapped to the simplex via the softmax function. However, dimensionality poses a significant challenge, as $\bOmega$ is of dimension $K^2_* \times A \times A$, and $K_*$ grows exponentially with $J$. We thus assume that the tensor can be decomposed into the sum of contributions from each individual stratification feature. This assumption is equivalent to a no-interaction log-linear model on the latent parameters $\omega$, implying that the contact rate modifier for a stratum tuple $(s,t)$ is a product of the per-feature modifier contributions. For this practically important case, we can construct the global tensor $\bOmega$ as Kronecker sums of component tensors $\{\bOmega^{(j)}\}^J_{j=1}$ each associated with a single feature $\mcX_j$:
\begin{equation}\label{eq:omega_kronecker_sum}
    \bOmega = \bOmega^{(1)} \oplus_1 \bOmega^{(2)} \oplus_1 \cdots \oplus_1 \bOmega^{(J)}
\end{equation}
where $\oplus_1$ denotes the Kronecker sum on the first mode (feature stratification dimension) of the tensor. The following result confirms that this approach preserves reciprocity: 

\begin{prop}[Preservation of Reciprocity Under Kronecker sums]
\label{proposition:kronecker-symm}
For each $j = 1,\dots,J$, let $\bOmega^{(j)} \in \mathbb{R}^{K_j^2 \times A \times A}$ denote the tensor associated with the $j$-th feature. Assume that, for each $j$, there exists a permutation $\pi_j : \{1,\dots,K_j^2\} \to \{1,\dots,K_j^2\}$ corresponding to the transpose of a $K_j \times K_j$ matrix (under a fixed identification of the first mode with matrix entries) such that the following reciprocity condition holds:
$$
\bigl(\bOmega^{(j)}\bigr)_{k,a,b}
=
\bigl(\bOmega^{(j)}\bigr)_{\pi_j(k),\,b,a}
\quad\text{for all } k \in \{1,\dots,K_j^2\},\ a,b \in \{1,\dots,A\}.
$$
Define the global tensor
\[
\bOmega
=
\bOmega^{(1)} \oplus_1 \cdots \oplus_1 \bOmega^{(J)}.
\]
Let $K_* = \prod_{j=1}^J K_j$ so that the first mode of $\bOmega$ has dimension $K_*^2$. Let $\pi : \{1,\dots,K_*^2\} \to \{1,\dots,K_*^2\}$ be the permutation induced by the $\pi_j$s via the Kronecker product of their permutation matrices, and corresponding to the transpose of a $K_* \times K_*$ matrix under the analogous identification of the first mode. Then $\bOmega$ satisfies the same reciprocity symmetry:
$$
\bOmega_{s,a,b}
=
\bOmega_{\pi(s),\,b,a}
\quad\text{for all } s \in \{1,\dots,K_*^2\},\ a,b \in \{1,\dots,A\}.
$$
\end{prop}

In short, \Cref{proposition:kronecker-symm} tells us that if the tensor for every feature satisfies reciprocity, then the global tensor constructed via mode 1 Kronecker sums will also satisfy reciprocity. This result enables us to build models incrementally in an additive fashion, considering the contribution of each feature separately and where necessary consider interaction effects between multiple features, all while maintaining valid contact intensities that satisfy reciprocity constraints.
To model interactions between two features, we merge their categories into a single composite feature (e.g., combining sex and SES into a joint stratum) and treat it as one feature within the additive framework in Eq.~\eqref{eq:omega_kronecker_sum}.
\section{Priors}\label{sec:priors}

In this section, we describe our prior specifications, which build primarily on Bayesian penalized splines for inference, and truncated Dirichlet priors for predictive inference.

\subsection{Bayesian Penalized Splines}
Bayesian P-splines \citep{lang_bayesian_2004} are particularly effective when the contact patterns exhibit non-isotropic features. We represent a vectorized matrix $\bm{f} \in \mathbb{R}^{A^2}$ as a linear combination of cubic B-spline basis functions. Let $\bm{\Phi}_1 \in \RR^{A \times M_1}$ and $\bm{\Phi}_2 \in \RR^{A \times M_2}$ be matrices of basis functions evaluated at each age where $M_1$ and $M_2$ are the number of basis functions. Let $\bm{\Phi} = \bm{\Phi}_1 \otimes \bm{\Phi}_2 \in \RR^{A^2 \times M_1 M_2}$ be the Kronecker product of the marginal basis matrices. Overfitting is mitigated by introducing difference penalties through a second order IGMRF prior on the basis coefficients $\boldsymbol{\beta}$, resulting in the induced prior,
\begin{equation}\label{prior:adaptive-splines}
\bm{f} = \bm{\Phi} \bm{\xi}, \quad \bm{\xi} \sim \text{MVNormal}\left(0, (\tau\bQ)^{-1}\right)
\end{equation}
where $\tau > 0$ is a precision parameter and $\bQ = \bQ_0 \oplus \bQ_0$ is constructed from the second-order difference matrix $\bQ_0 = \bD^\top\bD$; (for the exact form of $\bD$ see \citet{dan_supplement_2026}).

\subsection{Prior for the Baseline Contact Rates} We decompose the baseline age-by-age log contact rate into a baseline parameter and a smooth function evaluated on the age grid that captures deviations from this baseline: $\log \gamma_{a,b} = \beta_0 + f(a,b)$. 
The prior for $f(a,b)$ is the Bayesian P-Spline with $\tau \sim \text{Gamma}(2.0, 0.01)$, constrained to the space of symmetric matrices to meet the reciprocity requirements, which we achieve by evaluating the basis functions solely on the lower-triangular age pairs $\{(a,b) : a \ge b\}$.
The prior for $\beta_0$ is a centred Gaussian prior $\beta_0 \sim \text{Normal}(\bar{\beta}_0, 2.5^2)$ where $\bar{\beta}_0 = -(A^2K^2_*)^{-1} \sum_{\ell,b} \log P^\ell_b$ for the fully stratified model and $\bar{\beta}_0 = -A^{-2} \sum_b \log P_b$ for the partially stratified model. 
In our applications, the population sizes are large, typically ranging from thousands to millions. The centering neutralizes the effect of the offset terms $\log P^\ell_b$ (or $\log P_b$), allowing $f(a,b)$ to focus on capturing variations in contact rates, rather than expending flexibility to absorb the scale of the population sizes.

\subsection[Prior for the omega tensor]{Prior for the $\bOmega$ Tensor}
We place Bayesian P-Spline priors on the latent parameters $\bomega_{a,b}$ for all $a,b$. The description below is presented for a single stratification variable for notational clarity. In the multi-feature setting (\Cref{sec:multi-variable-extension}), the strategy and the corresponding Bayesian P-Spline priors are applied identically and independently to each component tensor associated with the $j$-th stratification variable.

The priors on the latent parameters must be centered appropriately to handle extreme compositional variations within the demographic data. When certain strata constitute a minute fraction of the overall population, such as racial and ethnic minority groups, or the active workforce within older age brackets, or when there is an abrupt demographic shift between neighboring ages, such as the discrete boundary between school-aged children and adults, the inverse population proportions $P^k_a P_b^\ell / (P_a P_b)$ (full stratification) or $P^k_a / P_a$ (partial stratification) can be very large or contain sharp boundaries that overwhelm the smoothing priors.
By anchoring the splines to the inverse softmax of the population proportion inverses, we ensure that if the smoothing function evaluates to zero, the resulting modifiers evaluate to exactly 1:
$$
\mathbb{E}[\delta_{a,b}^{k,\ell}] = \mathbb{E}\left[\{h(\bomega_{a,b})\}_{k,\ell}\right] \frac{P{a}}{P_{a}^{k}}\frac{P_{b}}{P_{b}^{\ell}} = 1.
$$
This centering strategy effectively neutralizes the influence of small population sizes and sharp demographic boundaries, allowing the splines to focus entirely on capturing true structural variations in mixing behavior.
To achieve this, we first define $\mathbf{U}_{k,\ell} \in \mathbb{R}^{A \times A}$ as the matrix containing the product of the population proportions $P_{a}^{k}P_{b}^{\ell}/(P_{a}P_{b})$. Let $\mathcal{U}$ be the $K^2 \times A \times A$ tensor where the matrices $\mathbf{U}_{k,\ell}$ are stacked along the first dimension. We apply the inverse softmax function (equivalent to the centered log-ratio transform) to obtain the centering tensor $\mathcal{W}$. We denote this with $\mathcal{W} = h^{-1}(\mathcal{U})$ where $h^{-1}$ is understood to operate fiber-wise along the first dimension of the tensor (i.e., it is applied independently to each $K^2$-dimensional vector $\mathcal{U}_{\cdot, a, b}$ for all age pairs $a, b \in \mathcal{A}$).
Now, let $\bOmega_i$ denote the $A \times A$ matrix slice of the latent tensor $\bOmega \in \RR^{K^2 \times A \times A}$ corresponding to the $i$-th category pair, where $i = 1, \dots, K^2$. Let $\mathbf{W}_i$ be the corresponding matrix slice from the centering tensor $\mathcal{W}$. For each $i$, we place the following prior on the vectorized matrix $\text{vec}(\bomega_i)$:
\begin{equation*}
\text{vec}(\bomega_i) = \text{vec}(\mathbf{W}_i) + \bm{\Phi} \bbeta, \quad \bbeta \sim \text{MVNormal}(0, (\tau \bQ)^{-1}), \quad \tau \sim \text{Gamma}(2.0, 0.01).
\end{equation*}

\subsection[Predictive Prior for Eta]{Predictive Prior for $\bm{\eta}$}
In the case when feature information is only available for the respondents, we aim to predict the fully stratified contact intensities. The attributable fractions $\eta_{b,a}^{\ell,k}$ must sum to one across $\ell$, and should respect their feasible mixing limits. 
\citet{fienberg_iterative_1970} proposed the iterative proportional fitting algorithm (IPF) for such tasks; however IPF is a deterministic optimization algorithm that does not integrate well into the Bayesian framework, is computationally expensive to re-run for every posterior sample, and can be sensitive to starting values. 
We thus choose a truncated Dirichlet prior: 
\begin{equation}\label{eq:trunc-dir}
    \{\eta^{k,\ell}_{a,b}\}^K_{\ell=1} \sim \text{TruncDirichlet}\left(\alpha\left\{\frac{P^\ell_b}{P_b}\right\}^K_{\ell=1}; \bm{l}^k_{a,b}, \bm{u}^k_{a,b} \right),
\end{equation}
where $\alpha$ serves as a concentration parameter, and the vectors $\bm{l}^k_{a,b}$ and $\bm{u}^k_{a,b}$ contain the lower and upper bounds of \Cref{proposition:frechet-hoeffding-bounds}. A priori, we prefer to 
center our predictions on the null of proportionate mixing. For this reason, we centered the prior on the demographic proportions, i.e., $\mathbb{E}[\eta^{k,\ell}_{a,b}] = P^\ell_b / P_b$.
Reciprocity enforces a deterministic relationship between $\eta^{k,\ell}_{a,b}$ and $\eta^{\ell,k}_{b,a}$:
\begin{equation}\label{eq:eta-reciprocity-relation}
\EE[Z^{\ell,.}_{b,a}] \eta^{\ell,k}_{b,a} = \EE[Z^{\ell,k}_{b,a}] = \EE[Z^{k,\ell}_{a,b}] = \EE[Z^{k,.}_{a,b}]\eta^{k,\ell}_{a,b},
\end{equation}
and therefore we restrict prediction to only half the age pairs, where $a \geq b$. Crucially, it is guaranteed that if $\eta^{k,\ell}_{a,b}$ satisfies its feasible mixing limits, the reciprocal fraction derived via \eqref{eq:eta-reciprocity-relation} also satisfies its own bounds; see Section A of \citet{dan_supplement_2026}.
\section{Numerical inference and predictions}\label{sec:model-selection}

We estimate a mean-field variational family that approximates the joint posterior density of all latent and model parameters with stochastic variational inference (SVI; \citet{hoffman_stochastic_2013}) as implemented in the probabilistic programming library NumPyro \citep{phan_composable_2019}. We optimize the evidence lower bound with the Adam optimizer \citep{kingma_adam_2017}, modulated by a one-cycle learning rate scheduler \citep{smith_super-convergence_2018} with a maximum learning rate of 0.01. Across all simulations and applications, the SVI algorithm is run for 20,000 iterations for convergence against the objective function. Once the optimized variational family is obtained, we generate 3,000 independent samples from the variational posterior to compute posterior expectations and quantile-based 95\% credible intervals. 

When the model is fitted to partial data, we also predict the attributable fractions $\bm{\eta}$. This is done as follows: for each sample from the variational posterior, we calculate $\EE[Z_{a,b}^{k,\ell}]$, derive the mixing bounds, and use the conditional sampler of \citet{ng_truncated_2011} to obtain one sample from the truncated Dirichlet Eq.~\eqref{eq:trunc-dir}.

Our implementation is available open-source in the Python package \texttt{cntmosaic}, \url{https://github.com/ShozenD/cntmosaic}.
\section{Model Selection}

Our model-based framework enables us to use well-established Bayesian machinery for model comparison, and for selecting the additive or interaction effects of features that best capture heterogeneity in contact patterns. 
We proceed in two stages. 
First, for every candidate model, we use 5-fold cross-validation to determine the optimal number of penalized spline basis functions in terms of Expected Log Predictive Density (ELPD). We started from 40 spline basis functions for the age-age contact rate surface and the contact rate modifier surfaces, and then iteratively reduced these by 5.
Second, following knot calibration, we compare the predictive performance of candidate models with e.g., different features using Leave-One-Out Cross-Validation (LOO-CV)~\citep{vehtari_practical_2017}. 
In this study, we conduct all comparisons using the common maximal stratification determined by the union of all of the covariates under consideration as in~\citet{dellaportas_markov_1999}. This ensures that coarser models were comparable against models with a higher degree of stratification and that all models were scored on their ability to predict the data at its finest granularity. 
%
%
When models exhibit similar predictive performance, we adhere to the parsimony principle and select the model with fewer parameters.
As LOO-CV applied to variational posteriors can be unreliable when the variational approximation is poor, we monitored Pareto-$\hat{k}$ diagnostics~\citep{vehtari_practical_2017} and found no problematic values in our applications.
%
\section{Validation and benchmarking}\label{sec:validation-benchmarking}
We validated and benchmarked the \texttt{g-mix} model of ~\Cref{sec:methodology} against established methods on simulated social contact data of increasingly many features. 
We tested on 5 scenarios, ranging from a scenario with only age stratification of respondents and contacts, to a scenario with three binary features in addition to age (Table~\ref{tab:experiment-III}).
We limited ourselves to a sample size of $N=1500$ survey respondents, which is typical of contemporary social contact surveys~\citep{mossong_social_2008, feehan_quantifying_2021, liu_rapid_2021}, and assumed that survey data were complete, i.e., that respondents reported the age and features of themselves and of their contacts.

\begin{table}
    \begin{tabularx}{\textwidth}{ c c c c >{\raggedright\arraybackslash}X } 
    \toprule
    Scenario & Cov. Structure & Matrices $K^2$ & $N/K^2$ & Rationale \\ 
    \midrule
    A & 1 & 1 & 1500 & Age by age stratification only. \\
    B & 2 & 4 & 375 & Basic reciprocal pair (e.g., Male-Female).\\
    C & 4 & 16 & 93.8 & Granular reciprocal matrix (e.g. Income).\\
    D & $2 \times 3$ & 36 & 41.7 & Cross-stratification (e.g., Sex $\times$ Educ.).\\
    E & $2 \times 2 \times 2$ & 64 & 23.4 & Multi-variable stratification. \\
    \bottomrule
    \end{tabularx}

    \vspace{0.1cm} 

    \begin{tabular*}{\textwidth}{@{\extracolsep{\fill}}llccccc@{}} 
    \toprule
    Metric & Model & \multicolumn{5}{c}{Scenario, mean (std.)} \\ 
    \cmidrule(lr){3-7} 
     & & A & B & C & D & E \\ 
    \midrule
    \multirow{7}{*}{MAPE} 
    & socialmixr & 22.9 (6.6) & -- & -- & -- & -- \\
    & socialmixr ext. & 22.9 (6.6) &\red{32.5} (9.5) & \red{53.7} (13.6) & \red{75.4} (17.2) & \red{97.1} (22.0) \\
    & Prem & \red{37.5} (20.3) & -- & -- & -- & -- \\
    & Prem-ext. & \red{37.5} (20.3) & 30.2 (10.0) & 39.8 (12.7) & 48.2 (15.1) & 56.1 (16.8) \\
    & vdK INLA & -- & 16.9 (4.5) & -- & -- & -- \\
    & vdK SVI & -- & 29.2 (3.6) & -- & -- & -- \\
    & vdK-ext. & 17.5 (2.7) &29.2 (3.6) & 42.2 (10.4) & 62.0 (25.9) & 81.1 (40.6) \\
    & g-mix & \blu{13.6} (4.3) & \blu{14.9} (4.4) & \blu{19.1} (4.9) & \blu{23.6} (5.7) & \blu{24.6} (5.5) \\
    \midrule
    \multirow{7}{*}{\parbox{1cm}{Interval\\score}} 
    & socialmixr & 0.10 (0.02) & -- & -- & -- & -- \\
    & socialmixr-ext. & 0.10 (0.02) & 0.09 (0.02) & 0.08 (0.01) & 0.08 (0.01) & 0.08 (0.02)\\
    & Prem & \red{0.25} (0.30) & -- & -- & -- & -- \\
    & Prem ext. & \red{0.25} (0.30) & 0.11 (0.02) & 0.10 (0.02) & \red{0.10} (0.02) & \red{0.10} (0.02)\\
    & vdK INLA & -- & 0.08 (0.02) & -- & -- & -- \\
    & vdK SVI & -- &\red{0.18} (0.02) & -- & -- & --\\
    & vdK-ext. & 0.13 (<0.01) &\red{0.18} (0.02) & \red{0.10} (0.01) & 0.07 (0.01) & 0.05 (0.01)\\
    & g-mix & \blu{0.04} (0.01) & \blu{0.04} (<0.01) & \blu{0.03} (<0.01) & \blu{0.02} (<0.01) & \blu{0.02} (<0.01) \\
    \midrule
    \multirow{7}{*}{\parbox{1cm}{95\%\\Coverage}} 
    & socialmixr & 65.0 (4.8) & -- & -- & -- & -- \\
    & socialmixr-ext. & 65.0 (4.8) &\red{81.2} (3.4) & \red{80.5} (5.9) & 75.4 (17.2) & \red{60.1} (6.9) \\
    & Prem & \red{50.3} (19.1) & -- & -- & -- & -- \\
    & Prem ext. & \red{50.3} (19.1) & 87.6 (2.5) & \blu{93.5} (2.1) & \blu{95.5} (1.9) & \blu{96.5} (1.7) \\
    & vdK INLA & -- & \blu{96.1} (1.2) & -- & -- & -- \\
    & vdK SVI & -- & 88.0 (1.9) & -- & -- & -- \\
    & vdK-ext. & \blu{86.4} (3.0) &88.0 (1.9) & 80.8 (3.3) & \red{71.8} (5.7) & 65.9 (7.4) \\
    & g-mix & 69.2 (9.2) & 90.5 (3.5) & 92.5 (3.0) & 86.4 (5.5) & 86.6 (4.9) \\
    \midrule
    \multirow{7}{*}{\parbox{1cm}{Inference\\Time (secs)}} 
    & socialmixr & \blu{10} (0.5) & -- & -- & -- & -- \\
    & socialmixr ext. & \blu{10} (0.5) & 22 (0.7) & \blu{42} (1.3) & \blu{64} (2) & \blu{87} (3) \\
    & Prem & 24 (2) & -- & -- & -- & -- \\
    & Prem-ext. & 24 (2) & 40 (5) & 68 (12) & 102 (16) & 148 (22) \\
    & vdK INLA & -- & \blu{18} (2) & -- & -- & -- \\
    & vdK SVI & -- &\red{3095} (324) & -- & -- & -- \\
    & vdK-ext. & \red{159} (87) &\red{3095} (324) & \red{3597} (640) & \red{4185} (572) & \red{5655} (981) \\
    & g-mix & 63 (13) & 235 (64) & 630 (110) & 1139 (186) & 1530 (206) \\
    \bottomrule
    \end{tabular*}
    \vspace{0.1cm}
    \caption{We benchmarked \texttt{g-mix} performance on estimating generalised contact matrices against existing methods under various scenarios in the complete data setting. The top panel outlines the experimental scenarios considered, with the number of survey respondents fixed at $N=1,500$ across all scenarios. The bottom panel details results for several benchmarking metrics. The best- and worst-performing models are highlighted in blue and red, respectively. The original implementations of socialmixr, Prem, and vdK did not natively support generalized features; we implemented extensions, denoted with the suffix "-ext". For specific scenarios, our extensions corresponded to the original. Detailed descriptions of all models and our extensions are in \citet{dan_supplement_2026}. The number of basis functions we use for each dimension in our penalized Bayesian spline prior for the contact rates and contact modifiers matrices is fixed at $15$ throughout.}
    \label{tab:experiment-III}
\end{table}

We benchmarked \texttt{g-mix} against the original \texttt{socialmixr} bootstrapping approach \citep{funk_socialmixr_2024}, the IGMRF based methods of \citet{prem_projecting_2017} and \citet{kassteele_efficient_2017}, but also extended these three methods to handle multiple features, and tested against these extended versions too. 
Conceptually, the extended \texttt{socialmixr} approach implements bootstrapping and thus represents a model-agnostic approach with no regularizing priors and no latent constraints, but enforces reciprocity post-hoc. The extended model of~\citet{prem_projecting_2017} induces regularization of the joint posterior through the IGMRF prior and hierarchical model structure, but does not intrinsically enforce reciprocity. The extended method of~\citet{kassteele_efficient_2017} enforces reciprocity and induces regularization but does not implement the consistency constraints of~\eqref{property:full-consistency}.
The extended models were implemented in NumPyro and likewise fitted with SVI, with the exception of \texttt{socialmixr} \citep{funk_socialmixr_2024}, which we re-implemented in Python (NumPy) (see sections Sections D to F in \citet{dan_supplement_2026}). For the scenario where there is one feature with two categories, we also benchmarked against the Integrated Nested Laplace Approximation (INLA) implementation in \citet{kassteele_efficient_2017}.
We first outline the data generating process for our simulation experiments and then present our benchmarking results. 

To simulate social contact data, we used a reference population (such as from the US census) and generated smooth stochastic deviations to obtain feature-specific population sizes. Survey respondents were then directly sampled from the synthetic population, corresponding to a representative survey design. 
To simulate contact data for these respondents, we first generated a baseline contact intensity matrix by mixing and transforming empirically derived contact pattern templates for different physical settings~\citep{mistry_inferring_2021}. 
Their data-driven approach provides template intensity matrices for contacts occurring in the household $\bT^{(H)}$, at school $\bT^{(S)}$, at work $\bT^{(W)}$ and in the community $\bT^{(C)}$  (\Cref{fig:simulation}, top row). From these, we simulated a single baseline age-by-age contact intensity matrix as a randomly weighted combination of these templates:
\begin{equation} \label{sim:template-mixture}
\bT = v^{(H)}\bT^{(H)} + v^{(S)}\bT^{(S)} + v^{(W)}\bT^{(W)} + v^{(C)}\bT^{(C)}
\end{equation}
where the mixing weights $(v^{(H)}, v^{(S)}, v^{(W)}, v^{(C)}) \sim \operatorname{Dirichlet}(1,1,1,1)$. Because $\bT$ does not satisfy reciprocity, we applied the reciprocity correction in~\citet{funk_socialmixr_2024}. We obtained the corresponding rate matrix by scaling the rows of $\bT$ with the synthetic population sizes.
Next, for each feature in $(s,t)$, we generated deviation matrices using the same template mixture approach,~\eqref{sim:template-mixture}, and transformed these further to enforce reciprocity and consistency (see Section C of \citet{dan_supplement_2026}) We then multiplied the age-by-age baseline contact rate matrix with the feature-specific deviation matrices and the synthetic population sizes to generate a fully-stratified contact intensity matrix $m^{s,t}_{a,b}$ (\Cref{fig:simulation}, second row).
With this, we simulated multi-feature contact data as follows. For each survey respondent $i = 1, \dots, N$ of age $a_i$ with features $s_i$, we generated the number of reported contacts with individuals of age $b$ and features $t$ from a Poisson process:
\begin{equation}\label{eq:contact-count-sim}
Y^t_{i,b} \sim \text{Poisson}( m^{s_i,t}_{a_i,b} \cdot \zeta_i ), \quad \zeta_i \sim \text{Gamma}(5, 5),
\end{equation}
where $\zeta_i$ is a mean-one multiplicative individual-level random effect that introduces the characteristic overdispersion typically observed in empirical social contact surveys. The data for the models are obtained by combining the replicates~\eqref{eq:contact-count-sim}: $Y^{s,t}_{a,b} = \sum_{a_i = a, s_i = s} Y^t_{i,b}$. We show an example of the generated empirical contact matrices in the bottom row of \Cref{fig:simulation}.

\begin{figure}
    \centering
    \includegraphics[width=\linewidth]{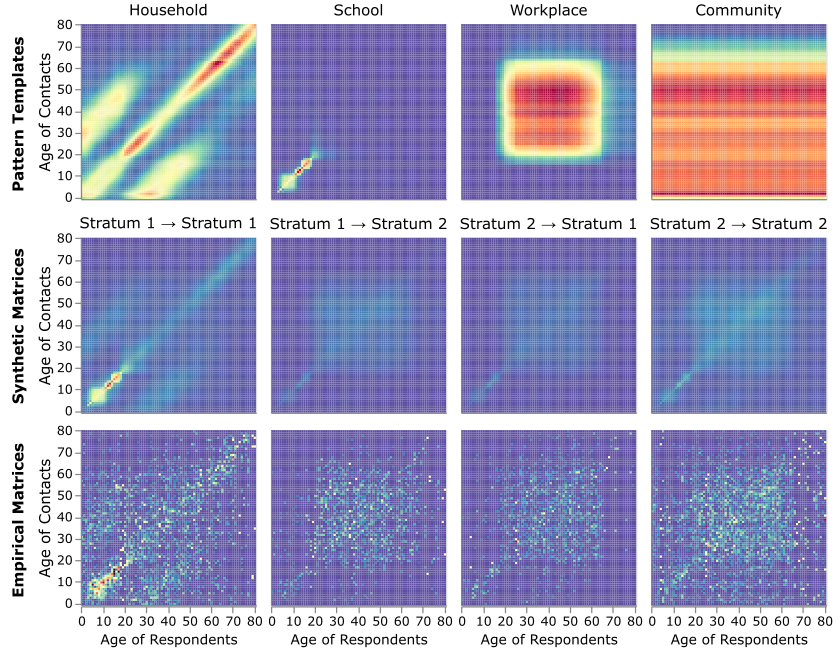}
    \caption{Illustration of the data simulation workflow for a complete data setting involving 80 1-year age bands and one binary feature (Stratum 1 and Stratum 2). The 1st row shows contact pattern templates for household, school, workplace, and community contexts obtained from \citet{mistry_inferring_2021}. These templates were used to derive the synthetic ground truth contact intensity matrices shown in the 2nd row. Synthetic contacts were then generated by a probability model under these ground truth intensities. The 3rd row shows the empirical contact intensity matrices corresponding to one synthetic social contact survey. Each cell value is calculated as $\hat{m}^{k,\ell}_{a,b} = y^{k,\ell}_{a,b}/N^k_a$ where $y^{k,\ell}_{a,b}$ are the synthetic contact counts between survey respondents in stratum $k$ and age $a$ to members of stratum $\ell$ and age $b$ in the population. $N^k_a$ is the number of respondents of stratum $k$ and age $a$. These matrices are a visual representation of what the models see as data.}
    \label{fig:simulation}
\end{figure}

\Cref{tab:experiment-III} summarizes our results.
We quantified accuracy with mean absolute percentage error (MAPE) on the complete generalized contact intensities $m^{s,t}_{a,b}$, and the quality of the uncertainty quantification via posterior 95\% coverage probability and the interval score \citep{gneiting_strictly_2007}. Inference time was measured in seconds.
We found that for typical survey sizes ($N=1,500$), the \texttt{socialmixr} approach was by far the fastest inference method but required substantial age group aggregation to avoid division-by-zero, which causes the bootstrapping to fail (see Section D of \citet{dan_supplement_2026}). This severely inflated errors for scenarios C to E. 
Prem's method \citep{prem_projecting_2017}, which is inherently specified in age ranges of 5 years, could avoid these zero-sample pitfalls without additional aggregation because the IGMRF prior provides a default estimate even when a cell contains no observations. The posterior simply reverts to the prior with no Bayesian update.
However, this is also why interval scores worsen in sparse settings: cells with little or no data are governed almost entirely by the prior, producing credible intervals that reflect prior uncertainty rather than data-informed uncertainty.

Among the two fine-age approaches (with 1-year age bins), the original model in \citet{kassteele_efficient_2017} (abbreviated vdK) achieved very fast runtimes owing to its INLA implementation.
Our SVI extension lost this speed advantage because the IGMRF priors result in dense autodifferentiable operations on the precision matrix in NumPyro, foregoing the sparse-matrix exploitations that make INLA efficient (see Section F of \citet{dan_supplement_2026}).
The extended vdK model also does not enforce the consistency constraints of~\eqref{property:full-consistency}, which likely contributed to its sharply deteriorating accuracy from scenario A to E. The overly tight credible intervals are likely an artefact of SVI rather than the structure of the model.
By contrast, \texttt{g-mix} enforces all reciprocity and consistency constraints, and its P-spline parameterization offers two further advantages: the spline prior provides stronger regularization, and inferring spline coefficients rather than full age-by-age cell values substantially reduces the effective parameter space. This explains lower error rates and faster runtimes. \texttt{g-mix} exhibited below-nominal 95\% coverage in Scenario A (69.2\%); this is a known artifact of mean-field SVI, which systematically underestimates posterior variance, and is most pronounced in the low-complexity, low-parameter regime of Scenario A. Additional experiments showed that increasing the number of basis functions in the P-Spline leads to improved posterior coverage at the cost of runtime (Table C.1 in \citet{dan_supplement_2026}).

\FloatBarrier
\section{Predictive accuracy when features of contacts are not observed}
We evaluated prediction accuracy for contact intensities when features are only observed for the respondents and not their contacts. 
We used the survey sample from Wave~1 of the Berkeley Interpersonal Contact Study (BICS)~\citep{feehan_quantifying_2021} directly, retaining the $N=2{,}627$ respondents and their demographic profiles by 1-year age groups and race/ethnicity (Hispanics, non-Hispanic whites, non-Hispanic Blacks and non-Hispanic Others). For each respondent, we simulated complete contact intensities and the corresponding contact survey data using the procedure in~\Cref{sec:validation-benchmarking} in three scenarios characterized by different degrees to which individuals in the same race/ethnicity group preferentially contacted one another: neutral, assortative, and disassortative.

First, we fitted the \texttt{g-mix} model in~\eqref{eq:fully-stratified-model} to the complete synthetic data $Y^{s,t}_{a,b}$ to obtain a best-possible reference estimate of the ground truth complete contact intensities.
Next, we withheld the race/ethnicity features of all reported contacts by setting $Y^{s,.}_{a,b} = \sum_t Y^{s,t}_{a,b}$, fitted the \texttt{g-mix} model in~\eqref{eq:partially-stratified-model} to the partial synthetic data to estimate $m^{s,.}_{a,b}$, calculated the feasible mixing bounds from the resulting posterior, predicted the attributable fractions $\bm{\eta}$, predicted the complete intensity matrix as $m^{s,.}_{a,b}\eta^{s,t}_{a,b}$, and assessed accuracy against the ground truth.
We repeated the experiment $100$ times, each time generating new contact patterns while keeping the participant demographics fixed as in BICS. 
For illustration, we present one instance of the results from the neutral assortativity scenario in~\Cref{fig:full-vs-partial}. 

We found that predicting the complete contact intensity matrices from partial data results in a substantial increase in MAPE for all three assortativity scenarios (neutral: 72.2\% vs.\ 18.1\%, assortative: 80.7\% vs.\ 21.7\%, disassortative: 78.4\% vs.\ 18.5\%). 
However, with the exception of the highly unrealistic disassortative scenario, the posterior 95\% coverage rates were good (neutral: 94.6\%, assortative: 98.5\%, disassortative: 90.4\%) indicating that predictive uncertainty, indicating well calibrated predictive uncertainty (see Tables C.2 to C.4 in \citet{dan_supplement_2026}).

\begin{figure}
    \centering
    \includegraphics[width=\linewidth]{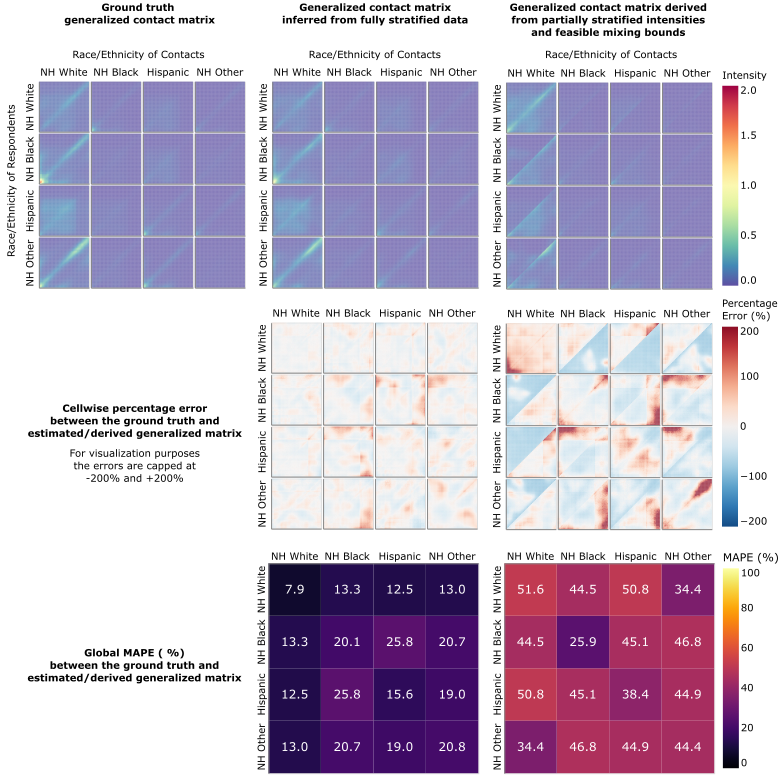}
    \caption{Synthetic case study assessing the prediction quality of the generalized contact matrix when only partial data are available. To capture realism, we used the sample size and demographic composition of respondents in the first wave of the Berkeley Interpersonal Contact Study (BICS). Population sizes were obtained from the American Community Survey. The synthetic ground truth generalized contact matrix were generated as in \Cref{sec:validation-benchmarking}. The 1st row shows the ground-truth generalized contact matrix (left), the posterior mean of the inferred generalized contact matrix when complete data were available (middle), and the posterior mean of the predicted generalized contact matrix derived from partial data. The 2nd row shows the cell-wise percentage error for the two cases calculated as $(\hat{m}^{k,\ell}_{a,b} - m^{k,\ell}_{a,b}) / m^{k,\ell}_{a,b} \times 100$ where $\hat{m}$ denotes the estimate and $m$ the ground-truth. Blue shades indicate underestimation, red shades overestimation. We capped errors at $\pm200\%$ for visualization. The 3rd row shows the global mean average percentage error (MAPE) for each age-age matrix, i.e., the average of the absolute values for all the cells in the matrices above.} 
    \label{fig:full-vs-partial}
\end{figure}
\section{Case Studies}\label{sec:case-studies}

\subsection{Time evolution of generalized contact patterns by age and SES during the COVID-19 pandemic in Germany}

\begin{figure}[!h]
    \centering
    \includegraphics[width=\linewidth]{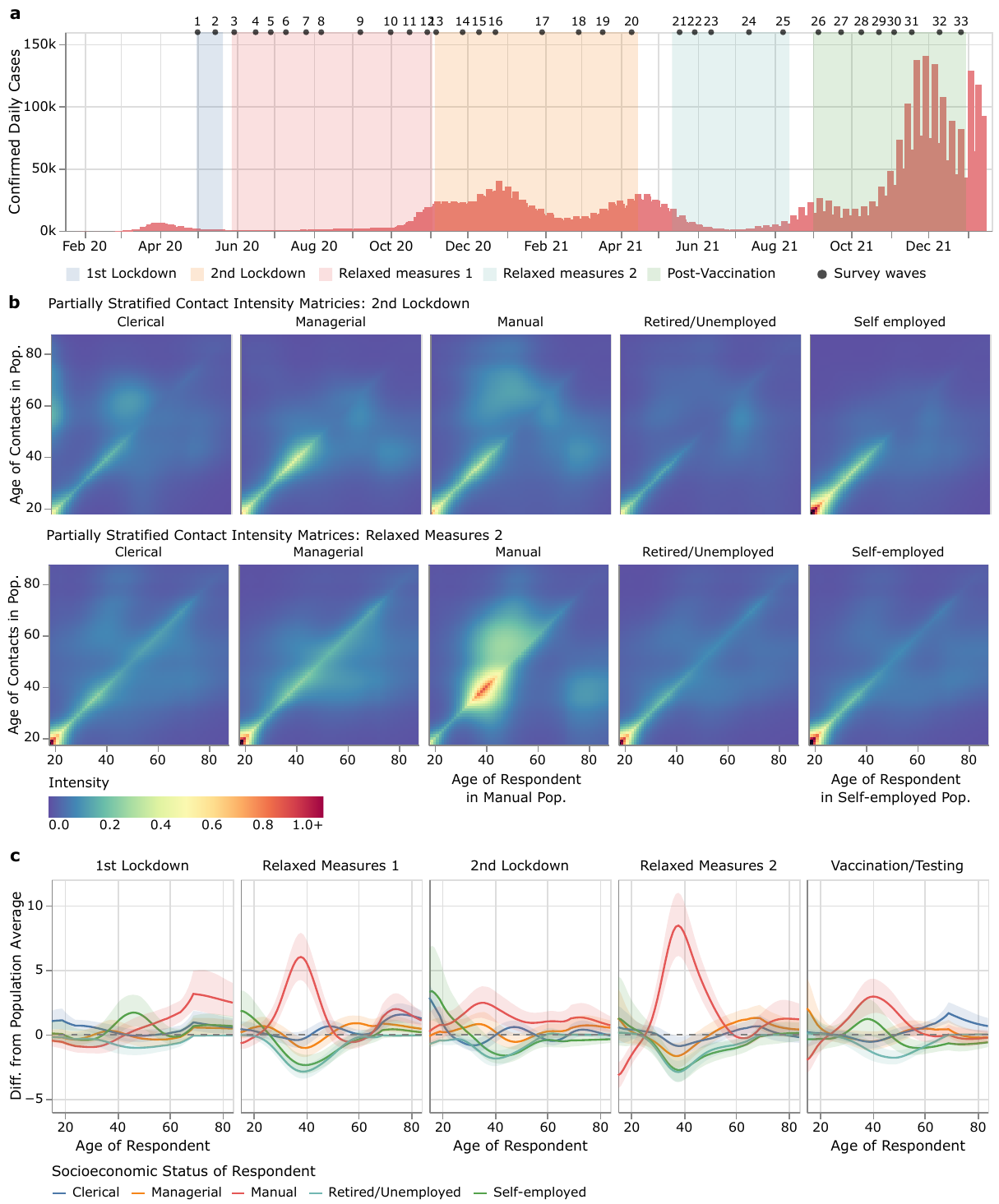}
    \caption{Time evolution of contact intensity patterns during the COVID-19 pandemic in Germany. \textbf{a}: Timeline of COVIMOD survey waves (dots), along with confirmed daily COVID-19 cases (y-axis) in Germany obtained from \citet{dong_interactive_2020}, and the five pandemic phases. 
    \textbf{b}: Posterior means of generalized non-household contact matrices for the second lockdown and relaxed measures 2 phase, by age of respondents and contacts and SES of respondents. \textbf{c}: Deviations in the intensity of non-household contacts from the population-weighted average, stratified by SES of respondents across pandemic phases. Solid lines represent posterior means, and shaded bands represent 95\% CIs.}
    \label{fig:COVIMOD}
\end{figure}

We first explored if the \texttt{g-mix} model uncovers social inequalities in contact patterns similar to a recent study by~\citet{manna_importance_2024} which uses \texttt{socialmixr} methodology. The \texttt{g-mix} model estimates contact intensities at considerably higher age resolution. It does so in a way that allows additive modeling of feature effects across all the data, rather than sub-setting the limited data by features and independently inferring separate contact matrices.
We also wanted to obtain valid generalized contact matrices that meet all reciprocity constraints across all dimensions, and track computational requirements.
For this, we applied \texttt{g-mix} to analyze differences in non-household contact patterns by socioeconomic status (SES) across distinct phases of the COVID-19 pandemic in Germany, using data from the COVIMOD study~\citep{tomori_individual_2021}. 
Part of the European CoMix consortium \citep{wong_social_2023}, the study utilized quota-based sampling to ensure broad representativeness of the German population across age, sex, and regional strata. 
Participants were recruited through the online panel, enrolling 7,851 respondents across 33 waves from April 2020 to December 2021 (\Cref{fig:COVIMOD}a). A contact was defined as direct physical contact or an in-person conversation involving at least a few words.

We fitted the \texttt{g-mix} model in~\eqref{eq:partially-stratified-model} to partial data of non-household contacts stratified by age and SES of the respondents and age of their contacts. Age was stratified by 1-year age bands from age 18 to 85, and SES  by 5 levels based on employment relations: manual, office-based clerical, managerial, self-employed, and retired/unemployed.
\Cref{fig:COVIMOD}b shows the posterior mean of the estimated, partial contact intensities $m^{s,.}_{a,b}$ for the second lockdown phase, and the subsequent second relaxed measures phase. Figures 6 and 7 in the Supplement illustrate the time evolution of these contact intensities across all phases, along with their standard deviations.
We found that during lockdown periods, the overall heterogeneity in contact patterns decreased significantly, dropping to a near-uniform baseline across all groups.
This is further substantiated in
\Cref{fig:COVIMOD}, suggesting that in Germany, stringent non-pharmaceutical interventions temporarily suppressed SES differences in contact patterns, acting as an epidemiological equalizer.

Inequalities in contact patterns by SES re-emerged during phases when non-pharmaceutical interventions were relaxed. In these phases, the manual SES group consistently recorded the highest number of contacts, likely reflecting occupational requirements. In contrast, clerical, retired, unemployed, and self-employed individuals maintained the lowest contact rates, suggesting they had greater agency in limiting their exposure risks (Figure H.1 of \citet{dan_supplement_2026}). These dynamics closely mirror recent empirical findings from other European social contact studies, which demonstrated that higher SES groups exhibited greater behavioral elasticity in response to non-pharmaceutical interventions, whereas manual or essential workers maintained rigid, constraint-driven contact networks~\citep{manna_importance_2024}. Additionally, we show the posterior mean and the width of 95\% CIs for generalized partially stratified contact matrices in Figures H.2 and H.3 of \citet{dan_supplement_2026}.

\subsection{Selecting statistically relevant features in generalized contact matrices for the United States based on ELPD} 
We explore how \texttt{g-mix} allows us to identify and select relevant features in generalized contact matrices in a way that propagates uncertainty arising from limited survey data.
We analyze non-household contact data from the Berkeley Interpersonal Contact Study (BICS), which aimed to quantify changes in close-range contacts during the COVID-19 pandemic in the United States \citep{feehan_quantifying_2021}. 
The study utilized a quota-based sampling approach through an online research panel to collect data rapidly across multiple waves, beginning with a pilot on March 22, 2020 and continuing through wave 6 in May 2021. A contact in this study was defined as a two-way conversational exchange of at least three words occurring in the physical presence of another person. We analyzed wave 1 (April 10 to May 4, 2020), wave 2 (June 17 to 23, 2020), and wave 3 (September 11 to 26, 2020) of BICS. The sample consisted of a total of 8306 (Wave 1 $N=2,627$, Wave 2 $N=2,431$, Wave 3 $N=3,248$) survey respondents. For the respondents, we include their age (1-year age band), and select relevant features from among sex (male, female), race/ethnicity (Hispanics, non-Hispanic whites, non-Hispanic Blacks and non-Hispanic Asian and Others), and education (non-high school graduate, high school graduate, some college, college graduate). For contacts, age and sex were reported, but we only consider age, limiting feature selection to respondents.

We repeatedly fitted the \texttt{g-mix} model in~\eqref{eq:partially-stratified-model} to all combinations of the 3 respondent features sex, education, and race/ethnicity while always including the age of respondents and their non-household contacts.
To select in features, we used leave-one-out ELPD as selection measure~\citep{vehtari_practical_2017}, which evaluates the log posterior predictive density at test points, and therefore promotes predictions that capture the test data with high probability mass.
\Cref{tab:feature-selection-bics} shows that \texttt{g-mix} feature selection consistently identified the model that described generalized contact intensities by age, sex and race/ethnicity of respondents as the best predictive model in all survey waves.
The differences in ELPD of this model and the age-age model or any other models were statistically significant in all waves, strongly suggesting that sex and race/ethnicity contain important information on additional heterogeneity contact patterns beyond age.

Figures G.3 in \citet{dan_supplement_2026} illustrate the deviations in the generalized contact intensities described by age, sex and race/ethnicity over time, relative to the corresponding population weighted averages. 
We observe no significant differences in contact behavior during the initial lockdown (Wave 1).
As non-pharmaceutical interventions were eased by the end of April 2020, substantial differences in generalized contact patterns emerged by sex and race/ethnicity, especially among younger age groups. 
In particular, the ``non-Hispanic Asians and Other'' group consistently had fewer contacts compared to other race/ethnicity groups, whereas Hispanics and non-Hispanic Blacks aligned towards the population average as time progressed (Figures G.3 and G.4 in \citet{dan_supplement_2026}). Additionally, we show the posterior and mean and the width of 95\% CIs of generalized partially stratified contact matrices estimated by the best model in Figures G.5 and G.6 of \citet{dan_supplement_2026}.

\begin{table}[htbp]
    \centering
    \begin{tabular}{c l l c r r}
        \toprule
        BICS & generalized Features & Features & Model & leave-one-out& Diff (SE) \\
        Wave & of respondents & of contacts & rank & ELPD (SE) & \\
        \midrule
        \multirow{8}{*}{Wave 1}
         & Age, Sex, Race/Ethnicity   & Age     & 1 & $-14605\ (179)$ & $0\ (0)$     \\
         & Age, Sex               & Age & 2 & $-14667\ (181)$ & $63\ (15)$   \\
         & Age, Race/Ethnicity              & Age & 3 & $-14711\ (181)$ & $107\ (15)$  \\
         & Age          & Age & 4 & $-14790\ (183)$ & $185\ (22)$  \\
         & Age, Sex, Education        & Age & 5 & $-15291\ (258)$ & $686\ (179)$ \\
         & Age, Sex, Education, Race/Ethnicity & Age & 6 & $-15317\ (263)$ & $713\ (186)$ \\
         & Age, Education              & Age & 7 & $-15331\ (256)$ & $726\ (176)$ \\
         & Age, Education, Race/Ethnicity       & Age & 8 & $-15428\ (261)$ & $823\ (183)$ \\
        \midrule
        \multirow{8}{*}{Wave 2}
         & Age, Sex, Race/Ethnicity        & Age & 1 & $-15199\ (183)$ & $0\ (0)$     \\
         & Age, Sex               & Age & 2 & $-15236\ (184)$ & $37\ (15)$   \\
         & Age, Race/Ethnicity              & Age & 3 & $-15299\ (184)$ & $100\ (15)$  \\
         & Age          & Age & 4 & $-15377\ (185)$ & $178\ (22)$  \\
         & Age, Sex, Education        & Age & 5 & $-15427\ (217)$ & $228\ (111)$ \\
         & Age, Sex, Education, Race/Ethnicity & Age & 6 & $-15474\ (220)$ & $275\ (116)$ \\
         & Age, Education              & Age & 7 & $-15499\ (217)$ & $300\ (109)$ \\
         & Age, Education, Race/Ethnicity       & Age & 8 & $-15602\ (220)$ & $403\ (115)$ \\
        \midrule
        \multirow{8}{*}{Wave 3}
         & Age, Sex, Race/Ethnicity       & Age & 1 & $-18095\ (199)$ & $0\ (0)$      \\
         & Age, Sex              & Age & 2 & $-18185\ (201)$ & $90\ (20)$    \\
         & Age, Race/Ethnicity             & Age & 3 & $-18218\ (202)$ & $123\ (18)$   \\
         & Age         & Age & 4 & $-18372\ (204)$ & $277\ (28)$   \\
         & Age, Sex, Education       & Age & 5 & $-19125\ (325)$ & $1029\ (241)$ \\
         & Age, Education             & Age & 6 & $-19165\ (320)$ & $1070\ (236)$ \\
         & Age, Sex, Education, Race/Ethnicity & Age & 7 & $-19176\ (336)$ & $1080\ (256)$ \\
         & Age, Education, Race/Ethnicity       & Age & 8 & $-19256\ (333)$ & $1161\ (252)$ \\
        \bottomrule
    \end{tabular}
    \vspace{0.1cm}
    \caption{Selection of features in the generalized contact matrix that best explain the  BICS survey data in terms of leave-one-out Expected Log Posterior Density (ELPD). For each wave, \texttt{g-mix} models with different candidate features for the respondents were fitted to the partial data. For illustration purposes, the contacts of respondents were only characterized by 1-year age bands. We report the model rank, numerically estimated ELPD with its standard error in parentheses, and the numerically estimated difference in ELPD relative to the best model (Diff) with its standard error in parentheses. Models were ranked within each wave from best (rank 1) to worst.}
    \label{tab:feature-selection-bics}
\end{table}
\section{Discussion and Conclusion}

In this work, we propose a Bayesian modeling framework for inferring generalized contact matrices from data collected in social contact surveys. 
%
%
%
The intention of our work is to allow epidemiologists to move beyond traditional age-only social contact matrices and capture additional heterogeneity, for example by race/ethnicity and socioeconomic status. 
%


Our framework has key advantages, and several limitations. Perhaps most importantly, we are able to infer generalized contact matrices that satisfy all inherent constraints using existing survey data with sample sizes in the range of $1,500-3,000$ respondents. Our model-based approach also provides a principled approach for feature modeling and selection through multiplicative and interaction effects. 
\texttt{g-mix} also provides a framework to predict  complete generalized contact matrices from partially observed data, ingest external data, and propagate uncertainty in these predictions, also harnessing feasible mixing bounds that are based on earlier literature on contingency tables.
In terms of limitations, our proposed framework does not incorporate survey sampling weights, and ignoring these could introduce notable bias into our contact intensity estimates. 
Furthermore, we assumed that the reported age of contacts is provided and accurate to a 1-year resolution. In reality, respondents may struggle to accurately recall exact ages, prompting recent large-scale surveys to collect data using coarse age ranges \citep{tomori_individual_2021, gimma_changes_2022}. 
Empirical contact data are also susceptible to recall bias, reporting fatigue, and outliers, phenomena that must be carefully accounted for during the modeling process \citep{dan_addressing_2025}. 
%
%
We also exclusively focused on characterizing average contact patterns, while infections may be driven by a small group of individuals with very large number of contacts \citep{ling_contact_2026}.
Finally, although the feasible mixing bounds limit the space of prediction possibilities, we found it is challenging to accurately predict complete generalized contact matrices from partial data. The tightness of these bounds is highly dependent on both the local demographic structure and the true underlying contact patterns. 
Future work could leverage auxiliary data streams such as mobile phone telemetry \citep{rudiger_predicting_2021}, to further narrow the feasible mixing limits, or provide external data for predicting generalized contact matrices. 
Moreover, future data collection strategies could benefit from modular designs, where respondents are asked to report separate information on the contacts (e.g., asking one group to report contact sex and another to report contact socioeconomic status), to systematically reduce respondent burden and preserve privacy. 
Here, the posterior estimates of the established bounds could be particularly useful to actively learn the optimal allocation of survey modules and adjust these appropriately over time, based on established contingency table theory \citep{dobra_bounds_2000, cacciarelli2024active}.

In conclusion, infectious diseases often spread along structural and socioeconomic fault lines within populations. There is a clear need for methodologies that can capture social contact patterns along dimensions such as deprivation, race/ethnicity or socioeconomic status. We hope our \texttt{g-mix} framework, alongside its open-source implementation, will facilitate the estimation of generalized contact matrices, which is a cornerstone to monitoring the impact of non-pharmaceutical interventions, estimation of $R_0$, and vaccine allocation.

\begin{acks}[Acknowledgments]
We thank members of the Machine Learning \& Global Health network (MLGH) (\url{https://mlgh.net}) for their feedback throughout this project, and the London Mathematical Society for funding of MLGH meetings during which this study was discussed and refined.
We are grateful to Dr. Dennis M. Feehan and Dr. Ayesha Mahmud at the University of California Berkeley for sharing the BICS data, and to Dr. Andr\'e Karch and Dr. Veronika Jaeger at the University of Münster for sharing the COVIMOD data. 
We acknowledge the Research Computing Service at Imperial College London (DOI: 10.14469/hpc/2232) for providing the high-performance computing resources used in part of this work, and Zulip for sponsoring team communications through the Zulip Cloud Standard chat app. 
\end{acks}
\begin{funding}
SD is funded by a doctoral studentship to the EPSRC Center for Doctoral Training in Modern Statistics and Statistical Machine Learning at Imperial and Oxford (EP/S023151/1 to Prof. Axel Gandy). OR acknowledges funding from the UK Health Security Agency, and the Moderna Charitable Foundation. ZL is funded by the National Medical Research Council, Singapore Ministry of Health (PREPARE-S2-2024-002).
\end{funding}

\putbib[definitions,references]
\end{bibunit}

\clearpage

\begin{bibunit}

\renewcommand\thesection{\Alph{section}}
\setcounter{section}{0}
\setcounter{equation}{0}
\numberwithin{figure}{section}
\numberwithin{table}{section}

\section{Derivations and Proofs} \label{sec:appendix-proofs}


\begin{table}[htpb!]
    \centering
    \renewcommand{\arraystretch}{1.3}
    \begin{tabular}{@{}lp{0.8\textwidth}@{}}
        \toprule
        \textbf{Notation} & \textbf{Description} \\
        \midrule
        
        \multicolumn{2}{@{}l}{\textbf{Sets and Indices}} \\
        $\mcA$ & Set of all age groups. \\
        $a, b$ & Indices representing specific age groups ($a, b \in \mcA$). \\
        $\mcX$ & Set of categories for a specific stratification variable. \\
        $k, \ell, t$ & Indices representing specific categories or strata ($k, \ell, t \in \mcX$). \\
        $\mcS$ & Set of composite strata across all stratification variables ($\mcS = \mcX_1 \times \cdots \times \mcX_J$). \\
        \addlinespace
        
        \multicolumn{2}{@{}l}{\textbf{Population Quantities}} \\
        $P_a$ & Total population size for individuals of age $a$. \\
        $P^k_a$ & Population size for individuals of age $a$ within stratum $k$. \\
        \addlinespace
        
        \multicolumn{2}{@{}l}{\textbf{Contact Parameters}} \\
        $\gamma_{a,b}$ & Baseline age-specific contact rate between age $a$ and age $b$. \\
        $\delta^{k,\ell}_{a,b}$ & Full contact rate modifier for contacts from stratum $k$, age $a$ to stratum $\ell$, age $b$. \\
        $\delta^{k,.}_{a,b}$ & Partial contact rate modifier for stratum $k$, age $a$ to age $b$. \\
        $z^{k,\ell}_{a,b}$ & Shorthand for $\EE[Z^{k,\ell}_{a,b}]$, the expected total number of contacts from stratum $k$, age $a$ to individuals in stratum $\ell$, age $b$. \\
        $z^{k,.}_{a,b}$ & Shorthand for $\EE[Z^{k,.}_{a,b}]$, the expected total number of contacts from stratum $k$, age $a$ to individuals of age $b$. \\
        $m^{k,\ell}_{a,b}$ & Fully stratified contact intensity from stratum $k$, age $a$ to stratum $\ell$, age $b$. \\
        $m^{k,.}_{a,b}$ & Partially stratified contact intensity from stratum $k$, age $a$ to individuals of age $b$. \\
        $\eta^{k,\ell}_{a,b}$ & Attributable fraction: the proportion of contacts made by an individual in stratum $k$, age $a$ that are with individuals in stratum $\ell$, age $b$. \\
        \addlinespace
        
        \multicolumn{2}{@{}l}{\textbf{Latent Variables and Operators}} \\
        $\bOmega^{(j)}$ & Tensor of latent two-dimensional surfaces for the $j$-th stratification variable. \\
        $\bOmega$ & Global tensor of latent surfaces derived via Kronecker sums. \\
        $\Pi, \bP$ & Permutation matrices used to mathematically enforce reciprocity constraints. \\
        $\phi(\cdot)$ & The softmax transformation function. \\
        
        \bottomrule
    \end{tabular}
    \vspace{0.2cm}
    \caption{Summary of mathematical notations used throughout the derivations and proofs.}
    \label{tab:notations}
\end{table}

\begin{remark}
Throughout the derivations and proofs below we use indices $k, \ell \in \mcX$ for the single-feature setting, which is the setting presented in Section 2.2--2.4 of the main manuscript. All results generalize directly to the multi-feature setting by replacing $k, \ell$ with composite strata $s, t \in \mcS = \mcX_1 \times \cdots \times \mcX_J$ as introduced in Section 2.5 of the main manuscript.
\end{remark}

\subsection{Derivation of Constraints 1 and 2}
\label{app:full-cnt-rate-mod-proof}
In closed populations and by definition of our model,
\begin{equation*}
   z^{k,\ell}_{a,b} = \gamma_{a,b} \delta^{k,\ell}_{a,b} P^{k}_a P^{\ell}_b = \gamma_{b,a} \delta^{\ell,k}_{b,a} P^{\ell}_b P^{k}_a = z^{\ell,k}_{b,a}, \quad \forall a,b \in \mcA,\ \forall k,\ell \in \mcX.
\end{equation*}
Since $\gamma_{a,b} = \gamma_{b,a}$, it follows that $\delta^{k,\ell}_{a,b} = \delta^{\ell,k}_{b,a}$.
Also, $z_{a,b} = \sum_{k,\ell} z^{k,\ell}_{a,b}$, hence
\begin{equation*}
    \gamma_{a,b}P_aP_b = \gamma_{a,b} \sum_{k,\ell} \delta^{k,\ell}_{a,b} P^{k}_a P^{\ell}_b.
\end{equation*}
Dividing both sides by $\gamma_{a,b}P_aP_b$ yields
\begin{equation}\label{app:sum-to-one}
    \sum_{k,\ell} \delta^{k,\ell}_{a,b} \frac{P^{k}_a}{P_a} \frac{P^{\ell}_b}{P_b} = 1 \quad \forall a,b\in\mathcal{A}.
\end{equation}

\subsection{Derivation of Constraints 3 and 4}
\label{app:partial-cnt-rate-mod-proof} In closed populations and by definition of our model,
\begin{equation*}
   z^{k,.}_{a,b} = \gamma_{a,b} \delta^{k,.}_{a,b} P^{k}_a P_b = \gamma_{b,a} \delta^{.,k}_{b,a} P_b P^{k}_a = z^{.,k}_{b,a}, \quad \forall a,b \in \mcA,\ \forall k \in \mcX.
\end{equation*}
Since $\gamma_{a,b} = \gamma_{b,a}$, it follows that $\delta^{k,.}_{a,b} = \delta^{.,k}_{b,a}$. Replace $\delta^{k,\ell}_{a,b}$ with $\delta^{k,.}_{a,b}$ and $P^\ell_b$ with $P_b$ in \eqref{app:sum-to-one} to obtain
\begin{equation*}
    \sum_{k} \delta^{k,.}_{a,b} \frac{P^{k}_a}{P_a} = 1 \quad \forall a,b\in\mathcal{A}.
\end{equation*}

\subsection{Proof of Proposition 1}
\begin{proof}
For case (1): When $\bS$ has full column rank, the unique solution is $\bdelta = (\bS^\top \bS)^{-1} \bS^\top \bm{1}_{A^2}$. Since $\bS \bm{1}_{K^2} = \bm{1}_{A^2}$ and the solution is unique, we have $\bdelta = \bm{1}_{K^2}$.

For case (2): When $\operatorname{rank}(\bS) = r < K^2$, its null space has dimension $K^2 - r > 0$. Since $\bm{1}_{A^2} \in \text{Col}(\bS)$, solutions exist. The general solution is $\bdelta = \bm{1}_{K^2} + \bm{n}$ where $\bm{n} \in \text{Null}(\bS)$.
\end{proof}

\subsection{Proof of Proposition 2}
\label{app:frechet-hoeffding-bounds}
\newcommand{\barZ}[2]{\bar{Z}^{#1}_{#2}}
We provide two versions of the proof. The first is easy to follow but does not prove that the bounds are sharp. The second establishes the theoretical link with contingency tables and decomposable graphs.

We first restate the result for convenience. Let $\eta^{k,\ell}_{a,b} := m^{k,\ell}_{a,b}/m^{k,.}_{a,b} \in [0,1]$ denote the attributable fraction for stratum $(k,a)$ and stratum $(\ell,b)$. For any $k,\ell,t \in \mcX$ and $a,b \in \mcA$, the following inequality is a \emph{sharp} (tightest) bound for $\eta^{k,\ell}_{a,b}$:
\begin{equation}
   \max\qty(
        0,
        \frac{z^{\ell,.}_{b,a}}{z^{k,.}_{a,b}} - \sum_{t \neq k} \frac{z^{t,.}_{a,b}}{z^{k,.}_{a,b}}
    ) \leq \eta^{k,\ell}_{a,b} \leq
    \min\qty(
        1,
        \frac{z^{\ell,.}_{b,a}}{z^{k,.}_{a,b}}
    ).
\end{equation}

\begin{proof}[Proof 1]
In a closed population, expected total contact counts satisfy:
\begin{enumerate}
    \item Non-negativity: $z^{k,\ell}_{a,b} \geq 0$
    \item Marginal Consistency: $z^{k,.}_{a,b} = \sum_\ell z^{k,\ell}_{a,b}$
    \item Reciprocity: $z^{k,\ell}_{a,b} = z^{\ell,k}_{b,a}$
\end{enumerate}

A upper bound on $z^{k,\ell}_{a,b}$ arises from the fact that the number of contacts between strata $(k,a)$ and $(\ell,b)$ cannot exceed the total contacts generated by either strata. This is expressed by the inequalities $z^{k,\ell}_{a,b} \leq z^{k,.}_{a,b}$ and $z^{k,\ell}_{a,b}=z^{\ell,k}_{b,a} \leq z^{\ell,.}_{b,a}$. Combining the two inequalities lead to
\begin{equation}
    z^{k,\ell}_{a,b} \leq \min\left(z^{k,.}_{a,b},\, z^{\ell,.}_{b,a}\right).
\end{equation}

For the lower bound we start with the marginal consistency for stratum $(\ell, b)$ with age $a$:
\begin{equation*}
    z^{\ell,.}_{b,a} = z^{\ell,k}_{b,a} + \sum_{t \neq k} z^{\ell,t}_{b,a}.
\end{equation*}
Rearranging for $z^{k,\ell}_{a,b}$ gives us
\begin{equation*}
    z^{k,\ell}_{a,b} = z^{\ell,.}_{b,a} - \sum_{t \neq k} z^{\ell,t}_{b,a}.
\end{equation*}
The lower bound is obtained using the fact that $z^{\ell,t}_{b,a}$ is bounded above by $z^{t,.}_{a,b}$ since $z^{\ell,t}_{b,a} = z^{t,\ell}_{a,b} \leq z^{t,.}_{a,b}$. Substituting this into the summation above gives us
\begin{equation*}
    z^{k,\ell}_{a,b} \geq \max\left(0, z^{\ell,.}_{b,a} - \sum_{t \neq k} z^{t,.}_{a,b}\right)
\end{equation*}
where we take the maximum with $0$ because $z^{k,\ell}_{a,b}$ must be non-negative. Dividing both sides by $z^{k,.}_{a,b}$ gives the desired inequality in terms of $\eta^{k,\ell}_{a,b}$.
\end{proof}

\begin{proof}[Proof 2]
Consider $k,\ell \in \mcX$ and $(a,b) \in \mcA \times \mcA$ to be vertices in an undirected graph. The edge set consist of the edge connecting $k$ to $(a,b)$ and the edge connecting $(a,b)$ to $\ell$. This graph contains two cliques $\{k, (a,b)\}$ and $\{(a,b),\ell\}$ and is \emph{decomposable} since removing the vertex $(a,b)$ would split the graph in two. Applying Theorem 4 of \citet{dobra_bounds_2000} gives us the inequality
\begin{equation*}
    \max\qty(0,\,z^{k,.}_{a,b} + z^{.,\ell}_{a,b} - \sum_{t \in \mcX} z^{t,.}_{a,b}) \leq z^{k,\ell}_{a,b} \leq \min\qty(z^{k,.}_{a,b},\, z^{.,\ell}_{a,b}).
\end{equation*}
Applying the reciprocity principle to substitute $z^{.,\ell}_{a,b} = z^{\ell,.}_{b,a}$ and dividing both sides by $z^{k,.}_{a,b}$ gives us the inequality in terms of $\eta^{k,\ell}_{a,b}$.
\end{proof}

\subsection{Proof that the Softmax Transform is Permutation Equivariant}
\label{app:softmax-perm-equiv-proof}
For $\bx \in \RR^{D}$ the softmax function is defined as
$$
\phi(\bx)_i = \frac{\exp(x_i)}{\sum^D_{j=1} \exp(x_j)}, \quad \text{for } i=1,\ldots,D.
$$
Let $\pi$ be a permutation of $\{1,\ldots,D\}$.
$$
\pi(x_i) := x_{\pi(i)}.
$$
Equivalently , we can express this permutation in the form a matrix $\Pi \in \RR^{D \times D}$ such that
$$
\qty(\Pi \bx)_i = x_{\pi(i)}
$$
We want to show that the softmax is equivariant to permutations, i.e.,
$$
\phi(\Pi \bx) = \Pi\phi(\bx) \quad \text{for all } \bx \in \RR^D.
$$

\begin{proof}
By definition,
$$
\phi(\Pi \bx)_i = \frac{
    \exp(\pi(x_i))
}{
    \sum^D_{j=1} \exp(\pi (x_j))
} = \frac{
    \exp(x_{\pi(i)})
}{
    \sum^D_{j=1} \exp(x_{\pi(j)})
} = \frac{
    \exp(x_{\pi(i)})
}{
    \sum^D_{j=1} \exp(x_j)
}.
$$
where the last equality comes from the fact that the sum of exponential remains the same regardless of ordering. This is equivalent to the $i$-th element of $\Pi \phi(\bx)$:
$$
\qty(\Pi \phi(\bx))_i = \pi \qty( \frac{\exp(x_i)}{\sum^D_{j=1} \exp(x_j)} ) = \frac{\pi(\exp(x_i))}{\sum^D_{j=1} \exp(x_j)} = \frac{\exp(x_{\pi(i)})}{\sum^D_{j=1} \exp(x_j)}.
$$
Hence, the forward transform is equivariant to an arbitrary permutation of the elements.
\end{proof}

\subsection{Proof of Proposition 3}
\label{app:kronecker-symm-proof}

For each $j=1,\ldots,J$, let
\begin{equation*}
    \bOmega^{(j)} \in \RR^{K^2_j \times A \times A}
\end{equation*}
denote the tensor corresponding to the latent 2-dimensional surfaces for the $j$-th stratification variable. Assume that, for each $j$, there exists a permutation
\begin{equation*}
    \pi_j : \qty{1,\ldots,K^2_j} \rightarrow \qty{1,\ldots,K^2_j}
\end{equation*}
such that the following symmetry holds:
\begin{equation*}
    (\bOmega^{(j)})_{k,a,b} = (\bOmega^{(j)})_{\pi_j(k),b,a} \quad \text{for all }k,a,b.
\end{equation*}
Here, $\pi_j(k)$ corresponds to the transpose operation on a $K_j \times K_j$ matrix (viewed as a rearrangement of the first-mode index).

Define the global tensor
\begin{equation*}
    \bOmega = \bOmega^{(1)} \oplus_1 \cdots \oplus_1 \bOmega^{(J)},
\end{equation*}
where $\oplus_1$ is the mode-1 Kronecker sum (i.e., the Kronecker product acting on the first index of the tensors). Let $K_* = \prod^J_{j=1} K_j$, so that the first mode of $\bOmega$ has dimension $K^2_*$. Define a permutation
\begin{equation*}
    \pi : \qty{1,\ldots,K^2_*} \rightarrow \qty{1,\ldots,K^2_*}
\end{equation*}
corresponding to the the transpose on a $K_* \times K_*$ matrix, induced in the natural way by the $\pi_j$)s and the Kronecker structure.

We would like to show that $\bOmega$ inherits the same symmetry:
\begin{equation*}
    \bOmega_{k,a,b} = \bOmega_{\pi(k),b,a} \quad \text{for all } k, a, b.
\end{equation*}

\begin{proof}
For $j \in \qty{1,\ldots,J}$. For each pair of indices $(a,b)$, define the vector 
\begin{equation*}
    \bomega^{(j)}_{a,b} \in \RR^{K^2_j}, \quad [\bomega^{(j)}_{k,a,b}] := \bOmega^{(j)}_{k,a,b}, \quad k=1,\ldots,K^2_j,
\end{equation*}
and the diagonal matrix
\begin{equation*}
    \bD^{(j)}_{a,b} := \text{diag}(\bomega^{(j)}_{a,b}) \in \RR^{K^2_j \times K^2_j}.
\end{equation*}
Let $\bP_j$ be the permutation matrix corresponding to $\pi_j$; that is,
\begin{equation*}
    (\bP_j \bx)_k = \bx_{\pi_j(k)} \quad \text{for all } \bx \in \RR^{K^2_j}, \ k=1,\ldots,K^2_j.
\end{equation*}
The assumed symmetry
\begin{equation*}
    \bOmega^{(j)}_{k,a,b} = \bOmega^{(j)}_{\pi_j(k),a,b} \quad \text{for all } k,a,b
\end{equation*}
is equivalent to
\begin{equation*}
    [\bomega^{(j)}_{a,b}]_k = \bOmega^{(j)}_{k,a,b} = \bOmega^{(j)}_{\pi_j(k),a,b} = [\bomega^{(j)}_{a,b}]_{\pi_j(k)} = (P_j \bomega^{(j)}_{b,a})_k.
\end{equation*}
Hence
\begin{equation}
    \bomega^{(j)}_{a,b} = \bP_j \bomega^{(j)}_{b,a} \quad \text{for all } a,b.
\end{equation}
Because $\bD^{(j)}_{a,b}$ is diagonal with these entries, we also get the matrix form
\begin{equation*}
    \bD^{(j)}_{a,b} = \text{diag}(\bomega^{(j)}_{a,b}) = \text{diag}(\bP_j \bomega^{(j)}_{b,a}) = \bP_j \text{diag}(\bomega^{(j)}_{b,a}) \bP^\top_j  = \bP_j \bD^{(j)}_{b,a} \bP^\top_j.
\end{equation*}
So
\begin{equation}
    \label{eq:local-symmetry}
    \bD^{(j)}_{a,b} = \bP_j \bD^{(j)}_{b,a} \bP^\top_j \quad \text{for all }a,b.
\end{equation}
Let $K^2_* = \prod^J_{j=1} K^2_j$. For each $(a,b)$, we define a diagonal matrix
\begin{equation*}
    \bD_{a,b} := \bD^{(1)}_{a,b} \oplus \cdots \oplus \bD^{(J)}_{a,b}.
\end{equation*}
Concretely, for $J=2$,
\begin{equation*}
    \bD^{(1)}_{a,b} \oplus \bD^{(2)}_{a,b} := \bD^{(1)}_{a,b} \otimes \bI_{K^2_2} + \bI_{K^2_1} \otimes \bD^{(2)}_{a,b},
\end{equation*}
for general $J$, we iterate this construction. The diagonal entries of $\bD_{a,b}$ define a vector $\bomega \in \RR^{K^2_*}$ by $\bD_{a,b} = \text{diag}(\bomega_{a,b})$. By definition of the global tensor $\bOmega$ via the mode-1 Kronecker sum of $\bOmega^{(j)}$ along the first mode, this $\bomega_{a,b}$ is exactly the global first-mode fibre:
\begin{equation}
    \label{eq:global-mode1-fibre}
    \bOmega_{k,a,b} = [\bomega_{a,b}]_k \quad k=1,\ldots,K^2_*.
\end{equation}
Define
$$
\bP := \bP_1 \otimes \cdots \otimes \bP_J \in \RR^{K^2_* \times K^2_*}.
$$
This is a gain a permutation matrix, corresponding to a permutation $\pi: \{1,\ldots,K^2_*\} \rightarrow \{1,\ldots,K^2_*\}$ such that $(\bP \bx)_k = \bx_{\pi(k)}$ for all $\bx \in \RR^{K^2_*}$, $k=1,\ldots,K^2_*$. To describe $\pi$ explicitly, we introduce the canonical bijection
$$
    \rho: \{1,\ldots,K^2_*\} \rightarrow \prod^J_{j=1} \{1,\ldots,K^2_j\}, \quad k \mapsto (k_1, \ldots, k_J).
$$
Then, for basis vectors,
$$
\bP (\be_{k_1} \otimes \cdots \otimes \be_{k_J}) = (\bP_1 \be_{k_1}) \otimes \cdots \otimes (\bP_J \be_{k_J}) = \be_{\pi_1(k_1)} \otimes \cdots \otimes \be_{\pi_J(k_J)}.
$$
Therefore,
$$
\rho(\pi(k)) = (\pi_1(k_1), \ldots, \pi_J(k_J)) \quad \text{whenever } \rho(k) = (k_1, \ldots, k_J),
$$
so $\pi$ acts component-wise by the $\pi_j$'s. Under the natural identification of the first-mode index with a pair of global stratification indices, this is exactly the transpose of the global $K_* \times K_*$ matrix.

We now show that
\begin{equation}
    \label{eq:global-symmetry-diag}
    \bD_{a,b} = \bP \bD_{b,a} \bP^\top \quad \text{for all } a,b.
\end{equation}
By definition of the Kronecker sum,
$$
\bD_{b,a} = \sum^J_{j=1}\qty(\bI_{K^2_1} \otimes \cdots \otimes \bI_{K^2_{j-1}} \otimes \bD^{(j)}_{b,a} \otimes \bI_{K^2_{j+1}} \otimes \cdots \otimes \bI_{K^2_J}). 
$$
Using the standard Kronecker conjugation property, we obtain
\begin{align*}
\bP \bD_{b,a} \bP^\top &= \sum^J_{j=1} \qty(\bP_1 \otimes \cdots \otimes \bP_J) \qty(\bI \otimes \cdots \otimes \bD^{(j)}_{b,a} \otimes \cdots \otimes \bI) \qty(\bP_1 \otimes \cdots \otimes \bP_J)^\top \\
&= \qty(\bI_{K^2_1} \otimes \cdots \otimes \bI_{K^2_{j-1}} \otimes \qty(\bP_j \bD^{(j)}_{b,a} \bP^\top_j) \otimes \bI_{K^2_{j+1}} \otimes \cdots \otimes \bI_{K^2_J}).
\end{align*}
Using the local symmetry \eqref{eq:local-symmetry},
$$
\bP \bD_{b,a} \bP^\top = \sum^J_{j=1} \qty(\bI \otimes \cdots \otimes \bD^{(j)}_{a,b} \otimes \cdots \otimes \bI) = \bD_{a,b}.
$$
Thus \eqref{eq:global-symmetry-diag} holds. Since $\bD_{a,b}$ and $\bD_{b,a}$ are diagonal, conjugation by $\bP$ simply permutes their diagonal entries. From \eqref{eq:global-symmetry-diag} it follows that
\begin{equation}
    \label{eq:global-symmetry-fibre}
    \bomega_{a,b} = \bP \bomega_{b,a} \quad \text{for all } a,b.
\end{equation}
Finally, unpack \eqref{eq:global-symmetry-fibre} in components, recalling \eqref{eq:global-mode1-fibre}:
$$
\bOmega_{k,a,b} = [\bomega_{a,b}]_k = (\bP \bomega_{b,a})_k = [\bomega_{b,a}]_{\pi(k)} = \bOmega_{\pi(k),b,a}.
$$
So we obtain the desired global symmetry:
$$
\bOmega_{k,a,b} = \bOmega_{\pi(k),b,a} \quad \text{for all } k,a,b.
$$
\end{proof}

\subsection{Proof: Attributable Fractions Derived via Reciprocity Satisfy Mixing Bounds}
\label{app:reciprocal-attributable-fraction-bounds-proof}

Suppose the attributable fraction $\eta^{k,\ell}_{a,b}$ satisfies the feasible mixing limits
\begin{equation}\label{app:feasible-mixing-bounds}
\max \qty(0, \frac{z^{\ell,.}_{b,a}}{z^{k,.}_{a,b}} - \sum_{t\neq k} \frac{z^{t,.}_{a,b}}{z^{k,.}_{a,b}}) \leq \eta^{k,\ell}_{a,b} \leq \min \qty(1, \frac{z^{\ell,.}_{b,a}}{z^{k,.}_{a,b}}).
\end{equation}
We wish to show that the reciprocal attributable fraction $\eta^{\ell,k}_{b,a}$ derived deterministically via the following reciprocity relation also satisfy its feasible mixing limits:
\begin{equation*}
    \eta^{\ell,k}_{b,a} = \frac{z^{k,.}_{a,b}}{z^{\ell,.}_{b,a}} \eta^{k,\ell}_{a,b}.
\end{equation*}
\begin{proof}
We show that the lower and upper bounds of $z^{k,.}_{a,b}/z^{\ell,.}_{b,a} \times \eta^{k,\ell}_{a,b}$ corresponds exactly with the lower and upper bounds for $\eta^{\ell,k}_{b,a}$. Multiply $z^{k,.}_{a,b}/z^{\ell,.}_{b,a}$ to the left side of the inequality \eqref{app:feasible-mixing-bounds}:
\begin{align*}
    \frac{z^{k,.}_{a,b}}{z^{\ell,.}_{b,a}} \max \qty(0, \frac{z^{\ell,.}_{b,a}}{z^{k,.}_{a,b}} - \sum_{t\neq k} \frac{z^{t,.}_{a,b}}{z^{k,.}_{a,b}}) &= \max \qty(0, 1 + \frac{z^{k,.}_{a,b}}{z^{\ell,.}_{b,a}}
    \qty{
       \frac{z^{k,.}_{a,b}}{z^{k,.}_{a,b}} - \frac{z_{a,b}}{z^{k,.}_{a,b}}
    }  ) \\
    &= \max \qty(0, 1 +
       \frac{z^{k,.}_{a,b}}{z^{\ell,.}_{b,a}} - \frac{z_{b,a}}{z^{\ell,.}_{b,a}}
    ) \\
    &= \max\qty(0, \frac{z^{k,.}_{a,b}}{z^{\ell,.}_{b,a}} - \sum_{t\neq\ell} \frac{z^{t,.}_{b,a}}{z^{\ell,.}_{b,a}})
\end{align*}
which corresponds to the lower bound of $\eta^{\ell,k}_{b,a}$. Similarly, multiplying $z^{k,.}_{a,b}/z^{\ell,.}_{b,a}$ to the right hand side of the inequality \eqref{app:feasible-mixing-bounds} yields
\begin{equation*}
    \frac{z^{k,.}_{a,b}}{z^{\ell,.}_{b,a}} \min \qty(1, \frac{z^{\ell,.}_{b,a}}{z^{k,.}_{a,b}}) = \min \qty(\frac{z^{k,.}_{a,b}}{z^{\ell,.}_{b,a}}, 1)
\end{equation*}
which corresponds to the upper bound of $\eta^{\ell,k}_{b,a}$.
\end{proof}
\section{Sampling Attributable Fractions from Truncated Dirichlet Distributions}
To sample attributable fractions from truncated Dirichlet distributions, we implement the conditional sampling method described by \citet{ng_truncated_2011}. Here, we detail the specifics of the sampler using our notation. We also provide an efficient NumPy implementation of this procedure in our Python package, \texttt{cntmosaic}.

\subsection{Sampling from a Truncated Beta Distribution}
Sampling from a truncated beta distribution is a fundamental component of the truncated Dirichlet sampler. The truncated beta distribution has the density function
\begin{equation*}
    \text{TruncBeta}(x \mid \alpha_1, \alpha_2; l, u) = c^{-1}_{\text{TB}} x^{\alpha_1-1} (1-x)^{\alpha_2-1},
\end{equation*}
where $x \in [l,u] \subseteq [0,1]$, and the normalising constant is given by
\begin{equation*}
c_{\text{TB}} = \int^u_l x^{\alpha_1-1} (1-x)^{\alpha_2-1} dx.
\end{equation*}
The inverse transform method can be applied to generate random variables from the truncated beta distribution. Let $F_B$ denote the cumulative distribution function (CDF) of the standard beta distribution, and let $F^{-1}_B$ be its inverse CDF. Furthermore, let $F_B(l)$ and $F_B(u)$ be the values of the CDF evaluated at the lower and upper bounds, respectively. If we draw a standard uniform random variable $W \sim \text{Unif}(0,1)$, a random variable $X$ following the truncated beta distribution can be obtained as:
\begin{equation*}
    X = F^{-1}_B \qty(F_B(l) + W\{F_B(u) - F_B(l)\}).
\end{equation*}

\subsection{Conditional Sampling from a Truncated Dirichlet Distribution}
Let $\{\eta^{k,\ell}_{a,b}\}^K_{\ell=1} \sim \text{TruncDirichlet}(\alpha \{P^\ell_b/P_b\}^K_{\ell=1}; \bm{l}^k_{a,b}, \bm{u}^k_{a,b})$, and let $W_1, \ldots, W_{K-1} \sim \text{Unif}(0,1)$, where $\bm{l}^k_{a,b}$ and $\bm{u}^k_{a,b}$ denote the lower and upper bounds of the feasible mixing limits at coordinate $(k,a,b)$, respectively. Theorem 7.3 of \citet{ng_truncated_2011} provides the following stochastic representation for the fractions $\{\eta^{k,\ell}_{a,b}\}^K_{\ell=1}$:
\begin{align*}
    \eta^{k,K-1}_{a,b} &\overset{d}{=} F^{-1}_{K-1}(W_{K-1}), \\
    \eta^{k,\ell}_{a,b} &\overset{d}{=} \qty(1 - \sum^{K-1}_{j=\ell+1} \eta^{k,j}_{a,b}) F^{-1}_\ell(W_\ell), \quad \text{for } \ell=K-2,\ldots,1, \\
    \eta^{k,K}_{a,b} &= 1 - \sum^{K-1}_{\ell=1} \eta^{k,\ell}_{a,b},
\end{align*}
where $F^{-1}_\ell$ denotes the inverse CDF of the truncated beta distribution
\begin{equation*}
    \text{TruncBeta}\qty(\alpha \frac{P^\ell_b}{P_b}, \alpha \qty( 1 - \sum^{K-1}_{j=\ell} \frac{P^j_b}{P_b} ); \tilde{l}^{k,\ell}_{a,b}, \tilde{u}^{k,\ell}_{a,b} ),
\end{equation*}
with the dynamically adjusted bounds defined as:
\begin{align*}
\tilde{l}^{k,\ell}_{a,b} &= \max\qty(
    \frac{l^{k,\ell}_{a,b}}{1 - \sum^{K-1}_{j=\ell+1}\eta^{k,j}_{a,b}},
    1 - \frac{u^{k,.}_{a,b} - \sum^{K-1}_{j=\ell} u^{k,j}_{a,b}}{1 - \sum^{K-1}_{j=\ell+1}\eta^{k,j}_{a,b}}
), \\
\tilde{u}^{k,\ell}_{a,b} &= \max \qty(
    \frac{u^{k,\ell}_{a,b}}{1 - \sum^{K-1}_{j=\ell+1}\eta^{k,j}_{a,b}},
    1 - \frac{l^{k,.}_{a,b} - \sum^{K-1}_{j=\ell} l^{k,j}_{a,b}}{1 - \sum^{K-1}_{j=\ell+1}\eta^{k,j}_{a,b}}
).
\end{align*}
\section{Details of the Simulation Study}

\subsection{Simulating Stratified Population Demographics} 
We begin by generating the population age distribution for each composite stratum $s \in \mathcal{S} = \mathcal{X}_1 \times \cdots \times \mathcal{X}_J$. Let $\bm{p} \in \mathbb{N}^A$ denote the vector of population sizes by age for a user chosen baseline reference population, and let $\bm{p}_s \in \mathbb{N}^A$ represent the corresponding population vector for stratum $s$. We derive $\bm{p}_s$ by applying smooth, stratum-specific multiplicative deviations to the baseline vector $\bm{p}$.

To ensure that demographic shifts vary smoothly across consecutive ages, we construct an $A \times A$ covariance matrix $\bm{\Sigma}$ by evaluating a squared exponential kernel, configured with an amplitude of $1$ and a lengthscale of $3$, over the discrete age grid. For each category $k_j \in \{1,\dots,K_j\}$ of the $j$-th stratification variable, we draw a random deviation vector $\bm{z}^{k_j} \in \mathbb{R}^A$ from a zero-mean multivariate normal distribution:
$$
\bm{z}^{k_j} \sim \text{MVNormal}(0, \bm{\Sigma}).
$$

To map these unconstrained deviations into valid proportions, we apply a softmax transformation across the $K_j$ categories at each specific age $a$:
$$
\varepsilon^{k_j}_a = \frac{\exp(\alpha_j z^{k_j}_a)}{\sum_{c=1}^{K_j} \exp(\alpha_j z^c_a)},
$$where the scalar parameter $\alpha_j$ can be altered to control the magnitude of the deviation from the reference age distribution. The population count for a specific category is then obtained by scaling the reference population by this proportion and rounding to the nearest integer:
$$p^{k_j}_a = \lfloor \varepsilon^{k_j}_a p_a \rceil.
$$

Finally, assuming conditional independence between the stratification variables for a given age, the population count for a specific combined stratum $s$ is calculated by multiplying the relative proportions of its constituent categories:
$$
p^s_a = \left\lfloor p_a \prod^J_{j=1} \frac{p^{k_j(s)}_a}{p_a} \right\rceil,
$$
where $k_j(s)$ denotes the specific category of the $j$-th variable associated with composite stratum $s$.

\subsection{Simulating a Representative Contact Survey Sample}

To simulate a representative survey, we determine the joint distribution of participants across all possible strata. For each $a$, let $\bm{q}^{(j)}_a$ be the vector of proportions for the $K_j$ categories of variable $j$. We calculate the joint proportion vector $\bm{q}_a$ by taking the Kronecker product of these marginal proportions
\begin{equation*}
    \bm{q}_a = \bm{q}^{(1)}_a \otimes \cdots \otimes \bm{q}^{(J)}_a.
\end{equation*}
This resulting vector $\bm{q}_a$ has $K_* = \prod K_j$ elements, representing the probability that an individual of age $a$ belongs to a specific combined stratum $s$.

Assuming the survey is representative of the total population, we generate a sample of $N$ participants. For each participant, we first sample an age $a$ from the global population distribution and then sample their stratum $s$ from a categorical distribution defined by $\bm{q}_a$. This gives us a synthetic participant pool that reflects the underlying demographic structure.

\subsection{Simulating Contact Patterns}

Real-world contacts consist of interactions that occur across various physical settings. Rather than generating these contact patterns from scratch, we leverage the synthetic contact matrices developed by \citet{mistry_inferring_2021}. Their data-driven approach provides realistic template matrices for contacts occurring in the household $\bT^{(H)}$, at school $\bT^{(S)}$, at work $\bT^{(W)}$ and in the community $\bT^{(C)}$.

We simulate a single baseline contact intensity matrix as a randomly weighted combination of these templates:
\begin{equation} \label{app:sim:template-mixture}
\bT = v^{(H)}\bT^{(H)} + v^{(S)}\bT^{(S)} + v^{(W)}\bT^{(W)} + v^{(C)}\bT^{(C)}
\end{equation}
where the mixing weights $(v^{(H)}, v^{(S)}, v^{(W)}, v^{(C)}) \sim \operatorname{Dirichlet}(1,1,1,1)$. We then scale this matrix by a intensity factor to obtain the contact intensity matrix $\bM = C \times \bT$ where $C > 0$ is the average marginal contact intensity.

By default, the matrix $\bM$ does not satisfy reciprocity, thus we apply the correction:
\begin{equation*}
    \tilde\bM = \frac12 (\bM + \bP \bM^\top \bP^{-1}).
\end{equation*}
where $\bP \in \RR^{A \times A}$ is a diagonal matrix of population sizes for the reference population. The global contact rate matrix is obtained as
\begin{equation*}
    \bGamma = \tilde \bM \bP^{-1}.
\end{equation*}

After establishing the baseline, we introduce variations for specific strata (e.g. how contact patterns differ by income or education). For each stratification variable $j$, we generate deviation matrices $\bD^{k_j,\ell_j}$ using the same mixture logic as the templates. We centre these in log-space to ensure they represent relative changes:
\begin{equation*}
    \bE^{k_j,\ell_j}_{a,b} := \log \bT^{k_j,\ell_j}_{a,b} - \frac{1}{A^2} \sum^A_{a=1} \sum^A_{b=1} \log \bT^{k_j,\ell_j}_{a,b}.
\end{equation*}
The final deviation matrix is
$$
\bD^{k_j,\ell_j} := \exp{\eta_j \bE^{k_j,\ell_j} + \nu_j \mathbb{I}(k_j = \ell_j)},
$$
where $\eta_j$ is a tuning parameter controlling the strength of the stratification effect (e.g. $\eta_j = 0.2$ mild, $\eta_j = 0.5$ moderate) and $\nu_j$ controls the degree of assortativeness within the same group. Reciprocity is enforced by setting the deviation matrices for category pairs where $k_j \neq \ell_j$ to
\begin{equation*}
    \bD^{\ell_j, k_j}_{b,a} \leftarrow \bD^{k_j, \ell_j}_{a,b}
\end{equation*}
and category pairs where $k_j = \ell_j$ to
\begin{equation*}
    \bD^{k_j,k_j}_{a,b} \leftarrow \sqrt{ \bD^{k_j,k_j}_{a,b} \times \bD^{k_j,k_j}_{b,a} }.
\end{equation*}
To find the deviation for a specific pair of combined strata $(s,t)$, we multiply the deviations of their constituent categories
\begin{equation*}
    d^{s,t}_{a,b} = \prod^J_{j=1} \bD^{k_j(s),k_j(t)}_{a,b}
\end{equation*}
This product captures the interaction of multiple stratification factors on contact behaviour. Let $S^{s,t}_{a,b} := P^s_a P^t_b / (P_a P_b)$ denote the product of population proportions. We use this to normalise the deviations
\begin{equation*}
    \delta^{s,t}_{a,b} := \frac{d^{s,t}_{a,b}}{\sum_{u\in \mcS} \sum_{v \in \mcS} d^{u,v}_{a,b} S^{u,v}_{a,b}}.
\end{equation*}
From here, we follow the definition of full contact rate modifiers to construct the stratified contact rates
$\gamma^{s,t}_{a,b} = \gamma_{a,b} \delta^{s,t}_{a,b}$ and subsequently the stratified contact intensities $m^{s,t}_{a,b} = \gamma_{a,b} \delta^{s,t}_{a,b} P^t_b$.

Finally, we generate the specific contact records. For each survey respondent $i = 1, \dots, N$ of age $a_i$ belonging to stratum $s_i$, we model the number of reported contacts with individuals of age $b$ in stratum $t$, denoted $Y^t_{i,b}$, using a Poisson process:
$$Y^t_{i,b} \sim \text{Poisson}( m^{s_i,t}_{a_i,b} \cdot \gamma_i ), \quad \gamma_i \sim \text{Gamma}(5, 5).
$$
Here, $m^{s_i,t}_{a_i,b}$ is the fully stratified contact intensity derived in the previous step, and $\gamma_i$ is a multiplicative individual-level random effect (mean 1) that introduces the characteristic overdispersion typically observed in empirical social contact surveys.

\subsection{Comparing fine-age and coarse-age methods}
\texttt{socialmixr} and Prem's model are designed to operate on data where the participant and contact age is aggregated into disjoint age groups $\mcB$ while van de Kassteele's model and the current framework utilize 1-year fine age records $\mcA$. To facilitate a direct comparison between these models, we must reconcile their resolutions. We can treat coarse-age models as approximations fine-age surfaces where the contact estimate for age pair $(a,b)$ is represented by the estimate of the coarse-age cell $(c,d)$ in which $a$ and $b$ are nested. We define two operations to transition between these resolutions:

Let $\bM \in \RR^{A \times A}$ be a fine-age contact matrix with elements $m_{a,b}$. For a set of coarse age group $\mcB$, the pixilation operation aggregates $\bM$ into a coarse-age matrix $\bW \in \mathbb{R}^{B \times B}$. Each element $w_{c,d}$ is defined as the population-weighted average:
\begin{equation}\label{eq:pixilation}
    w_{c,d} = \sum_{a \in c}\sum_{b \in d} \frac{P_a}{P_c} m_{a,b}
\end{equation}
where $P_a$ is the population of age $a$ and $P_c = \sum_{a \in c}P_a$ is the total population of age group $c$.

Let $\bW \in \mathbb{R}^{B \times B}$ be a coarse-age contact matrix. The depixilation operation projects $\bW$ back onto the fine-age resolution $A \times A$ to created an aligned matrix $\bar{\bM}$. For $a \in c$ and $b \in d$, the elements are defined as:
\begin{equation}\label{eq:depixilation}
    \bar{m}_{a,b} = \frac{P_c}{P_a} \frac{w_{c,d}}{|c||d|}.
\end{equation}
We note that depixilation does not recover lost information; it simply aligns the granularity for comparison. It is easy to verify that applying pixilation to a depixilated $\bar{\bM}$ recovers the original coarse matrix $\bW$.

The depixilation operation allows us to evaluate the performance of all models against a true 1-year age matrix. By depixilating the outputs of 
\texttt{socialmixr} or Prem's model, we can calculate standard error metrics, such as RMSE and MAPE, at the 1-year age resolution. Furthermore, we can assess uncertainty quantification using interval scores and quantile coverage percentages. It is important to note that coarse-age models face an inherent discretization error. This is the unavoidable loss of accuracy incurred by approximating heterogenous fine-age values with a single constant coarse-age estimate. This baseline error highlights the advantage of employing fine-age models when high-resolution data is available.

\subsection{Simulation Results}

Tables~\ref{tab:sim-cint-1}--\ref{tab:sim-cint-3} report contact intensity RMSE ($\times 10^2$), MAPE (\%), interval score (IS), and coverage (\%) for all 16 race/ethnicity-pair sub-matrices and the overall average, one table per assortativity level $\nu \in \{-1.5, 0, 1.5\}$.

\begin{table}    
    \begin{tabularx}{\textwidth}{ c c c c >{\raggedright\arraybackslash}X } 
    \toprule
    Scenario & Cov. Structure & Matrices $K^2$ & $N/K^2$ & Rationale \\ 
    \midrule
    A & 1 & 1 & 1500 & Age by age stratification only. \\
    B & 2 & 4 & 375 & Basic reciprocal pair (e.g., Male-Female).\\
    C & 4 & 16 & 93.8 & Granular reciprocal matrix (e.g. Income).\\
    D & $2 \times 3$ & 36 & 41.7 & Cross-stratification (e.g., Sex $\times$ Educ.).\\
    E & $2 \times 2 \times 2$ & 64 & 23.4 & Multi-variable stratification. \\
    \bottomrule
    \end{tabularx}

    \vspace{0.1cm}
    
    \begin{tabular*}{\textwidth}{@{\extracolsep{\fill}}lccccc@{}} 
    \toprule
    Model & \multicolumn{5}{c}{Scenario, mean (std.)} \\ 
    \cmidrule(lr){2-6}
     & A & B & C & D & E \\ \midrule
    MAPE & 14.2 (2.7) & 18.8 (2.9) & 23.6 (2.8) & 25.3 (3.2) & 26.5 (3.2) \\
    Interval score & 0.06 (<0.01) & 0.09 (0.01) & 0.05 (<0.01) & 0.04 (<0.01) & 0.03 (<0.01) \\
    95\% Coverage & 83.8 (4.9) & 95.6 (1.8) & 94.0 (1.4) & 95.4 (1.2) & 96.1 (1.1) \\
    Inference Time (s) & 132 (26) & 858 (257) & 1156 (215) & 2002 (343) & 2646 (413) \\
    \bottomrule
    \end{tabular*}
    \vspace{0.1cm}
    \label{app:tab:experiment-III}
    \caption{The top panel of the table outlines the various experimental configurations considered and the rationale behind each, with the number of survey respondents fixed at $N=1,500$ across all scenarios. The bottom panel details the performance of the g-BSCF models evaluated across multiple metrics. The number of basis functions we use in our penalized Bayesian spline priors are fixed at $40$.}.
\end{table}

\begin{sidewaystable}
\centering
\caption{Contact intensity matrix estimation errors for the case study with BICS demographics ($\nu = -1.5$). RMSE is scaled by $10^2$. Values shown as mean (standard deviation) across 100 simulation runs.}
\label{tab:sim-cint-1}
\resizebox{\linewidth}{!}{%
\begin{tabular}{lrrrrrrrr}
\toprule
 & \multicolumn{4}{c}{Fully-stratified model} & \multicolumn{4}{c}{Partially-stratified model} \\
\cmidrule(lr){2-5}\cmidrule(lr){6-9}
Contact pair & RMSE $(\times10^2)$ & MAPE (\%) & IS & Cov (\%) & RMSE $(\times10^2)$ & MAPE (\%) & IS & Cov (\%) \\
\midrule
Hispanic $\to$ Hispanic  & 0.38 (0.18) & 31.0 (12.6) & 0.01 (<0.01) & 72.6 (12.3) & 0.93 (0.94)   & 37.9 (6.9)   & 1.59 (0.29) & 100.0 (0.0) \\
Hispanic $\to$ NH Black  & 0.92 (0.45) & 18.6 (3.1)  & 0.02 (<0.01) & 85.1 (5.5)  & 4.50 (1.78)   & 88.2 (1.2)   & 0.58 (0.08) & 100.0 (0.1) \\
Hispanic $\to$ NH Other  & 0.72 (0.49) & 19.4 (3.1)  & 0.02 (<0.01) & 84.4 (5.5)  & 3.54 (1.95)   & 89.2 (1.3)   & 0.36 (0.05) & 99.9 (0.2)  \\
Hispanic $\to$ NH White  & 3.63 (1.95) & 12.1 (2.0)  & 0.10 (0.02) & 87.5 (3.7)  & 15.55 (5.50)  & 57.0 (2.1)   & 0.41 (0.08) & 87.0 (4.0)  \\
\midrule
NH Black $\to$ Hispanic  & 1.34 (0.76) & 18.6 (3.1)  & 0.03 (<0.01) & 85.1 (5.5)  & 6.91 (3.32)   & 88.2 (1.2)   & 1.06 (0.14) & 100.0 (0.1) \\
NH Black $\to$ NH Black  & 0.26 (0.08) & 30.2 (10.0) & 0.01 (<0.01) & 73.2 (10.8) & 0.97 (0.53)   & 82.0 (1.3)   & 1.64 (0.28) & 100.0 (0.0) \\
NH Black $\to$ NH Other  & 0.67 (0.31) & 19.5 (3.3)  & 0.02 (<0.01) & 84.6 (6.7)  & 3.68 (1.46)   & 97.3 (0.2)   & 0.71 (0.12) & 100.0 (0.0) \\
NH Black $\to$ NH White  & 3.63 (1.68) & 12.2 (2.1)  & 0.10 (0.01) & 88.8 (3.7)  & 15.80 (5.26)  & 50.1 (4.7)   & 0.48 (0.09) & 91.0 (3.5)  \\
\midrule
NH Other $\to$ Hispanic  & 1.43 (1.11) & 19.4 (3.1)  & 0.03 (<0.01) & 84.4 (5.5)  & 7.33 (4.67)   & 89.2 (1.3)   & 1.17 (0.21) & 99.9 (0.2)  \\
NH Other $\to$ NH Black  & 0.90 (0.40) & 19.5 (3.3)  & 0.02 (<0.01) & 84.6 (6.7)  & 4.91 (1.95)   & 97.3 (0.2)   & 1.26 (0.20) & 100.0 (0.0) \\
NH Other $\to$ NH Other  & 0.20 (0.06) & 33.3 (12.0) & <0.01 (<0.01) & 71.3 (11.8) & 0.77 (0.32)   & 95.4 (0.4)   & 1.60 (0.38) & 100.0 (0.0) \\
NH Other $\to$ NH White  & 4.29 (3.09) & 13.2 (2.7)  & 0.12 (0.02) & 88.8 (4.7)  & 15.98 (10.22) & 32.6 (4.9)   & 0.59 (0.10) & 96.4 (2.7)  \\
\midrule
NH White $\to$ Hispanic  & 1.26 (0.89) & 12.1 (2.0)  & 0.02 (<0.01) & 87.5 (3.7)  & 5.89 (2.88)   & 57.0 (2.1)   & 0.08 (0.01) & 87.0 (4.0)  \\
NH White $\to$ NH Black  & 0.81 (0.45) & 12.2 (2.1)  & 0.01 (<0.01) & 88.8 (3.7)  & 3.77 (1.52)   & 50.1 (4.7)   & 0.05 (0.01) & 91.0 (3.5)  \\
NH White $\to$ NH Other  & 0.72 (0.61) & 13.2 (2.7)  & 0.01 (<0.01) & 88.8 (4.7)  & 2.89 (2.13)   & 32.6 (4.9)   & 0.04 (<0.01) & 96.4 (2.7)  \\
NH White $\to$ NH White  & 0.69 (0.36) & 11.1 (2.2)  & 0.02 (<0.01) & 91.7 (3.7)  & 10.01 (1.90)  & 210.8 (15.5) & 0.83 (0.28) & 1.4 (1.7)   \\
\midrule
All                      & 2.02 (0.96) & 18.5 (2.4)  & 0.03 (<0.01) & 84.2 (2.8)  & 8.51 (3.02)   & 78.4 (2.0)   & 0.78 (0.08) & 90.6 (0.8)  \\
\bottomrule
\end{tabular}%
}
\end{sidewaystable}

\begin{sidewaystable}
\centering
\caption{Contact intensity matrix estimation errors for the case study with BICS demographics ($\nu = 0$). RMSE is scaled by $10^2$. Values shown as mean (standard deviation) across 100 simulation runs.}
\label{tab:sim-cint-2}
\resizebox{\linewidth}{!}{%
\begin{tabular}{lrrrrrrrr}
\toprule
 & \multicolumn{4}{c}{Fully-stratified model} & \multicolumn{4}{c}{Partially-stratified model} \\
\cmidrule(lr){2-5}\cmidrule(lr){6-9}
Contact pair & RMSE $(\times10^2)$ & MAPE (\%) & IS & Cov (\%) & RMSE $(\times10^2)$ & MAPE (\%) & IS & Cov (\%) \\
\midrule
Hispanic $\to$ Hispanic  & 1.06 (0.79) & 22.0 (5.2)  & 0.02 (<0.01) & 79.4 (10.1) & 4.69 (3.28)   & 78.7 (2.0) & 0.81 (0.18) & 100.0 (0.1) \\
Hispanic $\to$ NH Black  & 0.66 (0.32) & 21.6 (4.0)  & 0.02 (<0.01) & 82.2 (5.9)  & 3.12 (1.45)   & 86.4 (1.2) & 0.42 (0.06) & 100.0 (0.0) \\
Hispanic $\to$ NH Other  & 0.49 (0.33) & 22.2 (4.8)  & 0.01 (<0.01) & 81.3 (7.4)  & 2.43 (1.42)   & 87.9 (1.6) & 0.24 (0.04) & 100.0 (0.1) \\
Hispanic $\to$ NH White  & 2.42 (1.36) & 13.1 (2.3)  & 0.07 (0.01) & 89.1 (4.0)  & 10.80 (4.03)  & 61.5 (4.1) & 0.46 (0.11) & 93.4 (3.4)  \\
\midrule
NH Black $\to$ Hispanic  & 0.96 (0.55) & 21.6 (4.0)  & 0.02 (<0.01) & 82.2 (5.9)  & 4.84 (2.65)   & 86.4 (1.2) & 0.81 (0.15) & 100.0 (0.0) \\
NH Black $\to$ NH Black  & 0.66 (0.33) & 20.9 (5.3)  & 0.02 (<0.01) & 84.2 (8.1)  & 3.37 (1.71)   & 95.4 (0.3) & 0.85 (0.21) & 100.0 (0.0) \\
NH Black $\to$ NH Other  & 0.47 (0.20) & 22.1 (3.8)  & 0.01 (<0.01) & 81.5 (6.9)  & 2.57 (1.09)   & 97.0 (0.3) & 0.49 (0.13) & 100.0 (0.0) \\
NH Black $\to$ NH White  & 2.38 (1.20) & 13.2 (2.4)  & 0.07 (0.01) & 90.3 (3.8)  & 10.92 (4.12)  & 52.0 (5.9) & 0.45 (0.10) & 95.7 (2.6)  \\
\midrule
NH Other $\to$ Hispanic  & 0.98 (0.74) & 22.2 (4.8)  & 0.02 (<0.01) & 81.3 (7.4)  & 5.08 (3.39)   & 87.9 (1.6) & 0.81 (0.20) & 100.0 (0.1) \\
NH Other $\to$ NH Black  & 0.64 (0.26) & 22.1 (3.8)  & 0.02 (<0.01) & 81.5 (6.9)  & 3.42 (1.45)   & 97.0 (0.3) & 0.86 (0.22) & 100.0 (0.0) \\
NH Other $\to$ NH Other  & 0.47 (0.18) & 23.9 (6.6)  & 0.01 (<0.01) & 79.3 (11.4) & 2.49 (1.02)   & 98.9 (0.1) & 0.84 (0.28) & 100.0 (0.0) \\
NH Other $\to$ NH White  & 2.88 (2.29) & 14.3 (2.5)  & 0.08 (0.01) & 89.4 (4.1)  & 11.25 (8.19)  & 35.5 (5.2) & 0.49 (0.09) & 97.3 (2.2)  \\
\midrule
NH White $\to$ Hispanic  & 0.85 (0.61) & 13.1 (2.3)  & 0.01 (<0.01) & 89.1 (4.0)  & 4.19 (2.07)   & 61.5 (4.1) & 0.08 (0.02) & 93.4 (3.4)  \\
NH White $\to$ NH Black  & 0.54 (0.32) & 13.2 (2.4)  & 0.01 (<0.01) & 90.3 (3.8)  & 2.64 (1.16)   & 52.0 (5.9) & 0.05 (0.01) & 95.7 (2.6)  \\
NH White $\to$ NH Other  & 0.49 (0.45) & 14.3 (2.5)  & 0.01 (<0.01) & 89.4 (4.1)  & 2.08 (1.69)   & 35.5 (5.2) & 0.03 (<0.01) & 97.3 (2.2)  \\
NH White $\to$ NH White  & 2.47 (1.43) & 9.8 (1.9)   & 0.03 (<0.01) & 83.7 (5.5)  & 6.40 (2.13)   & 41.1 (4.2) & 0.36 (0.13) & 41.1 (7.7)  \\
\midrule
All                      & 1.57 (0.77) & 18.1 (2.0)  & 0.03 (<0.01) & 84.6 (2.9)  & 6.15 (2.55)   & 72.2 (1.4) & 0.50 (0.08) & 94.6 (0.8)  \\
\bottomrule
\end{tabular}%
}
\end{sidewaystable}

\begin{sidewaystable}
\centering
\caption{Contact intensity matrix estimation errors for the case study with BICS demographics ($\nu = 1.5$). RMSE is scaled by $10^2$. Values shown as mean (standard deviation) across 100 simulation runs.}
\label{tab:sim-cint-3}
\resizebox{\linewidth}{!}{%
\begin{tabular}{lrrrrrrrr}
\toprule
 & \multicolumn{4}{c}{Fully-stratified model} & \multicolumn{4}{c}{Partially-stratified model} \\
\cmidrule(lr){2-5}\cmidrule(lr){6-9}
Contact pair & RMSE $(\times10^2)$ & MAPE (\%) & IS & Cov (\%) & RMSE $(\times10^2)$ & MAPE (\%) & IS & Cov (\%) \\
\midrule
Hispanic $\to$ Hispanic  & 2.24 (2.05) & 19.8 (4.9)  & 0.04 (0.01) & 83.6 (8.3)  & 10.45 (6.98) & 92.2 (0.6)   & 0.19 (0.05) & 97.6 (3.2)  \\
Hispanic $\to$ NH Black  & 0.32 (0.17) & 28.1 (7.4)  & 0.01 (<0.01) & 77.2 (8.8)  & 1.33 (0.85)  & 80.2 (1.7)   & 0.23 (0.06) & 100.0 (0.0) \\
Hispanic $\to$ NH Other  & 0.23 (0.13) & 29.0 (7.0)  & <0.01 (<0.01) & 75.9 (8.8)  & 1.03 (0.68)  & 82.9 (1.8)   & 0.12 (0.04) & 100.0 (0.0) \\
Hispanic $\to$ NH White  & 1.06 (0.55) & 16.7 (2.9)  & 0.03 (<0.01) & 89.5 (4.1)  & 6.26 (1.90)  & 99.8 (11.2)  & 0.53 (0.16) & 98.1 (2.2)  \\
\midrule
NH Black $\to$ Hispanic  & 0.46 (0.28) & 28.1 (7.4)  & 0.01 (<0.01) & 77.2 (8.8)  & 2.07 (1.51)  & 80.2 (1.7)   & 0.49 (0.17) & 100.0 (0.0) \\
NH Black $\to$ NH Black  & 1.30 (0.86) & 19.1 (3.9)  & 0.03 (<0.01) & 87.0 (6.5)  & 6.80 (3.57)  & 98.5 (0.1)   & 0.19 (0.08) & 99.5 (1.3)  \\
NH Black $\to$ NH Other  & 0.23 (0.10) & 28.5 (6.2)  & 0.01 (<0.01) & 76.3 (7.6)  & 1.13 (0.59)  & 95.7 (0.4)   & 0.24 (0.11) & 100.0 (0.0) \\
NH Black $\to$ NH White  & 1.07 (0.54) & 17.4 (3.3)  & 0.03 (0.01) & 89.3 (4.3)  & 5.42 (2.22)  & 78.4 (9.7)   & 0.40 (0.13) & 99.1 (1.1)  \\
\midrule
NH Other $\to$ Hispanic  & 0.46 (0.28) & 29.0 (7.0)  & 0.01 (<0.01) & 75.9 (8.8)  & 2.16 (1.62)  & 82.9 (1.8)   & 0.42 (0.18) & 100.0 (0.0) \\
NH Other $\to$ NH Black  & 0.31 (0.13) & 28.5 (6.2)  & 0.01 (<0.01) & 76.3 (7.6)  & 1.50 (0.79)  & 95.7 (0.4)   & 0.42 (0.20) & 100.0 (0.0) \\
NH Other $\to$ NH Other  & 0.88 (0.45) & 21.1 (4.9)  & 0.02 (<0.01) & 85.3 (7.4)  & 4.92 (2.15)  & 99.7 (0.0)   & 0.19 (0.12) & 99.4 (1.7)  \\
NH Other $\to$ NH White  & 1.29 (1.21) & 18.3 (3.7)  & 0.04 (0.01) & 88.8 (5.2)  & 5.48 (4.73)  & 58.8 (8.8)   & 0.35 (0.09) & 99.0 (1.4)  \\
\midrule
NH White $\to$ Hispanic  & 0.36 (0.24) & 16.7 (2.9)  & 0.01 (<0.01) & 89.5 (4.1)  & 2.41 (0.99)  & 99.8 (11.2)  & 0.09 (0.03) & 98.1 (2.2)  \\
NH White $\to$ NH Black  & 0.24 (0.14) & 17.4 (3.3)  & 0.01 (<0.01) & 89.3 (4.3)  & 1.31 (0.62)  & 78.4 (9.7)   & 0.04 (0.01) & 99.1 (1.1)  \\
NH White $\to$ NH Other  & 0.22 (0.24) & 18.3 (3.7)  & <0.01 (<0.01) & 88.8 (5.2)  & 1.02 (0.96)  & 58.8 (8.8)   & 0.02 (<0.01) & 99.0 (1.4)  \\
NH White $\to$ NH White  & 5.95 (2.98) & 11.6 (1.3)  & 0.04 (<0.01) & 61.7 (7.5)  & 9.44 (8.26)  & 8.7 (1.9)    & 0.08 (0.02) & 87.0 (5.3)  \\
\midrule
All                      & 1.85 (0.88) & 21.7 (2.6)  & 0.02 (<0.01) & 82.0 (3.0)  & 5.31 (2.67)  & 80.7 (2.4)   & 0.25 (0.07) & 98.5 (0.7)  \\
\bottomrule
\end{tabular}%
}
\end{sidewaystable}


\section{socialmixr}
\label{app:socialmixr}
\texttt{socialmixr} is a popular R package for estimating social contact matrices from contact survey data~\citep{funk_socialmixr_2024}. This section details the statistical methods used in the package to derive age-stratified contact estimates. Because it is such a widely used tool for deriving contact estimates, we give a formal description of the uncertainty quantification methodology and give a theoretical discussion of why it may fail. Finally we provide details on how we extend the method to accommodate stratification by multiple variables.

\subsection{Estimation Procedure}

Let $\mcB$ denote a set of $B$ age ranges (e.g. 0-4, 5-9, ..., 75-79, 80-84). For a respondent $i$ belonging to age range $c_i \in \mcB$, we let $y_{i,d}$ denote the number of reported contacts with individuals in age range $d \in \mcB$. The aggregate number of contact from age range $c$ to age range $d$ is given by $y_{c,d} = \sum^N_{i=1} 1_c(c_i) y_{i,d}$ where $1_c(\cdot)$ is an indicator function for membership in age range $c$.

Empirical social contact intensities $\hat{m}_{c,d}$, is defined as the average number of contacts reported by respondents in age range $c$ with members of age range $d$ in the population:
\begin{equation}\label{eq:socialmixr-contact-intensity}
    \hat{m}_{c,d} = y_{c,d} / N_c
\end{equation}
where $N_c$ denotes the number or respondents in age range $c$. Given the population size of age range $d$, $P_d$, the empirical contact rate is estimated as $\hat{\gamma}_{c,d} = \hat{m}_{c,d}/P_d$.

To satisfy the reciprocity constraint, \texttt{socialmixr} applies a post-hoc correction. The adjusted contact intensity, $\tilde{m}_{c,d}$, is computed as:
\begin{equation} \label{eq:socialmixr-adj-contact-intensity}
    \tilde{m}_{c,d} = \frac{1}{2 P_c} (\hat{m}_{c,d}P_c + \hat{m}_{d,c}P_d)
\end{equation} The corrected contact rate $\tilde{\gamma}_{c,d}$, simplifies to the arithmetic mean of the original rates
$$
\tilde{\gamma}_{c,d} = \frac{\tilde{m}_{c,d}}{P_d}= \frac{\hat{m}_{c,d}P_c + \hat{m}_{d,c}P_d}{2 P_c P_d} = \frac12(\hat{\gamma}_{c,d} + \hat{\gamma}_{d,c}).
$$

\subsection{Uncertainty Quantification via Bootstrapping}
Uncertainty in the contact estimates is quantified using a non-parametric bootstrap. The procedure assigns a weight of $1/N$ to each survey respondents and performs a multinomial resampling step at each iteration to select the respondent to include in the bootstrap sample. Let $\mathcal{I}^{(j)}$ be the set of respondent indices for the $j$-th (for $j=1,\ldots,J$) bootstrap sample. The contact intensities and rates are calculated as
\begin{equation*}
    \hat{m}^{(j)}_{c,d} = y^{(j)}_{c,d} / N^{(j)}_c, \quad \hat{\gamma}^{(j)}_{c,d} = \hat{m}^{(j)}_{c,d} / P_d.
\end{equation*}
where $y^{(j)}_{c,d} := \sum_{i \in \mathcal{I}^{(j)}} 1_c(c_i) y_{i,d} $ and $N^{(j)}_c := \sum_{i \in \mathcal{I}^{(j)}} 1_c(c_i)$
The distribution of these bootstrap estimates allows for the calculation of variances and quantiles. For example, the sampling variance of the contact intensity can be estimated as
\begin{equation*}
    \text{Var}(\hat{m}_{c,d}) = \frac{1}{J-1} \sum^J_{j=1} \qty(\hat{m}^{(j)}_{c,d} - \hat{m}_{c,d})^2.
\end{equation*}

\subsection{Bootstrap Failure Probability}
The multinomial resampling step can cause the bootstrap algorithm to fail when the sample size of certain age ranges are small. More specifically, at each bootstrap iteration, there is a non-zero probability that $N^{(j)}_c = 0$ for a subset of $c \in \mcB$, which results in \eqref{eq:socialmixr-contact-intensity} being undefined (the program terminates due to a zero-division error). To reduce the probability of failure, a common pragmatic work-around is to the define the age ranges $\mcB$ such that for every $c \in \mcB$, $N_c$ is large enough for $\PP(N^{(j)}_c = 0)$ to be sufficiently small. Here, we examine the failure probability in detail, and provide a rule-of-thumb for determining the appropriate $N_c$.

For the $j$-th bootstrap iteration, let the vector of sample sizes for age ranges be denoted as $\{N_c^{(j)}\}^B_{c=1}$. This vector follows a multinomial distribution:
\begin{equation*}
    \{N_c^{(j)}\}^B_{c=1} \sim \text{Multinomial}(N, \{N_c/N\}^B_{c=1}).
\end{equation*}
When the sample size $N$ is sufficiently large and $N_c$ is reasonably large for most groups, we can assume that the events $N_c = 0$ are nearly independent. Hence, we can approximate the probability of any $N_c = 0$ using the union bound:
\begin{equation*}
\PP(\exists c: N^{(j)}_c = 0) \approx \sum^B_{c=1} \PP(N^{(j)}_c = 0) = \sum_{c=1}^B  ( 1 - N_c/N )^N \approx \sum_{c=1}^B e^{-N_c}.
\end{equation*}
This probability is dominated by the smallest $N_c$ which means that even if the total sample size $N$ is large, if a single age range has a small sample size, e.g. $N_c=3$, then the probability of failure at any given bootstrap iteration is approximately $e^{-3} \times 100 \approx 4.9\%$. The union bound also provides a practical rule-of-thumb for adjusting the size of age ranges. If the user wants to control the failure rate at $\alpha$ for a bootstrap procedure consisting of $J$ replicates, adjust the size of age ranges such that
\begin{equation}\label{eq:socialmixr-rule-of-thumb}
    \sum^B_{c = 1} e^{-N_c} \leq \frac{\alpha}{J}.
\end{equation}

\subsection{Limitations for Uncertainty Quantification of Social Contacts Estimates}

While the bootstrap is simple to implement and highly scalable, relying on it to quantify uncertainty in age-stratified social contact intensities can be inherently risky in sparse data regimes. The reliability of bootstrap-derived uncertainty is fundamentally limited by how well the empirical distribution of the observed contact counts represents the true underlying data-generating distribution. If the stratum sample size $N_c$ is small, the empirical distribution will almost certainly be a poor representation of the true distribution.

Formally, let $\Pi_{c,d}$ denote the true underlying probability distribution of the contact counts $Y_{i,d}$ between a respondent $i$ and members of age range $d$ in the population. The standard bootstrap procedure approximates this underlying distribution by resampling with replacement directly from the observed data. Because it relies strictly on the support of this empirical distribution, the estimate $\hat{m}^{(j)}_{c,d}$ is confined to the interval $[\min_{i: c_i = c}(y_{i,d}), \max_{i:c_i=c} (y_{i,d})].$
If the original sample size is small, this empirical support is potentially unrepresentative of the true population support. If the survey fails to capture a sufficient range of contact counts, the bootstrap will systematically underestimate $\text{Var}(\hat{m}_{c,d})$, leading to overly narrow confidence intervals.

Furthermore, for asymptotic consistency to hold, the bootstrap distribution of the standardised statistic $(N^{(j)}_c)^{1/2}(\hat{m}^{(j)}_{c,d} - \hat{m}_{c,d})$ must converge in distribution to $(N_c)^{1/2} (\hat{m}_{c,d} - m_{c,d})$ \citep{bickel_asymptotic_1981}. When $N_c$ is small, this asymptotic justification breaks down. The discrete nature of the empirical distribution dominates the resampling process, producing a spiky bootstrap distribution that fails to capture the underlying smoothness of the true sampling distribution. Consequently, the resulting confidence intervals for $m_{c,d}$ may deviate significantly from their nominal $1 - \alpha$ coverage levels.

\subsection{Extensions for Multiple Stratification and Reciprocity}
We aim to estimate contact intensities $m^{s,t}_{c,d}$ for all $c,d \in \mcB$ and $s,t \in \mcS = \mcX_1 \times \cdots \times \mcX_J$. To do so, we independently estimate the sample mean then apply a post-hoc reciprocity correction to enforce reciprocity.

The empirical social contact intensity is estimated as:
\begin{equation*}
    \hat{m}^{s,t}_{c,d} = y^{s,t}_{c,d}/N^s_c
\end{equation*}
provided that $N^s_c > 0$. To ensure that the bootstrap algorithm does not terminate due to zero-division error, we iteratively combine the age range with smallest sample sizes with their neighbouring age range until the threshold in \eqref{eq:socialmixr-rule-of-thumb} is satisfied ($\alpha = 0.05$ and $J = 3000$). If a bootstrap iteration still fails due to zero division error, we skip to the next iteration.

To enforce reciprocity, we apply post-hoc corrections. For the within-stratum case where $s = t$, correction is the same as \eqref{eq:socialmixr-adj-contact-intensity}:
\begin{equation*}
    \tilde{m}^{s,s}_{c,d} = \frac{1}{2 P^s_c} (\hat{m}^{s,s}_{c,d}P^s_c + \hat{m}^{s,s}_{d,c}P^s_d).
\end{equation*}
For the between-stratum case where $s\neq t$, we apply the following correction:
\begin{equation*}
    \tilde{m}^{s,t}_{c,d} = \frac12 \qty(\hat{m}^{s,t}_{c,d} + \hat{m}^{t,s}_{d,c} \frac{P^t_d}{P^s_c}).
\end{equation*}
\section{Details on Prem's model}\label{app:prem}
\subsection{Original Model Definition}
\citet{prem_projecting_2017} proposed a hierarchical Bayesian model to infer contact intensity patterns in different contexts (household, school, workplace, and community). We revisit the original model definition in \citet{prem_projecting_2017}, then describe in detail how we implement and extend their model.

Let $\mcB$ denote a set of $B$ age ranges. For a respondent $i$ in age $c_i \in \mcB$, let $Y_{i,d}$ be the number of reported contacts with members of age range $d$ in the population. The contact counts are modelled using the following Poisson model
\begin{equation*}
    Y_{i,d} \sim \text{Poisson}(m_{c_i, d}\sigma_i), \quad \log m_{c_i,d} = \beta_0 + \beta_{c_i,d}.
\end{equation*}
where $\beta_0$ represents the global contact intensity,  $\beta_{c_i,d}$ are age range specific offsets from $\beta_0$ for $c$ and $d$, and $\sigma_i$ is a random effect capturing individual-level heterogeneity.

The original prior specifications were designed for implementations in JAGS and are as follows:
\begin{align*}
    \exp(\beta_0) &\sim \text{Exp}(0.0001), \\
    \sigma_i &\sim \text{Gamma}(\theta, \theta), \\
    \theta &\sim \text{Exp}(0.0001).
\end{align*}
To enforce smoothness across the contact surface, the age range specific offsets $\beta_{c,d}$ are defined by their neighbours in a two-dimensional lattice
\begin{equation*}
        \beta_{c_i,d} = \frac{1}{|\mathcal{N}_{c_i,d}|} \sum_{(c',d') \in \mathcal{N}_{c_i,d}} \beta_{c',d'}
\end{equation*}
where $\mathcal{N}_{c_i,d}$ denotes the set of age ranges adjacent to $(c_i, d)$. The underlying coefficients $\beta_{c',d'}$ are assigned a diffuse normal prior $\beta_{c',d'} \sim \text{Normal}(0, 100^2)$.

Crucially, the prior on the vector of offsets $\bbeta \in \RR^{B^2}$ is equivalent to a first-order IGMRF expressed as $\bbeta \sim \text{MVNormal}(\bm{0}, (\tau \bQ)^{-1})$. Here, $\tau \in \RR$ represents the precision, fixed at $1/100^2$ in the original work, and $\bQ$ is a precision matrix that encodes the adjacency of age ranges on the 2D lattice. This matrix is constructed via the Kronecker sum $\bQ := \bQ_0 \oplus \bQ_0$ where $\bQ_0 = \bD^\top \bD$ and $\bD$ is a $(B -1)\times B$ first order difference matrix:
\begin{equation*}
    \bD = \begin{pmatrix}
    -1 & 1 & & & \\
       & -1 & 1 & & \\
       & & \ddots &\ddots & \\
       & & & -1 & 1
    \end{pmatrix}.
\end{equation*}

\subsection{Adapted Model and Implementation for Modern Probabilistic Programming Languages}
\label{app:prem-adaptation}
To adapt this model for NumPyro (or other modern probabilistic programming languages such as Stan or PyMC), and ensure numerical stability during NUTS or SVI inference, we modify the prior distributions as such:
\begin{align*}
    \beta_0 &\sim \text{Normal}(0, 2.5^2), \\
    \log(\sigma_i) &\sim \text{Normal}(\mu_\sigma, \tau_\sigma^2), \\
    \mu_\sigma &\sim \text{Normal}(0, 1), \\
    \tau_\sigma &\sim \text{HalfNormal}(1), \\
    \tau &\sim \text{Gamma}(2, 1).
\end{align*}

Sampling from the IGMRF presents a computational challenge because $\bQ$ is singular, preventing the use of standard Cholesky-based multivariate Gaussian sampling. To overcome this, we utilise the eigenvalue decomposition of the precision matrix. We first factorise the 1D component $\bQ_0 = \bU_0 \bm{\Lambda}_0 \bU_0^\top$, where $\bU_0$ is a matrix of eigenvectors and $\bm{\Lambda}_0$ is the diagonal matrix of eigenvalues. Given that $\bQ_0$ has rank $B - 1$, we form the reduced matrices $\tilde{\bU}_0 \in \RR^{B \times (B - 1)}$ and $\tilde \Lambda_0 \in \RR^{(B-1) \times (B-1)}$ by removing the zero eigenvalue and its corresponding eigenvector. The Moore-Penrose pseudo-inverse of $\bQ_0$ is then $\bQ_0^+ = \tilde{\bU}_0 \tilde{\bm{\Lambda}}_0^{-1} \tilde{\bU}^\top_0$.

By defining the Kronecker products $\tilde{\bU} = \tilde{\bU}_0 \otimes \tilde{\bU}_0$ and $\tilde{\bm{\Lambda}} = \tilde{\bm{\Lambda}}_0 \otimes \bI_{B-1} + \bI_{B-1} \otimes \tilde{\bm{\Lambda}}_0$ we can express the desired precision matrix as $\bQ = \tilde{\bU} \tilde{\bm{\Lambda}} \tilde{\bU}^\top$
and its pseudo-inverse as $\bQ^+ = \tilde{\bU} \tilde{\bm{\Lambda}}^{-1} \tilde{\bU}^\top$. 
This decomposition allows us to perform an efficient reparameterised sampling of the IGMRF. We compute $\bbeta = \tau^{-1/2} \tilde{\bU} \tilde{\bm{\Lambda}}^{-1/2} \bm{z}$, where $\bm{z}$ is a vector of independent standard normal draws $\bm{z} \sim \text{MVNormal}(\bm{0}, \bI_{(B-1)^2})$. This approach ensures that the sampled offsets remain within the identifiable subspace of the IGMRF while maintaining the gradient information necessary for efficient NUTS transitions and SVI updates.

\subsection{Extensions for Stratification and Reciprocity}
In the original work, Prem's model does not account for additional covariate stratification nor does it enforce reciprocity. Here, we propose an extension that address both limitations. We aim to estimate contact intensities $m^{s,t}_{c,d}$ for all $c,d \in \mcB$ and $s,t \in \mcS = \mcX_1 \times \cdots \times \mcX_J$.

A defining characteristic of Prem's model is that it infers contact intensity surfaces directly rather than contact rate surfaces. While this approach avoids the necessity of external population size data during the primary inference stage, it precludes the use of the symmetry constraints on contact rates employed in \citet{kassteele_efficient_2017} or \citet{dan_estimating_2023}. Because $\bM^{s,t}$ and $\bM^{t,s}$ are not generally equal, we must independently infer $\bM^{s,t}$ for all pairs. Consequently, we extend the observation model as
\begin{equation*}
    Y^{s_i,t}_{c_i,d} \sim \text{Poisson}(m^{s_i,t}_{c_i,d} \sigma_i), \quad \log m^{s_i,t}_{c_i,d} = \beta^{s_i,t}_0 + \beta^{s_i,t}_{c_i,d}.
\end{equation*}
We apply the same prior structures described previously to these stratified parameters. Notably, this formulation differs from fitting separate models to each $(s,t)$ pair, as the individual random effect terms $\sigma_i$ are shared across all stratifications, allowing the model to pool information regarding individual-level reporting heterogeneity.

For scenarios involving partial or mixed stratification, reciprocity adjustments are not required since the underlying contact rate matrices lack inherent symmetry. However, in the full stratification scenario, we employ a post-hoc reciprocity adjustment for each posterior sample of the contact intensity matrix. For the within-stratum case where $s = t$, the adjustment is defined as:
\begin{equation*}
    \bM^{s,s} \leftarrow \frac12 \qty(\bM^{s,s} + (\bP^s)^{-1} (\bM^{s,s})^\top \bP^s)
\end{equation*}
where $\bP^s$ denotes the diagonal matrix of population sizes for stratum $s$. For the between-stratum case where $s\neq t$, the reciprocal relationship is enforced by setting:
\begin{equation*}
    \bM^{s,t} \leftarrow \frac12 \qty(\bM^{s,t} + (\bM^{t,s})^\top (\bP^s)^{-1} \bP^t).
\end{equation*}
This adjustment ensures that the total number of contacts between any two age ranges across strata is balanced.
\section{Details on van de Kassteele's model}
\label{app:vdKassteele}
\subsection{Original Model Definition}
In \citet{kassteele_efficient_2017} the authors propose a hierarchical Bayesian model for inferring age- and sex- specific contact patterns when 1 year age data is available. Here, we revisit the original model definition and subsequently describe how we extend and implement the framework for our analysis.

Let $\mcA = \{0,\ldots,A\}$ denote the set of 1 year age, and let $\mcS = \{M, F\}$ denote the stratification space comprising male and female sexes. For ages $a,b \in \mcA$, and sexes $s, t \in \mcS$, let $Y^{s,t}_{a,b}$ denote the total number of reported contacts between participants of sex $s$ and age $a$, and contacts of sex $t$ age $b$. The observed contact counts are modelled using the following negative binomial distribution to account for overdispersion:
\begin{equation*}
    Y^{s,t}_{a,b} \sim \text{NegBin}(\mu^{s,t}_{a,b}, \varphi), \quad \log \mu^{s,t}_{a,b} = \beta_0 + \beta^{s,t}_{a,b} + \log N^s_a + \log P^t_b.
\end{equation*}
where $\varphi$ is the overdispersion parameter, $\beta_0$ is the global intercept (baseline contact rate), and $\beta^{s,t}_{a,b}$ are age- and sex- specific deviations from the baseline. The term $N^s_a$ denotes number of participants of sex $s$ and age $a$, while $P^t_b$ represents the population size of sex $t$ and age $b$. The following priors are assigned to the unknown parameters:
\begin{equation*}
    \log(\varphi) \sim \text{Normal}(0, 1000), \quad \beta_0 \sim \text{Normal}(0, 1000).
\end{equation*}

The prior on the interaction terms $\beta^{s,t}_{a,b}$ is a custom second-order IGMRF that imposes smoothness and reciprocity. Let $\bD \in \RR^{A-2 \times A}$ be the second-order difference operator matrix
\begin{equation}\label{app:second-order-difference-operator}
    \bD = \begin{pmatrix}
    1 & -2 & 1 &  &  & \\
    & 1 & -2 & 1  &  &\\
    & & \ddots & \ddots  & \ddots & \\
    & & & 1 & -2 & 1
    \end{pmatrix}.
\end{equation}
Difference operator matrices that operates on the participant (source) and contact (target) age dimensions can be constructed using Kronecker producs:
\begin{equation*}
    \bD_1 = \bI_A \otimes \bD, \quad \bD_2 = \bD \otimes \bI_A.
\end{equation*}
Here, $\bD_1$ applies smoothing across the participant ages (rows), while $\bD_2$ applies smoothing across the contact ages (columns). From these operators, we construct the precision matrix for the IGMRF.

For non-symmetric cases where $s \neq t$ (e.g., Male-to-Female contacts), the precision matrix $\bQ^{s,t}$ is given by
\begin{equation}\label{app:construct-precision-matrix}
    \bQ^{s,t} = \bD_1^\top\bD_1 + \bD^\top_2 \bD_2 \quad \text{for } s \neq t.
\end{equation}
Let $\bB^{s,t} \in \RR^{A \times A}$ denote the matrix of parameters $\beta^{s,t}$ and $\bbeta^{s,t} = \text{vec}(\bB^{s,t})$ be its vectorization. The prior on $\bbeta^{s,t}$ is specified as:
\begin{equation*}
    \bbeta^{s,t} \sim \text{MVNormal}(\bm{0}, (\tau \bQ^{s,t})^{-1})
\end{equation*}
where $\tau$ is a global precision parameter. Due to the reciprocity of the contact rates, we only need to infer one off-diagonal matrix (e.g. $\bB^{M,F}$), as the reciprocal matrix is determined by $\bB^{t,s} = (\bB^{s,t})^\top$.

For the symmetric cases where $s = t$ (Male-to-Male or Female-to-Female contacts), reciprocity implies that $\beta^{s,s}_{a,b} = \beta^{s,s}_{b,a}$. Therefore, the inference is restricted to the lower-triangular elements (including the diagonal) of $\bB^{s,s}$, defined by the index set $\mathcal{I}_{\text{tril}} = \{ (a,b): 1 \leq b \leq a \leq A \}$.
To construct the reduced precision matrix $\bQ^{s,s}_{\text{tril}}$, we modify the difference operators $\bD_1$ and $\bD_2$ by removing specific rows and columns.
\begin{itemize}
\item Column Removal: We retain only the columns corresponding to the indices in $\mathcal{I}_{\text{tril}}$. This reduces the parameter vector $\bbeta^{s,s}$ from length $A^2$ to $A(A+1)/2$.
\item Row Removal: A row in the full difference operator corresponds to a second-order difference constraint centered at a specific coordinate $(a,b)$. We retain a row if and only if all three nodes involved in its finite difference stencil lie within $\mathcal{I}_{\text{tril}}$. For $\bD_1$, the stencil involves $\{(a-1, b), (a, b), (a+1, b)\}$. We keep the row corresponding to $(a,b)$ if $(a-1, b) \in \mathcal{I}_{\text{tril}}$ and $(a+1, b) \in \mathcal{I}_{\text{tril}}$. For $\bD_2$, the stencil involves $\{(a, b-1), (a, b), (a, b+1)\}$. We keep the row corresponding to $(a,b)$ if $(a, b-1) \in \mathcal{I}_{\text{tril}}$ and $(a, b+1) \in \mathcal{I}_{\text{tril}}$.
\end{itemize}
Consequently, smoothness constraints are not enforced across the diagonal boundary, effectively treating the diagonal as a natural boundary of the random field. The reduced operators $\bD_{\text{tril},1}$ and $\bD_{\text{tril},2}$ are then used to construct $\bQ^{s,s}_{\text{tril}} = \bD_{\text{tril},1}^\top \bD_{\text{tril},1} + \bD_{\text{tril},2}^\top \bD_{\text{tril},2}$. The prior is then:
\begin{equation*}
    \bbeta^{s,s}_\text{tril} \sim \text{MVNormal}(\bm{0}, (\tau \bQ_\text{tril})^{-1}).
\end{equation*}
The full matrix $\bB^{s,s}$ is reconstructed by mirroring the lower-triangular elements to the upper triangle. In the original implementation, the global precision parameter was assigned a prior $\tau \sim \text{Gamma}(1, 0.0001)$ and inference was performed using Integrated Nested Laplace Approximations (INLA) \citep{rue_approximate_2009}.

\subsection{Adpated Model and Implementation for Modern Probabilistic Programming Languages}
To adapt this model for NumPyro (or other Hamiltonian Monte Carlo-based PPLs), and ensure numerical stability during NUTS or SVI inference, we modify the prior distributions as follows:
\begin{equation*}
    \varphi \sim \text{Exp}(1), \quad \beta_0 \sim \text{Normal}(- \log \bar{P}, 2.5^2), \quad \tau \sim \text{Gamma}(2.0, 0.1).
\end{equation*}
where $\log \bar{P}$ is the logarithm of the average population size across all strata and ages, serving as a centering constant. As noted in Section \ref{app:prem-adaptation}, sampling from an IGMRF with a rank-deficient precision matrix is an issue. For the full non-symmetric matrices, we utiliz the eigen-decomposition technique described in \ref{app:prem-adaptation}. For the symmetric matrices subject to the lower-triangular constraint, we first construct the precision matrix $\bQ_\text{tril}$ and then apply the same eigen-decomposition trick to sample the free parameters efficiently.

\subsection{Extensions for Multiple Stratification}
van de Kassteele's model extends naturally to accommodate stratification by additional variables. Instead of being defined by a single variable, the stratification space $\mcS$ consists of composite strata defined by multiple variables. For any two strata $s, t \in \mcS$, if $s \neq t$, we sample a standard non-symmetric matrix from the second-order IGMRF. If $s = t$, we sample from the lower triangular constrained prior and enforce symmetry by mirroring the elements to the upper triangle. In the full stratification scenario for a stratification space $\mcS$ consisting of $K$ strata, we sample $(K-1)K/2$ non-symmetric matrices and $K$ symmetric matrices. The remaining $(K-1)K/2$ are determined via the mirroring trick.

\section{Additional Details for BICS}

\subsection{Data Preprocessing}

\subsubsection{Participant Data} 
To maintain consistency in age resolution within our fine-age modeling framework, participants who reported their age only as a age range were assigned a 1-year age drawn uniformly from within their reported range. Furthermore, to ensure the integrity of the stratification covariates, the analysis was restricted to complete cases; any participant missing age, sex, education, or race/ethnicity information was excluded. Finally, the study sample was restricted to adults aged 18 to 84 years.

\subsubsection{Contact Data} 
The theoretical reciprocity constraints underpinning our statistical framework mathematically require the contact intensity matrices to be square. Consequently, to conform with the respondent age domain, the contact records were similarly restricted to contacted individuals aged 18 to 84 years.

\subsubsection{Population Data} 
Stratified population size estimates for the year 2020 were obtained from the American Community Survey via IPUMS USA \citep{ruggles_ipums_2025}. To mitigate the effects of sampling noise, we stabilised the population counts by applying a 5-year moving average across the age dimension for each composite stratum.

\subsection{Modeling of Group Contacts}

\begin{figure}
    \centering
    \includegraphics[width=\linewidth]{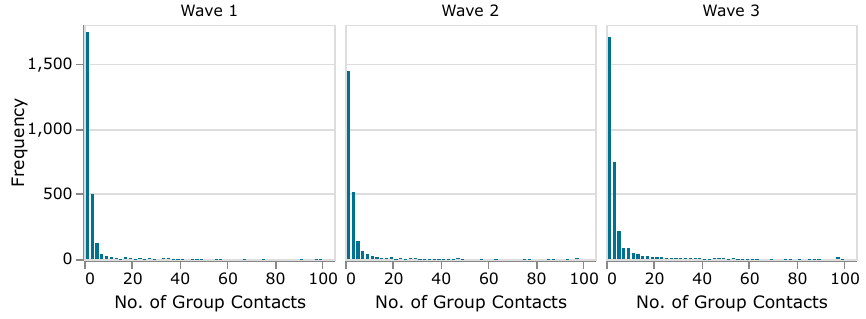}
    \caption{The distribution of group contact report in the Berkeley Interpersonal Contact Study by survey wave.}
    \label{fig:bics-group-contacts-distribution}
\end{figure}

\begin{figure}
    \centering
    \includegraphics[width=\linewidth]{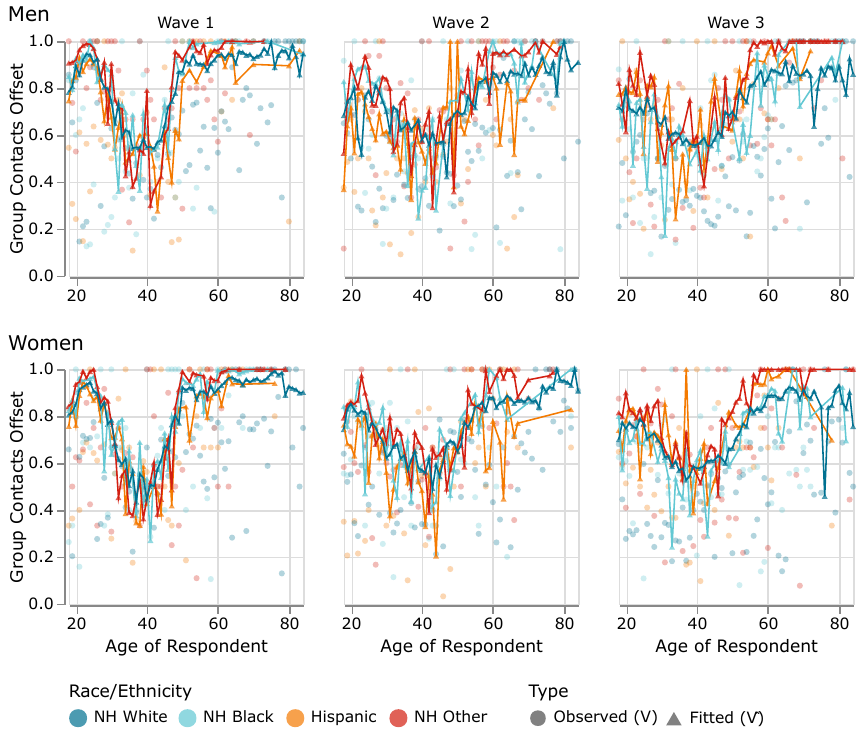}
    \caption{The value of the group contact offset terms $V^k_a$ by respondent age, sex, and race/ethnicity, faceted by BICS survey wave. Round points denote values calculated from the raw data. Triangular points connected by line segments denote values calculated from the predictions of a quantile Generalized Additive Model for individual group contact reports.}
    \label{fig:bics-group-contacts-modelled}
\end{figure}

The Berkeley Interpersonal Contact Study (BICS) allowed participants to report \emph{group contacts}, which are bulk counts of non-household contacts without any demographic details. Because the inclusion of these records significantly impacts the estimated contact intensities, they must be appropriately accounted for in the modeling process.

Following \citet{dan_estimating_2023}, we incorporate these group contacts into our statistical framework via an offset term in the linear predictor. Let $Y^k_a$ and $U^k_a$ denote the number of contacts with attribute information and group contacts by respondents in stratum $k$ and age $a$, respectively. To scale the marginal contact intensity $m^k_a$ so that it reflects the true total contact volume, we first calculate the proportion of contacts that have attribute information:
\begin{equation}\label{app:group-contact-offset}
   V^{k}_a = \frac{Y^k_a}{Y^k_a + U^k_a} \in [0,1]. 
\end{equation}
By dividing the estimated intensity by $V^k_a$, we inflate the contact intensity to account for the group contacts. This is achieved by adding $\log V^k_a$ as an offset to the log-linear predictor for the observed detailed contacts. This ensures that model parameters target the inflated contact intensity:
$$\log \mathbb{E}[Y^{k,.}_{a,b} | N^k_a] = \log \gamma_{a,b} + \log \delta^{k,.}_{a,b} + \log N^k_a + \log P_b + \log V^k_a.
$$
As illustrated in \Cref{fig:bics-group-contacts-distribution}, the distribution of group contacts across survey waves is heavily right-skewed, with a small subset of participants reporting extremely large counts. Without appropriate preprocessing, these outliers severely distort the final contact intensity estimates. A common preprocessing procedure is to truncate the group contacts at a chosen threshold (e.g., 90\%, 95\%, or 99\% quantile) \citep{mossong_social_2008, feehan_quantifying_2021}. However, the choice of threshold is difficult to justify and does not account for patterns in the group contacts. To robustly handle this overdispersion and estimate the typical group contact volume, we model the group contacts using a quantile Generalized Additive Model targeting the median \citep{fasiolo_fast_2021}.

Specifically, for a respondent $i \in \{1,\ldots,N\}$ of age $a_i$, we specify the linear predictor for the median group contacts as:
$$\eta_i = \beta_0 + f(a_i) + \bm{\beta}^\top\bm{x}_i,
$$
where $\beta_0$ is the global intercept, $f(a_i)$ is a non-linear effect over the respondent's age modeled via a thin plate regression spline, and $\bm{x}_i$ is a vector of categorical covariates including sex, education, race/ethnicity, and urban-rural residency.

We implement this using the \texttt{qgam} package in R \citep{fasiolo_qgam_2021}. Because smooth quantile GAMs replace the non-differentiable pinball loss with a smooth approximation, the model requires a learning rate to appropriately balance the loss function against the smoothing priors. We first calibrate this optimal learning rate dynamically using a data-driven routine. We then pass this calibrated parameter into the final median quantile regression model to estimate the regression coefficients.

Finally, we replace the raw observed unstructured contact counts for each individual, $U_i$, with the median predictions, $\hat{U}_i$, generated by the fitted quantile model. These smoothed predictions are then aggregated by stratum and age to compute $\hat{U}^k_a$. We substitute the raw $U^k_a$ with the predicted $\hat{U}^k_a$ in the calculation of the offset terms $V^k_a$ \eqref{app:group-contact-offset}. In \Cref{fig:bics-group-contacts-modelled}, we compare the values of the group contact offsets calculated using the raw observed values and predicted values from the quantile GAM model.

\subsection{Additional Results} Here, we provide additional figures that supplement the main analysis. In \Cref{fig:bics-marginal-contact-intensity-diff} and \Cref{fig:bics-marginal-contact-intensity} we show the posterior deviations in marginal contact intensity from the population weighted average and marginal contact intensity stratified by biological sex and race/ethnicity for BICS waves 1 to 3, respectively. In \Cref{fig:bics-partial-contact-intensity} we show the posterior mean of the partially stratified contact matrices inferred by our model.

\begin{figure}[!h]
    \centering
    \includegraphics[width=\linewidth]{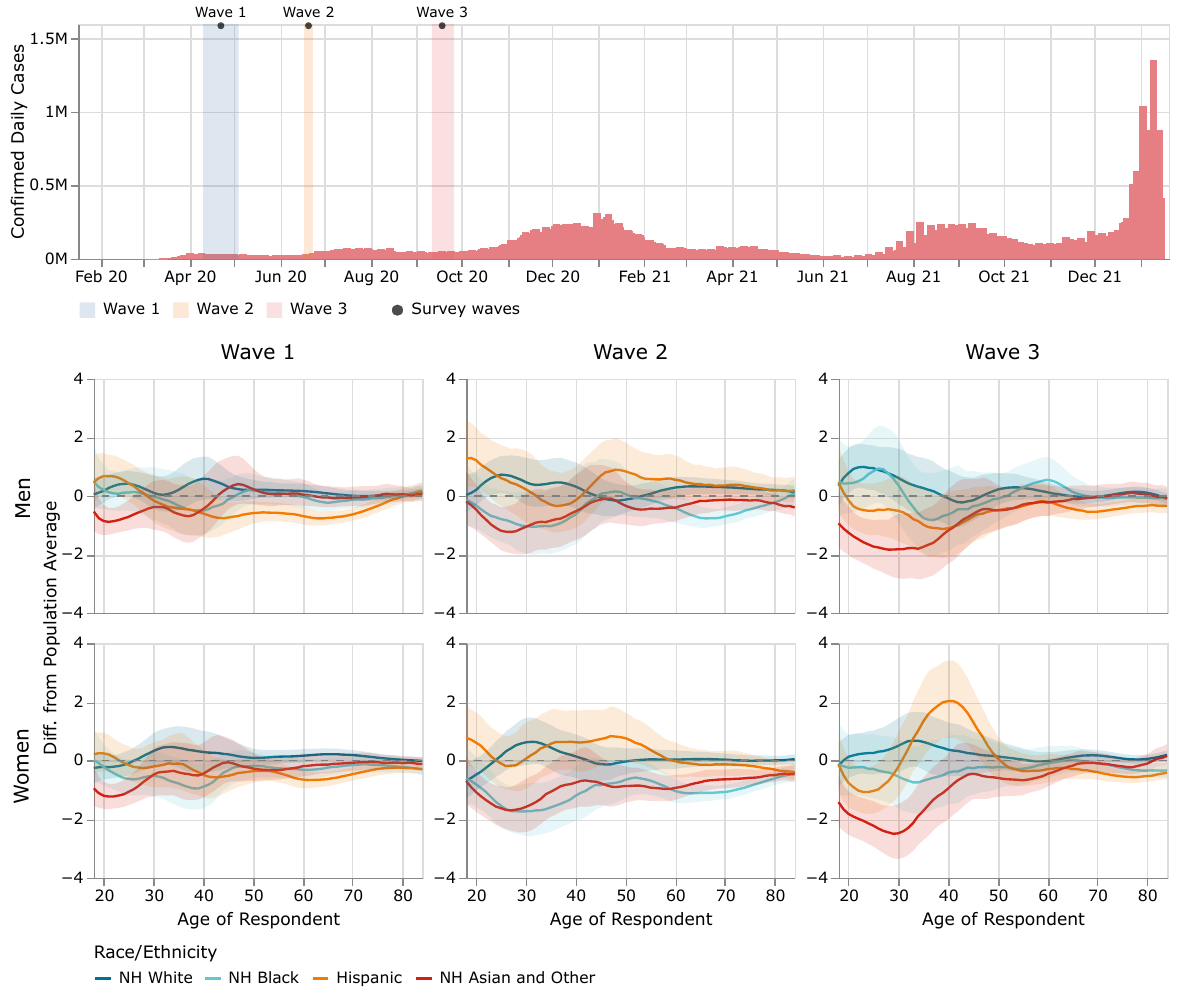}
    \caption{Top row: Confirmed daily cases of COVID-19 in the United States obtained from \citet{dong_interactive_2020}, overlaid with the data collection periods for BICS waves 1 to 3. Rows 2 and 3: Deviations in marginal contact intensity from the population-weighted average, stratified by biological sex and race/ethnicity across BICS waves 1 to 3. Solid lines represent the posterior means, and shaded bands represent the 95\% CI.}
    \label{fig:bics-marginal-contact-intensity-diff}
\end{figure}

\begin{figure}[!h]
    \centering
    \includegraphics[width=\linewidth]{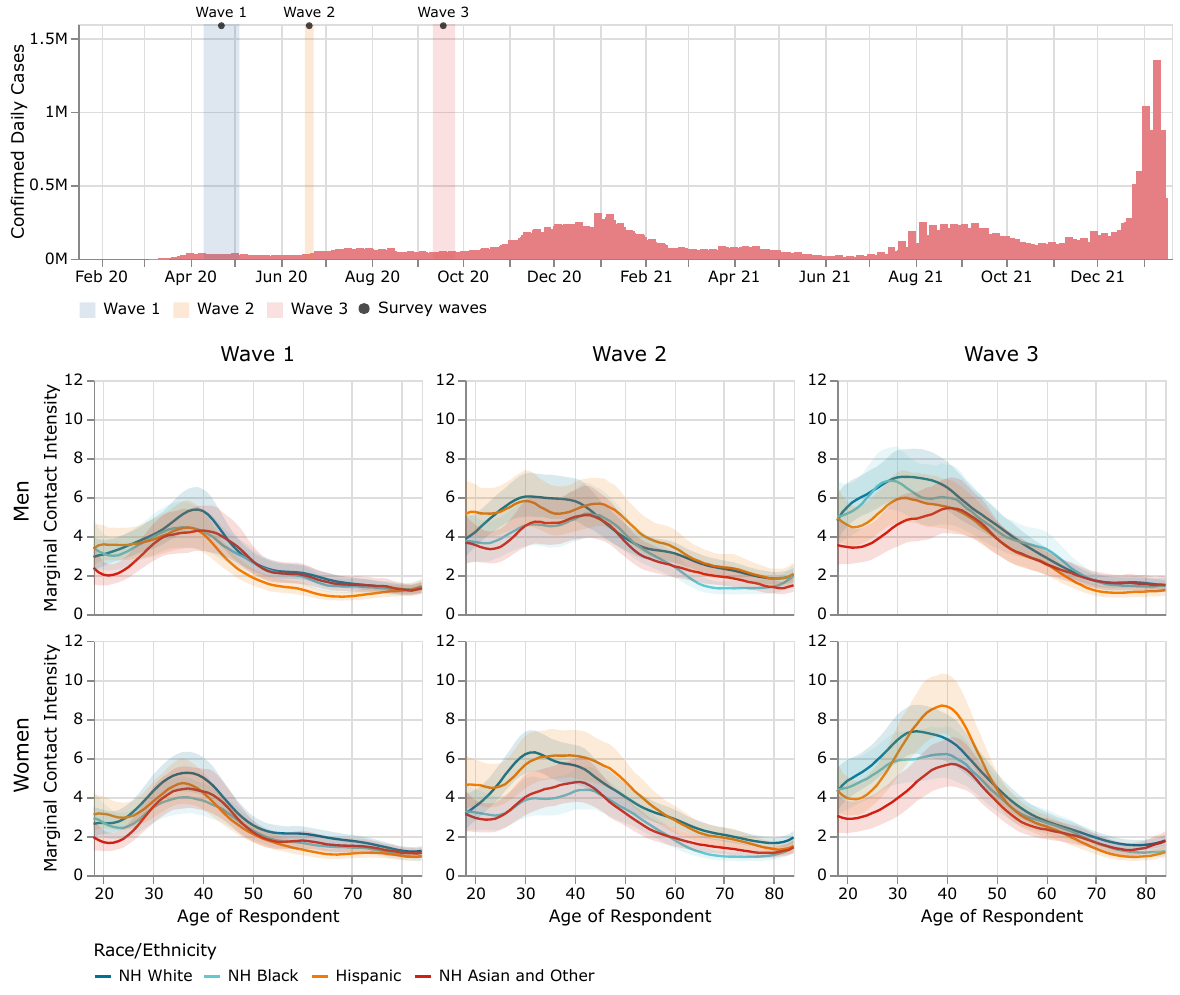}
    \caption{Top row: Confirmed daily cases of COVID-19 in the United States obtained from \citet{dong_interactive_2020}, overlaid with the data collection periods for BICS waves 1 to 3. Rows 2 and 3: Marginal contact intensity stratified by biological sex and race/ethnicity across BICS waves 1 to 3. Solid lines represent the posterior means, and shaded bands represent the 95\% CI.}
    \label{fig:bics-marginal-contact-intensity}
\end{figure}

\begin{figure}[!h]
    \centering
    \includegraphics[width=\linewidth]{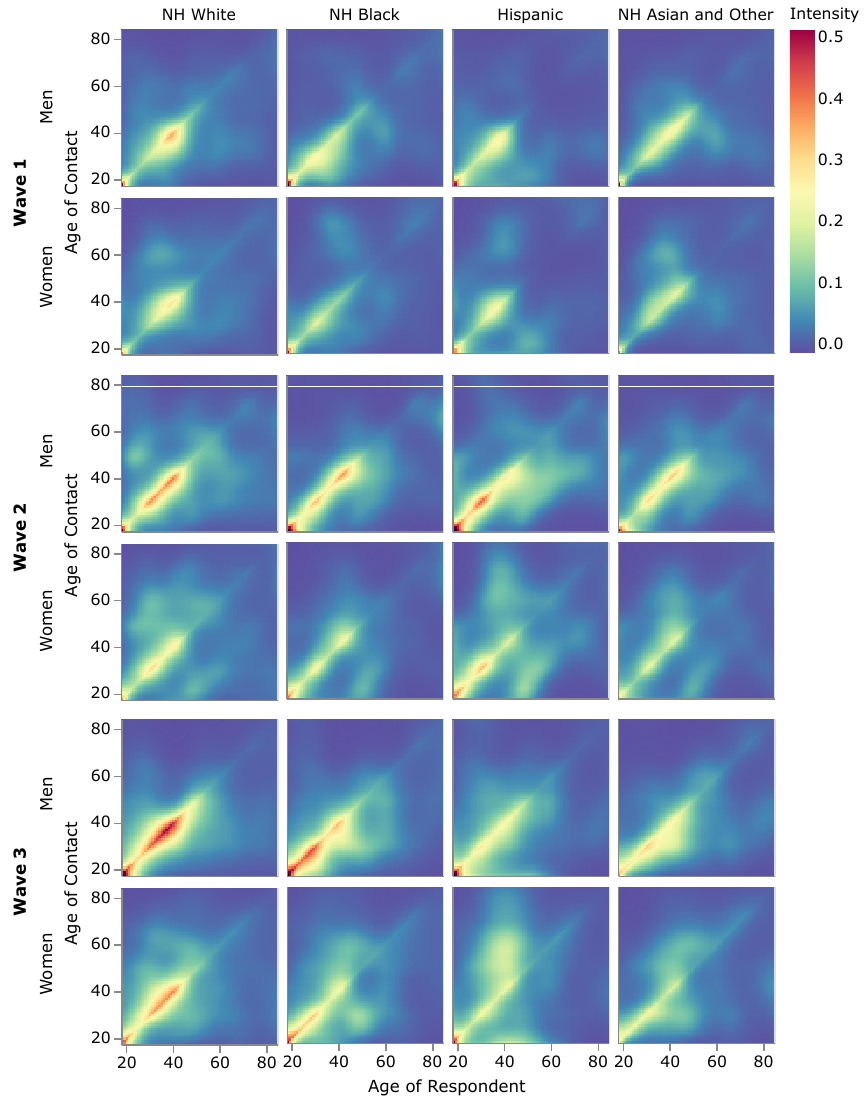}
    \caption{Posterior mean of partially stratified contact intensity matrices stratified by sex and race/ethnicity across BICS waves 1 to 3. Warmer colors indicate higher intensity with values capped at 0.5 for visualization purposes.}
    \label{fig:bics-partial-contact-intensity}
\end{figure}

\begin{figure}[!h]
    \centering
    \includegraphics[width=\linewidth]{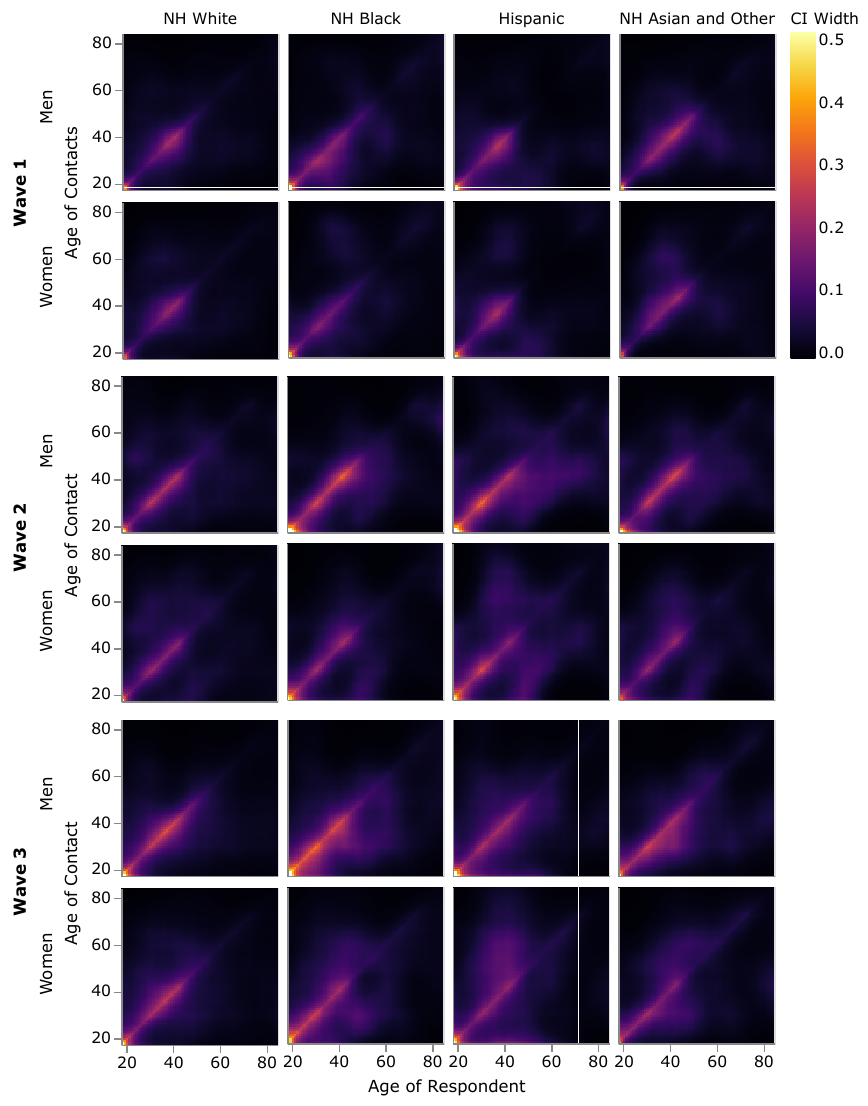}
    \caption{Width of posterior 95\% credible intervals of partially stratified contact intensity matrices stratified by sex and race/ethnicity across BICS waves 1 to 3. Warmer colors indicate higher  with values capped at 0.5 for visualization purposes.}
    \label{fig:bics-partial-contact-intensity-ci}
\end{figure}
\section{Additional Analysis Details for COVIMOD}

\subsection{Data Preprocessing for Participant and Contact Data}
\subsubsection*{Participant Data}
To maintain consistency in age resolution within our modelling framework, participants who reported their age only as an age range were assigned 1-year age drawn uniformly from within their reported range. The study sample was restricted to individuals aged 15 to 84 (above the legal working age). We cap the number of group contacts at the 95\% quantile and apply Gaussian smoothing to mitigate the effect of outliers and sampling noise in our final estimates.

\subsubsection*{Contact Data}
The theoretical reciprocity constraints underpinning our statistical
framework mathematically require the contact intensity matrices to be square. Consequently,
to conform with the respondent age domain, the contact records were similarly restricted to
contacted individuals aged 15 to 84 years.

\subsection{Population Data}

\subsubsection*{Data Sources}
Population size estimates by 1-year age resolution for the SES categories defined in COVIMOD (Managerial, Clerical, Retired/Unemployed, Self-employed, and Manual) are not directly available. The estimates used in our analysis are constructed using a synthetic triangulation method, relying on four specific datasets obtained from the Eurostat Data Browser (\url{https://ec.europa.eu/eurostat/databrowser}). All data is filtered for the reference year 2019.

The baseline population denominator (the absolute population counts stratified by 1-year age) is derived from Eurostat's demographic database \citep{eurostat_demo_pjan}. To stratify this baseline into the five SES categories, we utilise three subsets of the European Union Labour Force Survey. Employment, unemployment, and inactivity totals are extracted to establish the non-working baseline \citep{eurostat_lfsa_pganws}. The active workforce is then partitioned into empolyees and self-employed individuals \citep{eurostat_lfsa_egaps}. Finally, we apply ISCO-08 occupational classifications to the employee subset, allowing us to derive the managerial, clerical and manual labour population estimates \citep{eurostat_lfsa_egais}.

\subsubsection*{Derivation of Stratified Estimates}
Population sizes in the European Union Labour force surveys (tables with prefix \texttt{lfsa}) are reported in broad age groups (e.g., 5-year bands, denoted here by $g$), we first calculate the ratios of each category for a given age group. We then apply these ratios to the 1-year population counts $P_{a}$ from \texttt{demo\_pjan}, where $a \in g$.

To determine the baseline non-working distribution, we utilise the labour status dataset (\texttt{lfsa\_pganws}). We extract the total reference population count ($P_{\text{total}, g}$), the unemployed count ($P_{\text{une}, g}$), and the inactive count ($P_{\text{inac}, g}$) for age group $g$. The ratio for the retired/unemployed stratum is calculated as follows:
$$
R_{\text{non-working}, g} = \frac{P_{\text{une}, g} + P_{\text{inac}, g}}{P_{\text{total}, g}}
$$

Next, we calculate the self-employed ratio. Because the professional status dataset (\texttt{lfsa\_egaps}) lacks a total population denominator, we must first compute the ratio of the employed workforce ($P_{\text{emp},g}$) relative to the total population ($P_{\text{total},g}$) using \texttt{lfsa\_pganws}. We then multiply this by the self-employment share derived from \texttt{lfsa\_egaps}:
\begin{align*}
R_{\text{workforce}, g} &= \frac{P_{\text{emp}, g}}{P_{\text{total}, g}} \\
R_{\text{self-emp}, g} &= R_{\text{workforce}, g} \times \frac{P_{\text{self}, g}}{P_{\text{emp}, g}}
\end{align*}

Similarly, we calculate the salaried employee (i.e., non-self-employed) share ($R_{\text{sal}, g}$) within the general population using the corresponding salaried employee count ($P_{\text{sal}, g}$):
\begin{equation}
    R_{\text{sal}, g} = R_{\text{workforce}, g} \times \frac{P_{\text{sal}, g}}{P_{\text{emp}, g}}
\end{equation}

We then distribute the salaried employees into distinct occupational strata. Using the \texttt{lfsa\_egais} dataset filtered strictly for salaried employees ($SAL$), we map the International Standard Classification of Occupations (ISCO-08) groups to our SES classifications. The ISCO-08 framework categorises occupations into nine major groups: Managers (Group 1), Professionals (Group 2), Technicians and associate professionals (Group 3), Clerical support workers (Group 4), Service and sales workers (Group 5), Skilled agricultural, forestry and fishery workers (Group 6), Craft and related trades workers (Group 7), Plant and machine operators and assemblers (Group 8), and Elementary occupations (Group 9).

To find the occupational sub-shares, we sum the relevant ISCO group counts ($P_{\text{ISCO}i, g}$) and divide by the total reported ISCO employees ($P_{\text{ISCO\_total}, g}$), before multiplying by the overall salaried employee ratio ($R_{\text{sal}, g}$). The managerial stratum comprises Groups 1 and 2:
$$
R_{\text{managerial}, g} = R_{\text{sal}, g} \times \frac{P_{\text{ISCO}1, g} + P_{\text{ISCO}2, g}}{P_{\text{ISCO\_total}, g}}
$$
The clerical comprises Groups 3, 4, and 5:
$$
R_{\text{clerical}, g} = R_{\text{sal}, g} \times \frac{P_{\text{ISCO}3, g} + P_{\text{ISCO}4, g} + P_{\text{ISCO}5, g}}{P_{\text{ISCO\_total}, g}}
$$
The manual labor stratum aggregates Groups 7, 8, and 9:
$$
R_{\text{manual}, g} = R_{\text{sal}, g} \times \frac{P_{\text{ISCO}7, g} + P_{\text{ISCO}8, g} + P_{\text{ISCO}9, g}}{P_{\text{ISCO\_total}, g}}
$$
Finally, we apply these ratios to the precise 1-year age population data. For each individual age $a$ belonging to group $g$, the final population count $P$ for a given socioeconomic stratum $s \in \{\text{managerial, clerical, non-working, self-emp, manual}\}$ is computed as:
$$
P^s_a = P_a \times R_{s, g}
$$
The sum of all stratum ratios ($R_{s, g}$) equals 1.0, ensuring that the sum of our estimated SES strata strictly recovers the official population counts, such that $\sum_{s} P^s_a = P_{a}$.

\subsection{Preprocessing of Group Contacts}
The group contacts in COVIMOD were not as extreme or noisy as those in BICS, hence the did not require additional modelling. However, to prevent outliers and influential points from overly influencing the estimates we apply smoothing with a Gaussian kernel to the group contact offsets. An appropriate value for the variance parameter of the Gaussian kernel was selected using leave-one-out cross-validation.

\subsection{Additional Modeling Considerations}
In COVIMOD, the age information of the contacts are reported in coarse age ranges:
$$
\mcB = \qty{\text{15-19},\ \text{20-24},\ \text{25-34},\ \text{35-44},\ \text{45-54},\ \text{55-64},\ \text{65-69},\ \text{70-74},\ \text{75-79},\ \text{80-84}}.
$$
We handle this using the framework in \citep{dan_estimating_2023} where for age range $g \in \mcB$ the likelihood function is altered to
\begin{equation*}
Y^{k,.}_{a,g} \sim \text{NegBin}\qty(\sum_{b \in g}\EE[Y^{k,.}_{a,b}], \varphi).
\end{equation*}

It is known that in longitudinal social contact surveys such as COVIMOD where a subset of participants are enrolled multiple times, contact reports are biased by response fatigue \citep{wong_social_2023, dan_addressing_2025}. To address this, we follow the modeling approach in \citep{dan_addressing_2025}. Let $r \in \mathbb{N}$ denote the number of times a participant has been enrolled to answer the survey. We include an additional response fatigue adjustment term in our linear predictor such that: 
\begin{align*}
Y^{k,.}_{a,g,r} &\sim \text{NegBin}\qty(\sum_{b \in g}\EE[Y^{k,.}_{a,b,r}], \varphi) \\
\log \EE[Y^{k,.}_{a,b,r}] &= \log \gamma_{a,b} + \log \delta^{k,.}_{a,b} + \rho(r) + \log N^k_{a,r} + \log P_b + \log V^k_{a,r}.
\end{align*}
Here, $\rho(r)$ is the Hill function used to model dose-response relationships and has the following functional form:
$$
\rho(r) = -\gamma \frac{e^\zeta r^\eta }{1 + e^\zeta r^\eta}, \quad \gamma, \eta \in \RR_+, \zeta \in \RR.
$$
We assign the following priors to the parameters:
$$
\gamma \sim \text{Half-Normal}^+(0,1), \quad \eta \sim \text{Exponential}(1), \quad \zeta \sim \text{Normal}(0,1).
$$
This term down weights the value of the log expectation to match the reduction in the data. As a result, the contact intensity calculated as $m^{k,.}_{a,b} = \gamma_{a,b}\delta^{k,.}_{a,b}P_b$ is inflated. We fit this fatigue and group contact adjusted model independently to each of the five phases of COVIMOD.

\subsection{Additional Results}
We provide additional figures that supplement the main analysis. In \Cref{fig:covimod-marginal-contact-intensity} we show the posterior marginal contact intensity stratified by socioeconomic status for the 5 COVIMOD phases. In \Cref{fig:covimod-partial-contact-intensity} we show the posterior mean of the partially stratified contact matrices inferred by our model for each phase.

\begin{figure}[!h]
    \centering
    \includegraphics[width=\linewidth]{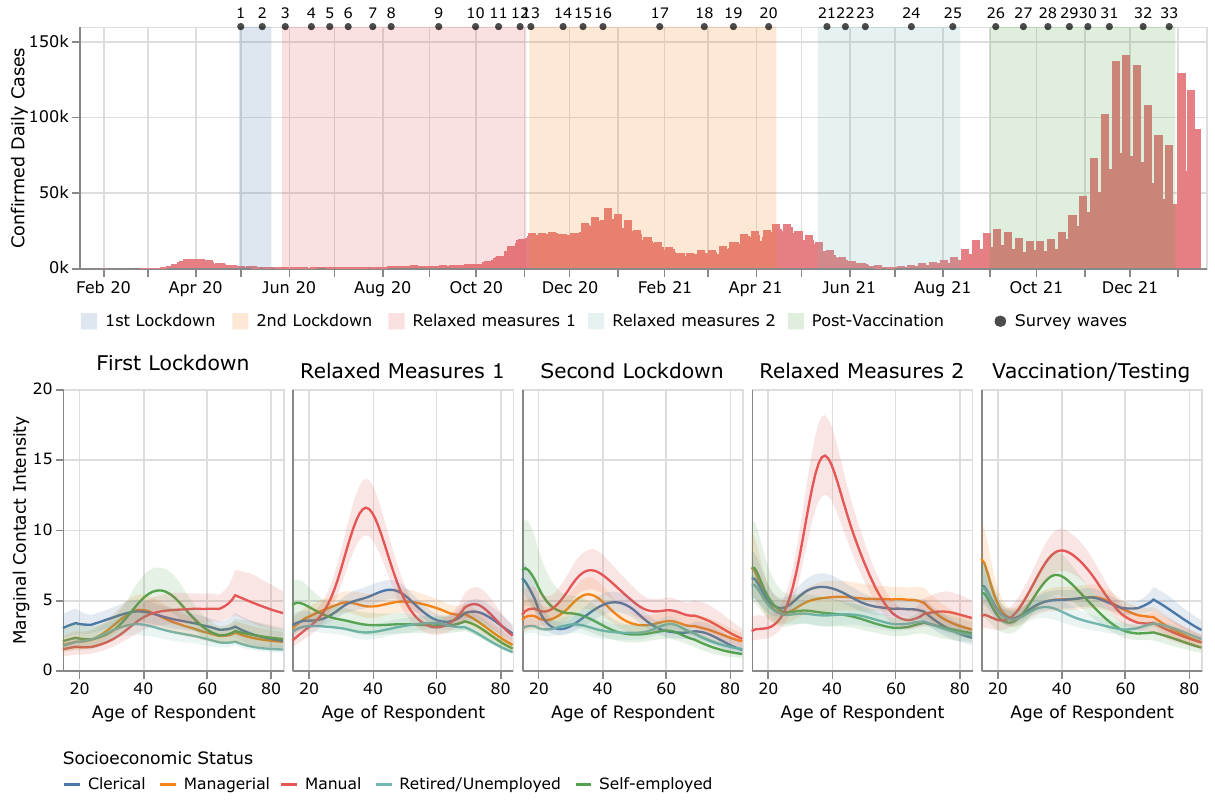}
    \caption{Top row: Confirmed daily cases of COVID-19 in the United States obtained from \citet{dong_interactive_2020}, overlaid with the data collection periods for COVIMOD waves 1 to 33. Rows 2 and 3: Marginal contact intensity stratified by socioeconomic status (SES) for COVIMOD waves 1 to 33. Solid lines represent the posterior means, and shaded bands represent the 95\% CI.}
    \label{fig:covimod-marginal-contact-intensity}
\end{figure}

\begin{figure}[!h]
    \centering
    \includegraphics[width=\linewidth]{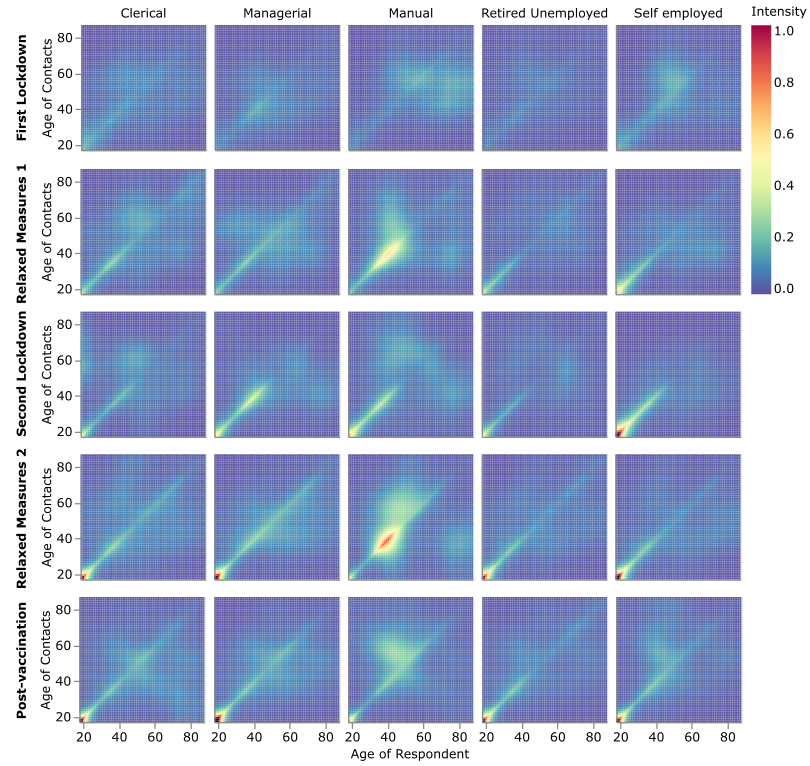}
    \caption{Posterior mean of partially stratified contact intensity matrices by socioeconomic status across COVIMOD phases. Warmer colors indicate higher intensity with values capped at 1.0 for visualization purposes.}
    \label{fig:covimod-partial-contact-intensity}
\end{figure}

\begin{figure}[!h]
    \centering
    \includegraphics[width=\linewidth]{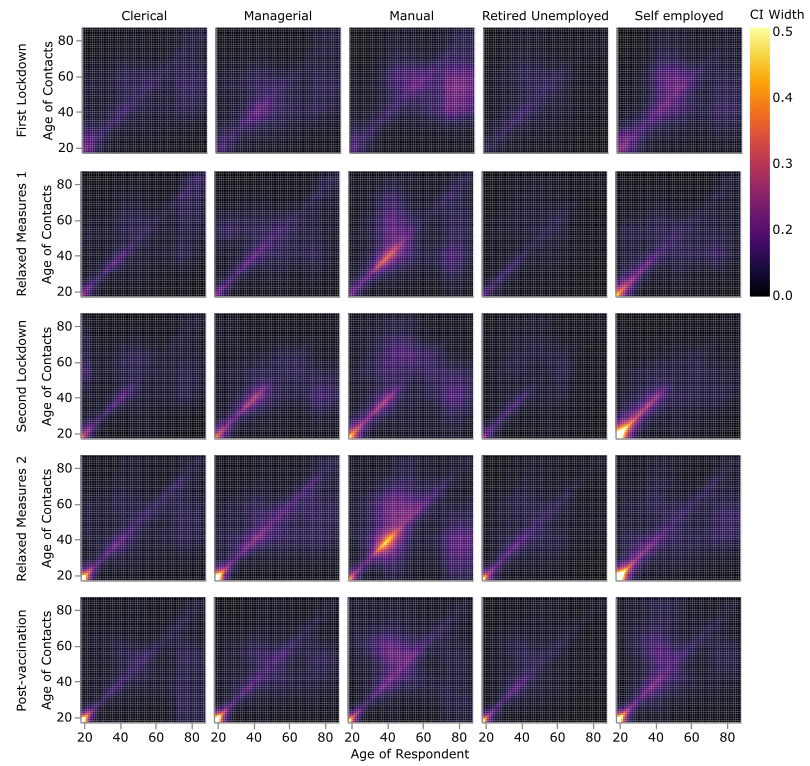}
    \caption{Width of posterior 95\% credible intervals (CI) of partially stratified contact intensity matrices by socioeconomic status across COVIMOD phases. Warmer colors indicate higher uncertainty. Values capped at 0.5 for visualization purposes.}
    \label{fig:covimod-partial-contact-intensity-ci}
\end{figure}

\putbib[definitions,references]
\end{bibunit}

\end{document}